\documentclass{article}
\usepackage[T1]{fontenc}
\usepackage[utf8]{inputenc}
\usepackage{amsmath}											
\usepackage{amsfonts}
\usepackage{amssymb}
\usepackage{amsthm}
\usepackage{bbm}
\usepackage[normalem]{ulem}
\usepackage{caption}
\usepackage{subcaption}
\usepackage{comment}
\usepackage{esint}
\usepackage[a4paper, total={6.5in, 8in}]{geometry}
\usepackage{graphicx}
\usepackage{mathtools}
\usepackage[final]{showlabels}
\usepackage{xcolor}
\usepackage[section]{placeins}

\DeclareMathOperator{\Op}{Op}
\DeclareMathOperator{\Tr}{Tr}
\DeclareMathOperator{\tr}{tr}
\DeclareMathOperator{\sgn}{\text{sign}}
\DeclareMathOperator{\supp}{supp}

\newcommand{\norm}[1]{\left\lVert#1\right\rVert}
\newtheorem{proposition}{Proposition}

\newtheorem{theorem}[proposition]{Theorem}
\newtheorem{remark}{Remark}
\newtheorem{corollary}[proposition]{Corollary}
\numberwithin{proposition}{section}

\usepackage{tikz}
\usetikzlibrary{backgrounds}

\iftrue

\newcommand{\tb}[1]{{{{#1}}}}

\newcommand{\tm}[1]{{{{#1}}}}

\newcommand{\sq}[1]{{{{#1}}}}

\newcommand{\pc}[1]{{#1}}
\fi

\newcommand{\eps}{\varepsilon}
\newcommand{\Rm}{\mathbb{R}}
\newcommand{\Cm}{\mathbb{C}}
\newcommand{\Zm}{\mathbb{Z}}
\newcommand{\Nm}{\mathbb{N}}

\newcommand{\dint}{\displaystyle\int}

\newcommand{\aver}[1]{\langle #1 \rangle}
\newcommand{\sign}{\text{sign}}

\newcommand{\fs}{\mathfrak{S}}

\newcommand{\kone}{k_1}
\newcommand{\ktwo}{k_2}
\newcommand{\tx}{\tilde{x}}
\newcommand{\ty}{\tilde{y}}
\newcommand{\fm}{u}

\newcommand{\Hmu}{H^{(\mu)}}

\newcommand{\projP}{\bar{P}}
\newcommand{\genH}{\tilde{H}}
\newcommand{\pert}{W}
\newcommand{\Uminusone}{V}
\newcommand{\Umu}{U^{(\mu)}}

\newcommand{\anot}{\alpha_0}
\newcommand{\sigmamu}{\sigma^{(\mu)}}

\newcommand{\sbz}{\tau_{\pm,z}}
\newcommand{\invariantplus}{I_+}
\newcommand{\invariantminus}{I_-}
\newcommand{\invariantpm}{I_\pm}
\newcommand{\modfs}{\bar{\fs}}

\numberwithin{equation}{section}

\title{Mathematical models of topologically protected transport in twisted bilayer graphene}
\author{Guillaume Bal  \thanks{Departments of Statistics and Mathematics and CCAM, University of Chicago, Chicago, IL 60637; {\tt guillaumebal@uchicago.edu}} \and Paul Cazeaux \thanks{ Department of Mathematics, Virginia Tech, Blacksburg, VA 24060; {\tt cazeaux@vt.edu}} \and Daniel Massatt \thanks{ Department of Mathematics, Louisiana State University, Baton Rouge, LA 70803;  {\tt dmassatt@lsu.edu}} \and Solomon Quinn \thanks{Department of Statistics and CCAM, University of Chicago, Chicago, IL 60637; {\tt solomonquinn@uchicago.edu}}}

\begin{document}

\maketitle

\begin{abstract}
    Twisted bilayer graphene gives rise to large moir\'{e} patterns that form a triangular network upon mechanical relaxation. If gating is included, each triangular region has gapped electronic Dirac points that behave as bulk topological insulators with topological indices depending on valley index and the type of stacking. Since each triangle has two oppositely charged valleys, they remain topologically trivial.
    
    In this work, we address several questions related to the edge currents of this system by analysis and computation of continuum PDE models. 
    Firstly, we derive the bulk invariants corresponding to a single valley, and then apply a bulk-interface correspondence to quantify asymmetric  transport along the interface.
    Secondly, we introduce a valley-coupled continuum model to show how valleys are approximately decoupled in the presence of small perturbations using a multiscale expansion, and how valleys couple for larger defects. 
    Thirdly, we present a method to prove for a large class of continuum (pseudo-)differential models that a quantized asymmetric current is preserved through a junction such as a triangular network vertex.

    We support all of these arguments with numerical simulations using spectral methods to compute relevant currents and wavepacket propagation.
\end{abstract}

\section{Introduction}

Twisted bilayer graphene (tBLG) is widely studied for its unique mechanical and electronic properties including the magic angle superconductivity \cite{pablo2018,macdonald_2011}. Upon gating, it acts as host of a network of topological interface channels, which can be seen experimentally and theoretically \cite{san_jose_2013,klaus2018,carr_2018,Peeters_2018}.  tBLG is constructed by taking two periodic 2D sheets of graphene and stacking them with a relative twist, typically small. The atoms relax to minimize energy, forming large triangular regions of the energetically favorable AB and BA Bernal stacking \cite{kaxiras_2021,carr_2018, srolovitz2016}. It is relevant to note this asymmetric transport under gating is a separate phenomena from magic angle superconductivity. Indeed, the asymmetric transport phenomena only requires sufficiently small twist angles and vertical gating, i.e. by inducing a potential difference between the two layers, while superconductivity can only occur precisely
at the magic twist angles. In this section, we make use of Figure \ref{fig:tblg} for discussion of the geometry. 

The interior of one of these triangles can be considered as approximately an infinite periodic material \tm{as the side of one of these triangles scales inversely proportional to the small twist angle \cite{carr_2018}}. We note that in this work we will consider only continuum models and not lattice models, but for understanding of the purpose and origins of the continuum models in this work, it is useful to see the lattice model origins. If $\mathcal{L}$ is the lattice matrix corresponding to the periodicity of the material, we define the reciprocal lattice unit cell as $\Gamma^* := 2\pi \mathcal{L}^{-T}[0,1)^2$. Consider periodic Hamiltonian $H$ describing the infinitely extended bulk corresponding to one of the triangular regions. Its spectra, which in turn describes electronic properties, can be described by the Bloch states satisfying
\begin{equation}
    H(\xi) \psi^{(j)}(\xi) = E^{(j)}(\xi)\psi^{(j)}(\xi), \qquad \xi \in \Gamma^*
\end{equation}
where $H(\xi) = e^{-i\xi \cdot r} H e^{i\xi\cdot r}$. This expression is understood as operator composition. We assume $E^{(1)}(\xi) \leq E^{(2)}(\xi) \leq \cdots$ are the ordered eigenvalues (or energies). Gating opens a band gap at $E=0$, meaning all $E^{(j)}(\xi)$ are bounded away from $0$. A topological integer index can be associated with the operator $H$ by the formula
\begin{equation}\label{eq:IH}
    I[H] := \sum_{j\geq  n}\frac{i}{2\pi } \int_{\Gamma^*} d(\psi^{(j)},d\psi^{(j)}).
\end{equation}
Here $n$ is the smallest index such that $E^{n}(\xi) > 0$ for all $\xi \in \Gamma^*$. \tm{Many of these more complex Hamiltonians can be well approximated by simpler continuum models that capture the relevant dispersion relations \cite{Hughes_2013, fruchart_2013, witten_2016}. In the case of a bilayer graphene,} $I[H]$ can be computed by expanding only around the Dirac points $K,K' \in \Gamma^*$ as they possess the topological information \cite{Hughes_2013, Drouot:19b, fruchart_2013}. 
\pc{
Such effective moiré-scale continuum Hamiltonians for low-energy physics of bilayer graphene, in particular the Bistritzer-MacDonald model~\cite{macdonald_2011}, have been recently derived in the small twist angle regime from atomic-scale models: in~\cite{cances2022simple}, using a formal variational approximation of the tBLG Kohn-Sham Hamiltonian, and in~\cite{watson2023bistritzer}, by identifying an appropriate asymptotic parameter regime of the discrete tight-binding model.}
A local in momentum Hamiltonian $H(\xi)$ can be constructed for each of these valleys through a  $4\times 4$ model as follows:
\begin{equation}
\label{e:gated_bulk}
    H(\xi) := \begin{pmatrix} \Omega I + \xi_1 \sigma_1 + \eta \xi_2 \sigma_2 & \lambda U^* \\ \lambda U & -\Omega I + \xi_1 \sigma_1 + \eta \xi_2 \sigma_2 \end{pmatrix}.
\end{equation}
Here $\eta = 1$ for valley $K$ and $\eta = -1$ for valley $K'$, $\tm{2}\Omega\in \mathbb{R}$ measures the electrostatic potential difference between the two layers due to vertical gating, and $U \in \{A,A^*\}$ models an interlayer coupling term with strength $0\not=\lambda\in\mathbb{R}$, where $A = \begin{pmatrix} 0& 1 \\ 0 &0 \end{pmatrix}$ \cite{McCann_2013}: $A^*$ and $A$ correspond to whether the triangle has $AB$ or $BA$ stacking geometry respectively. This Hamiltonian has  been rescaled to ensure Fermi velocity is $1$ and unitless, i.e. there is no prefactor in front of the Dirac term $\xi_1\sigma_1 + \xi_2\sigma_2$. A standard result in topological insulators is the bulk-interface correspondence \cite{Hughes_2013,fruchart_2013}, which states that if two topological materials are glued together at an interface, then there is quantized edge transport given by the difference of the topological index on either side of the interface; \tm{see} \cite{bal3,B-higher-dimensional-2021,Drouot:19b,QB} for the derivation of the correspondence for differential models. In bilayer graphene the $K$ and $K'$ valleys contribute opposite signed topological charge, so gated bilayer graphene is technically trivial. However, it is understood that the two valleys in many settings decouple, with each valley Hamiltonian having possibly nontrivial topology, which is the basis of {\em valleytronics} \cite{Eugene_2013, ValleyNature_2016}. 

\medskip

In this work, we address gated tBLG via numerics and analysis of continuum models built to address three different questions corresponding to the two regions highlighted in Figure \ref{fig:tblg}. For the first question addressed in Section \ref{sec:tblg}, we use the $4\times 4$ single-valley bulk model, compute its topological invariant, and then using the corresponding interface model apply the bulk-interface correspondence  at an edge between two triangles (region $1$). In particular, we make use of the bulk-difference invariant \cite{bal3} between AB and BA stacked graphene for a single valley to find a quantized current of $\pm 2$. In Section \ref{sec:valleys}, we focus on analyzing  the separation of the $K$ and $K'$ valley at the interface (region $1$). Since the $4\times 4$ bilayer graphene model is only accurate near the valleys, we study valley coupling with a toy $2\times 2$ model that exhibits two conical Dirac points. To model defects, we consider a two-scale model to introduce fast small fluctuations, and show the valley coupling is weak in this regime. We use numerics and scattering theory to further show that as the defects are tuned to couple the valleys, then valley separation breaks and quantized edge current is lost. 
In Section \ref{sec:junctions}, we define a conductivity that quantifies the asymmetric transport through a junction (e.g. region $2$). Following pseudodifferential calculus arguments from \cite{QB}, we derive an analytic formula for this junction conductivity and show that the latter is immune to perturbations. This allows for explicit evaluations of the conductivity for both the $4\times 4$ bilayer graphene and $2 \times 2$ Dirac models.
Finally in Section \ref{sec:numerics}, we include numerical verification of results from the preceding sections using spectral methods to properly capture the Dirac operator. Direct computation of traces is used to compute currents either isolated to the edge (Region 1) and currents through a junction (Region 2), and wavepacket propagation is used to illustrate the connection between the two regions by evolving a state localized on an edge through a network junction.

\begin{figure}[ht]
\centering
\includegraphics[width=.5\textwidth]{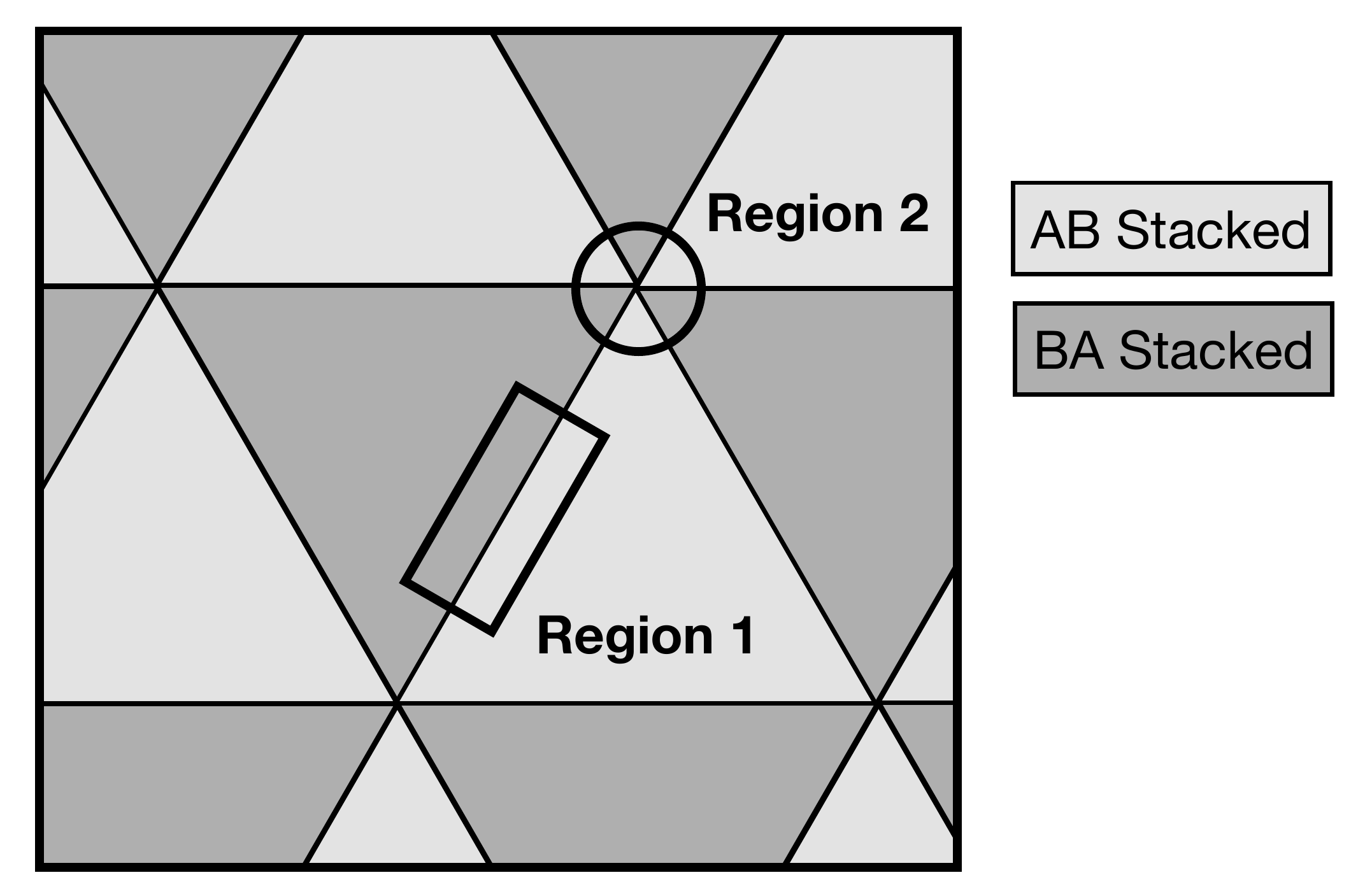}
\caption{Triangular domains form in mechanically relaxed tBLG. We highlight two regions of interest for our analysis. Region $1$ corresponds to an edge, where a standard bulk-interface correspondence will be studied as well as a continuum valley coupling model. Region 2 is at a junction, where we will analyze current conservation.}
\label{fig:tblg}
\end{figure}

\section{Gated tBLG model}
\label{sec:tblg}

In this section, we analyze region 1 from Figure \ref{fig:tblg} and use the bulk-interface correspondence to show quantized asymmetric transport for a single valley. To begin, we consider the bulk Hamiltonian in \eqref{e:gated_bulk} described by the continuum model
\[
H = \begin{pmatrix}  \Omega I + D \cdot \sigma^{(\eta)} & \lambda U^* \\ \lambda U & -\Omega I + D\cdot\sigma^{(\eta)}\end{pmatrix}.
\]
Here $\sigma^{(\eta)} := (\sigma_1,\eta\sigma_2)$, and $\xi \cdot \sigma^{(\eta)} := \xi_1\sigma_1 + \eta \xi_2\sigma_2$ where we recall $\eta \in \pm 1$ is the valley index. We will sometimes denote  $\sigma^{(1)}$ by $\sigma$. This operator in two-dimensional physical variables is an unbounded operator on $L^2(\Rm^2;\Cm^4)$. 
The operator in the Fourier variables is given as in \eqref{e:gated_bulk}  for each $\xi\in\Rm^2$ by the $4\times4$ matrix:
\[
H(\xi) = \begin{pmatrix} \Omega I+ \xi \cdot \sigma^{(\eta)} & \lambda U^* \\ \lambda U & -\Omega I+ \xi \cdot \sigma^{(\eta)}\end{pmatrix}.
\]
These models are defined on a plane $\mathbb{R}^2$, and hence are not defined over a compact manifold. To construct topological invariants, we introduce as in \cite{bal3} the bulk-difference invariant describing  the difference between AB and BA stacked graphene, the two bulks enclosing an interface. The two bulk Hamiltonians are given by
\begin{align*}
    & H_+(\xi) = \begin{pmatrix} \Omega I + \xi\cdot\sigma^{(\eta)} & \lambda A^* \\ \lambda A & -\Omega I + \xi\cdot\sigma^{(\eta)} \end{pmatrix}, 
    & H_-(\xi) = \begin{pmatrix} \Omega I + \xi\cdot\sigma^{(\eta)} & \lambda A \\ \lambda A^* & -\Omega I + \xi\cdot\sigma^{(\eta)} \end{pmatrix}.
\end{align*}

\tb{These $4\times4$ Hermitian matrices may be diagonalized with eigenvectors denoted by $\psi_\pm^{(j)}(\xi)$ for $1\leq j\leq 4$. The eigenprojectors $(\pm,\xi)\to \Pi^{(j)}_\pm(\xi) := |\psi_\pm^{(j)}(\xi)\rangle\langle\psi_\pm^{(j)}(\xi)|$ are uniquely defined as the eigenvalues are all simple  $E_1(\xi)<E_2(\xi)<0<E_3(\xi)<E_4(\xi)$ and independent of the phase $\pm$ as we will show in the appendix. Since $\xi\in\mathbb{R}^2$ is simply connected, the normalized eigenvectors $\psi_\pm^{(j)}(\xi)$, a priori defined up to a multiplicative phase, may be chosen smoothly in $\xi$; an explicit expression is presented in Appendix \ref{sec:aa}. Associated to each projector $\Pi$ is a curvature defined by ${\rm tr}\, \Pi\, d\Pi\wedge d\Pi$. We introduce the following integrals
\begin{equation}
\label{e:Wpm}
    W_\pm^j := \frac{i}{2\pi}  \int_{\mathbb{R}^2} {\rm tr}\, \Pi^{(j)}_\pm d\Pi^{(j)}_\pm \wedge d\Pi^{(j)}_\pm = \frac{i}{2\pi}  \int_{\mathbb{R}^2} d(\psi_\pm^{(j)},d\psi_\pm^{(j)})
    =\frac{i}{2\pi}  \lim_{R\to\infty} \dint_{{\mathbb S}_R^1}(\psi_\pm^{(j)},d\psi_\pm^{(j)}).
\end{equation}
Here, ${\mathbb S}_R^1$ is the circle of radius $R$ with counterclockwise orientation. The second equality can be obtained by explicit calculations. The third equality is a consequence of the Stokes theorem and the smoothness of the vector field $\psi_\pm^{(j)}(\xi)$. 

We will show that $H_\pm(\xi)$ are gaped at energy $E=0$. As a consequence, we expect $W_\pm:=W_\pm^3+W_\pm^4$ 
corresponding to the integral of the curvature of the rank-two projectors $\Pi_\pm=\chi(H_\pm>0)$ to have topological significance. A main difference with \eqref{eq:IH}, however, is that unlike $\Gamma^*$, $\mathbb{R}^2$ is not compact. When $H$ is a Dirac operator for instance, we observe that $W_\pm\in\{-\frac12,\frac12\}$ is not an integer \cite{bal3}. 
This motivates the definition of the bulk-difference invariant
\begin{equation}\label{eq:W}
    W := W_+ - W_- = W_+^3+W_+^4 - (W_-^3+W_-^4).
\end{equation}
The invariant is modeling a transition between the bulk Hamiltonians $H_+$ and $H_-$. We will show that $W$ is indeed an integral-valued invariant unlike its components $W_\pm^j$. Heuristically, this shows that it is easier to define a phase transition between insulators rather than absolute phases, which are often ill-defined for partial differential models such as Dirac equations.}


We will compare this bulk invariant to an edge current, which describes asymmetric transport along an interface separating the bulk phases and is defined next. Firstly, we introduce an edge model by making the bulk above into a continuum model with an interface:
\begin{equation}\label{eq:He}
    H_e := \begin{pmatrix} \Omega I + D \cdot \sigma^{(\eta)} & \lambda U^*(y) \\ \lambda U(y) & -\Omega I + D \cdot \sigma^{(\eta)}\end{pmatrix}
\end{equation}
with \sq{$U(y) = \frac{1}{2}((1+m(y)) A + (1-m(y)) A^*)$} and $D = (D_x,D_y) = \frac{1}{i}(\partial_x,\partial_y)$. We choose $m$ such that for large $y$ we obtain BA stacking, while for large negative $y$ we obtain AB stacking.  In order to do this, we define the switch function space $\modfs(c_1,c_2;y_1,y_2)$ 
\sq{as follows. Given $f_1,f_2 \in \mathbb{C}$ and $c_1, c_2 \in \mathbb{R}$ with $c_1 \le c_2$, we say that a function $f: \mathbb{R} \rightarrow \mathbb{R}$ belongs to $\modfs(f_1,f_2;c_1,c_2)$ if and only if 
\begin{align}\label{eq:sfdef}
    f(y) = 
    \begin{cases}
    f_1, & y < c_1\\
    f_2, & y > c_2
    \end{cases}.
\end{align}
We then define the space of \emph{smooth switch functions} by $\fs(f_1,f_2;c_1,c_2)=\modfs(f_1,f_2;c_1,c_2) \cap \mathcal{C}^\infty (\mathbb{R})$.}

We assume \sq{$m \in \fs(-1,1;-y_0, y_0)$} for some $y_0 > 0$. 
We define an edge state capturing energy function $\varphi \in \fs(0,1;-E_0,E_0)$ where $E_0 > 0$ is chosen to be smaller than the bulk band gap. We have not yet shown there is a bulk band gap, but prove this below in Theorem \ref{thm:tblg-bulk}. We note that $\varphi'$ is supported in the gap, and thus $\varphi'(H_e)$ captures edge state spectra only. Finally we define $P \in \fs(0, 1; -x_0, x_0)$ for some $x_0 > 0$, which in the limit $x_0\to0$ approximates the Heaviside function. The rate of propagation of current through the edge of $P$ is given by \cite{B-bulk-interface-2019,bal3,elbau2002equality,prodan2016bulk}
\begin{equation}\label{eq:sigmaI}
    \sigma_I := \text{Tr} \,i[H_e,P]\varphi'(H_e).
\end{equation}

The main objective of this section is to show that $2\pi \sigma_I = W = -2\eta\sgn{\Omega}$. The first equality is a bulk-interface correspondence, which has been derived in a variety of contexts \cite{B-bulk-interface-2019,bal3,elbau2002equality,prodan2016bulk,QB}, and in particular in \cite{bal3,QB} for the operators of interest here. The advantage of such a correspondence is that the explicit computation of $W$ is often simpler than that of $\sigma_I$. We will obtain the former by adapting results derived in \cite{massatt2021}.

The Fourier transform of $H_e$ can be taken along the interface, obtaining the Hamiltonian
\begin{equation}
H_e(\xi_1) = \begin{pmatrix} \Omega I + \xi_1\sigma_1 + \eta D_{y} \sigma_2 & \lambda U^*(y) \\ \lambda U(y) & -\Omega I + \xi_1 \sigma_1 + \eta D_{y} \sigma_2\end{pmatrix}.
\end{equation}
In Kronecker product notation, this yields
\begin{equation}
    H_e(\xi_1) = \Omega \sigma_3\otimes I_2 + I_2 \otimes (\xi_1\sigma_1 + \eta D_y \sigma_2) + \frac{1}{2} (\sigma_1 \otimes \sigma_1 + m(y) \sigma_2\otimes\sigma_2).
\end{equation}
In Figure \ref{fig:edge_state}, we plot the bulk band structure and superimpose the eigenvalues within the bulk band gap of a periodized version of $H_e(\xi_1)$, where we only include eigenvalues weighted towards one interface. This is necessary as the periodization creates a second domain wall with oppositely directed edge states. We then see  that there are two right-ward propagating states, as expected (two curves crossing the bulk band gap with positive group velocity $\partial E/\partial\xi$). We note that the symmetry in the edge states $E(\xi_1) = -E(-\xi_1)$ comes from the following symmetry operation. Let $K = \sigma_2 \otimes \sigma_1$. Then 
$$K^{-1}H_e(\xi_1)K = -H_e(-\xi_1).$$

To prove that $W$ is an invariant, the first step will be verifying there is indeed a bulk gap, which turns out to have a ring shape. We then use this to calculate the bulk-difference invariant $W$. 
See Figure \ref{fig:tblg_bs} to see concentration of curvature on the ring and the corresponding band structure. 
Combining the above yields our main result:

\begin{figure}[ht]
\begin{subfigure}{.5\textwidth}
\centering
\includegraphics[width=1\textwidth]{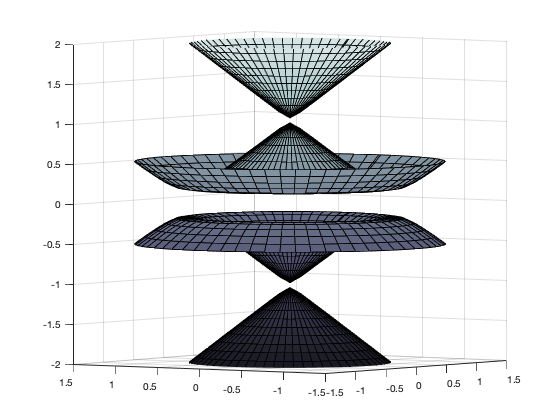}
\end{subfigure}
\begin{subfigure}{.5\textwidth}
    \centering
\includegraphics[width=1\textwidth]{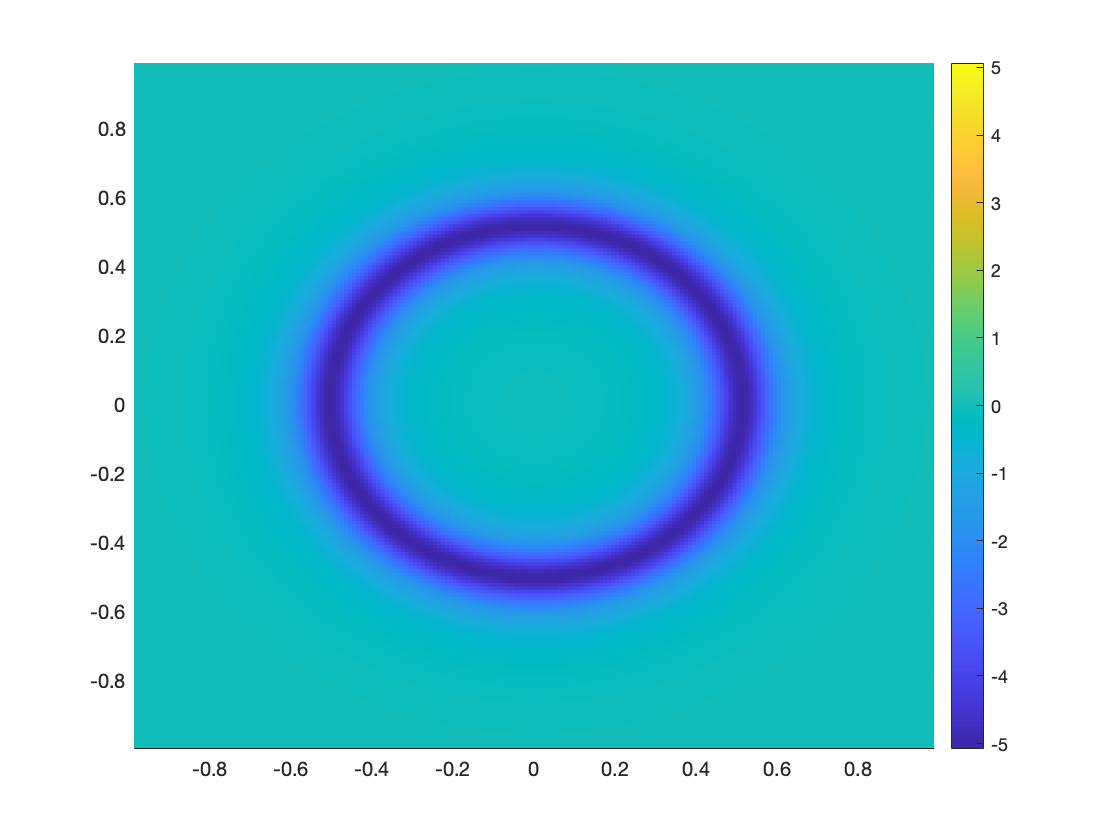}
\end{subfigure}
\caption{Here we plot the eigenvalues of $H_+(\xi)$, i.e. the band structure on the left. On the right, we plot the curvature, i.e., the integrand of $W_+^3+W_+^4$ in \eqref{e:Wpm}. Here we set $\lambda = 0.2$ and $\Omega = 1$. As expected, curvature concentrates on the line $|\xi| = 1$ for $|\lambda|$ small.
}
\label{fig:tblg_bs}
\end{figure}
\begin{figure}[htbp!]
    \centering
    \includegraphics[width=.6\textwidth]{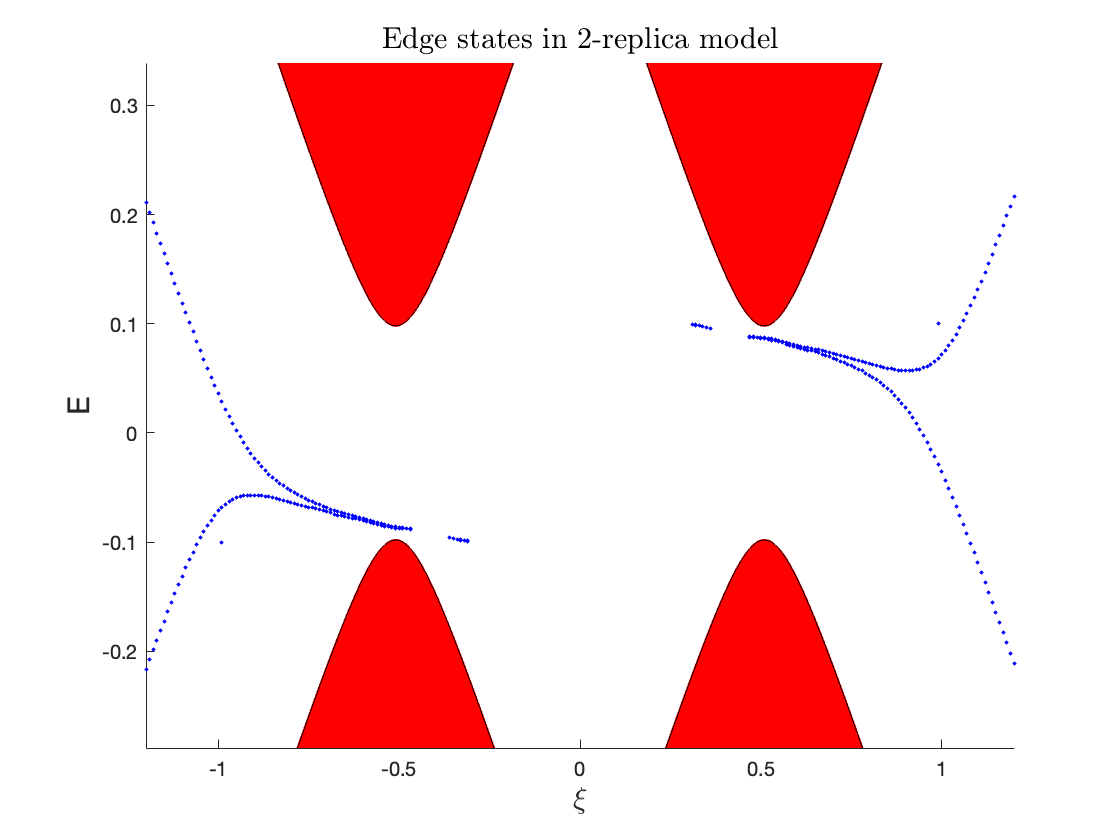}
    \caption{Bulk band structure (filled) is plotted along with calculated edge state eigenvalues (points) of a periodized $H_e(\xi_1)$. In particular, the plotted interface eigenvalues are the nearest two eigenvalues to $E = 0$.}
    \label{fig:edge_state}
\end{figure}
\begin{theorem}
\label{thm:tblg-bulk}
The bulk system $H$  has a band gap for all $\lambda \neq 0$. Further, 
\begin{equation}\label{eq:invthm1}
   \tb{-} W_\pm = \pm \eta \, \sign(\Omega).
\end{equation}
The bulk-interface correspondence then gives us
\begin{equation}\label{eq:invthm2}
    2\pi \sigma_I = W = \tb{-}2\eta \, \sign(\Omega).
\end{equation}
\end{theorem}
\begin{figure}[htbp!]
\centering
\includegraphics[width=.6\linewidth]{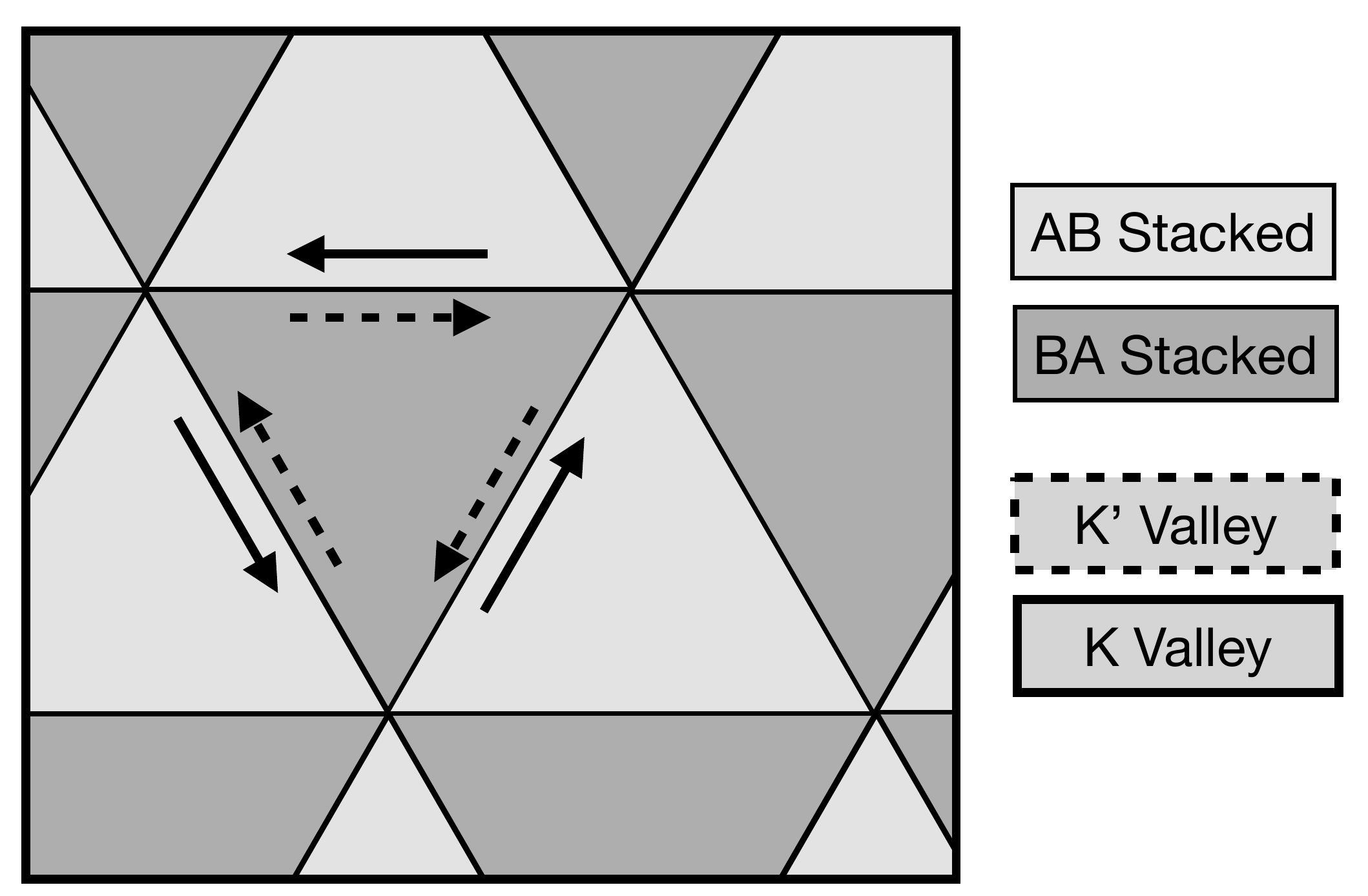}
\caption{ $K$ and $K'$ valley current directions are shown for each edge. Note that the currents per valley are always moving in opposite directions. }
\label{fig:edge_currents}
\end{figure}
\tb{ The proof of this result is postponed to Appendix A. The proof provides a more detailed expression for the components describing $W$. In particular, assume $\Omega>0$ and $\lambda>0$ for concreteness and define $\beta=2\Omega+\sqrt{4\Omega^2+\lambda^2}$. We will show that
\[
  -W^3_+= +W^3_- = \frac{\lambda^2+3\beta^2}{2(\lambda^2+\beta^2)},\quad -W^4_+ = + W^4_- = \frac{\lambda^2-\beta^2}{2(\lambda^2+\beta^2)}.
\]
This is a surprising result. If $W^j_\pm$ involved an integration over a compact domain such as $\Gamma^*$, then each of these numbers would be guaranteed to be an integer (a Chern number). We already mentioned that for the Dirac operator $W_\pm=\pm\frac12$. Here, however, the terms $W^j_\pm$ take a continuous range of values as $(\Omega,\lambda)$ vary. Moreover, we show in Appendix~\ref{sec:aa} using the gluing procedure proposed in~\cite{bal3} that $W^3_+-W^4_-$ and $W^4_+-W^3_-$ are both topological invariants (Chern numbers) with integral values (here equal to $1$). We thus obtain the peculiar result that the gluing of the projector of the fourth/third energy band of the AB Hamiltonian with (an appropriate rotation of) the projector of the third/fourth energy band of the BA Hamiltonian defines a topological invariant. The details of this computation are provided in Appendix~\ref{sec:aa}. 

The results of Appendix~\ref{sec:aa} also provide a spectral gap about $0$ of size $2E_0=\frac{2|\Omega||\lambda|}{\sqrt{4\Omega^2+\lambda^2}}$. So the above theorem applies for a spectral density $\varphi'(h)$ in \eqref{eq:sigmaI} supported in the interval $(-E_0,E_0)$.
}

We note that we assumed an ordering to the AB and BA stacking. If we switch the order, we would swap the sign of $W$, and hence $\sigma_I$. Hence each edge has an edge current dependent on the order of the AB-BA regions and also the valley index of interest (see Figure~\ref{fig:edge_currents}).

We also note that the structure of the Hamiltonians $H$ and $H_e$ are identical to the so called {\em 2-replica model} static Hamiltonian used to approximate the dynamics of the circularly polarized laser-driven graphene system, which is a Floquet Topological Insulator~\cite{massatt2021, Perez_2014}. The reason these two very different physical systems share a common Hamiltonian is that in both situations the bulk systems are described as a coupling between Dirac cones shifted relative to each other in energy. The overlapping cones' intersection forms a ring, and the coupling between these Dirac cones hybridizes the bands to form ring-shaped band gaps. We remark that we expect bulk multilayer stacked graphene with all layers either AB stacked or BA stacked form the larger {\em replica} systems corresponding to a collection of uniformly shifted Dirac cones. Here the number of replicas or shifted Dirac cones corresponds to the number of sheets, while in the Floquet setting the number corresponded to the number of included time Fourier frequencies. In keeping with the Floquet terminology, we also refer to our Hamiltonians $H$ and $H_e$ as a {\em 2-replica model} as it corresponds to two shifted Dirac cones.

\section{Valley Coupling}\label{sec:valleys}

As we saw from the previous section, both valley Hamiltonians are topologically non-trivial. However, if valley coupling is allowed, then the total asymmetric current summed over the two valleys vanishes making the total system topologically trivial.


In section \ref{subsec:valley_model}, we introduce a simplified differential operator that captures the two-valley physics, and construct one-dimensional and two-dimensional effective valley-coupling models using a multiscale expansion. In section \ref{subsec:conductivity}, we then define a valley conductivity object. In section \ref{subsec:numericalvalleys}, we use numerics to calculate currents demonstrating what regime of defects yield good valley separation, and when does valley decoupling fail. In section \ref{subsec:scattering} we use scattering theory to explain the numeric results from section \ref{subsec:numericalvalleys}.

\subsection{Valley coupling models}
\label{subsec:valley_model}

To model a two-valley system and inter-valley coupling, consider first the bulk Hamiltonian without mass
\[
H_b = \Phi(D_x)\sigma_1 + D_y \sigma_2.
\]
\tm{Here $\Phi \in C(\mathbb{R})$ is a function with exactly two zeros corresponding to two Dirac points. We will assume $\Phi$ is twice differentiable at these zeros.}
In frequency space, we denote the wavenumber pair $(\xi,\zeta)\in\mathbb{R}^2$. We have changed the notation from section \ref{sec:tblg} as now we wish to emphasize the direction along the edge $\xi$. We have
\[
H_b(\xi,\zeta) = \Phi(\xi)\sigma_1 + \zeta \sigma_2,
\]
which admits band structure.  We will assume for simplicity of presentation throughout this section that the Dirac points are located at frequencies $(\xi,\zeta) = (\pm 1,0)$.
\begin{figure}[ht]
\centering
\includegraphics[width=.6\textwidth]{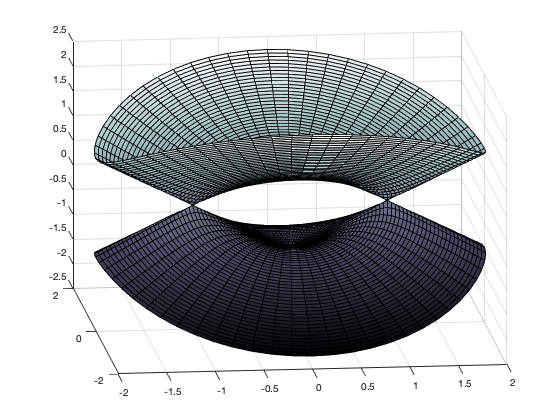}
\caption{Band structure of two-valley linear Dirac Hamiltonian, $\Phi(\xi) = 1-|\xi|$.}
\label{fig:bs_v3}
\end{figure}
For example, see Figure \ref{fig:bs_v3} for $\Phi(\xi) = |\xi| - 1$. Another natural choice for $\Phi$ leading to a differential operator $H_b$ is $\Phi(\xi)=\frac{1}{2}(\xi^2-1)$. In both cases, the valley at Dirac points $\xi = 1$ and $\xi = -1$ have topological indices given by $1$ and $-1$ respectively.
We will  primarily focus on small highly oscillatory defects as a model for distribution of small defects. 

Consider a semiclassical theory with $0<\eps\ll1$ scaling near Dirac points so that the two valleys are separated by wavenumbers of order $2\eps^{-1}$. We focus on a small energy window near $E=0$ and thus analyze an equation of the form
\begin{align}
\label{e:semiclassical}
    \biggl(\Phi(\eps D_x) \sigma_1 + \eps D_y \sigma_2 + \lambda \frac y\eps \sigma_3 + \eps V(x,\frac x\eps,\frac y\eps) - \eps E\biggr)\psi = 0.
\end{align}
In the absence of a domain wall ($\lambda=0$) and perturbing potential ($V=0$), assuming $|\Phi'(\pm 1)| = 1$, this system has a circular \tm{dispersion relation level set} at small energies as is typical for Dirac points. 

We consider a formal derivation in the limit $\eps\ll1$ of an inter-valley coupling model whose qualitative predictions will be confirmed by numerical simulations. While many steps can undoubtedly be justified more rigorously using methodology such as in \cite{weinstein2011}, we will not do so here.

The domain wall $\lambda\frac y\eps \sigma_3$ is assumed to be steep with $\lambda\gg1$ as this prevents coupling between various confined modes in $y$. Introduce the microscopic variable $Y=y/\eps$ and look for solutions of the form $\psi(Y=\frac y\eps,x)$. Applied to such functions, \eqref{e:semiclassical} becomes
\begin{align}\label{eq:multiy}
     \big(D_Y \sigma_2 +  \lambda Y \sigma_3 + \Phi(\eps D_x) \sigma_1 + \eps V(x,\frac x\eps,Y) - \eps E\big)\psi=0.
\end{align}
When $\lambda$ is large, then the two first terms dominate in the sense that 
\[
  (D_Y \sigma_2 +  \lambda Y \sigma_3)\phi_0(Y) =0,\quad \phi_0(Y) = (\lambda \pi)^{-1/4} 
  e^{-\frac12\lambda  Y^2} \cdot \frac{1}{\sqrt{2}}
  \begin{pmatrix} 1 \\ -1 \end{pmatrix},
\]
while $D_Y \sigma_2 +  \lambda Y \sigma_3$ applied to any other Hermite function provides a term of order $|\lambda|\gg1$ formally much larger than the term $\Phi(\eps D_x) \sigma_1$.

Decomposing $\psi(Y,x)=\phi_0(Y)\psi(x)$ for a scalar function $\psi(x)$ thus neglecting higher-order Hermite functions in $Y$, we obtain from \eqref{eq:multiy} (multiplying by $\phi_0^T$ and integrating in $Y$) a reduced Hamiltonian for $\psi(x)$ given by


\begin{align}\label{eq:model1d}
    \Big(-\Phi(\eps D_x) + \eps V(x,\frac x\eps) -\eps E \Big) \psi =0.
\end{align}
Here, $V(x,X)=(\phi_0(Y),V(x,X,Y)\phi_0(Y))_Y$. We also used that the current
\[
  (\phi_0,\sigma_1\phi_0)_Y=-1.
\]
In both cases, $(\cdot,\cdot)_Y$ is the usual inner product on $L^2(\Rm_Y;\Cm^2)$.

Define now multiscale functions $f(x,X=\frac x\eps)$  in the variable $x$ so that $D_x$ becomes $D_x+\eps^{-1}D_X$. Then the multiscale operator is
\begin{align}
   -\Phi(D_X) + (-\Phi'(D_X)D_x+V(x,X)-E)\eps + O( \eps^2).
\end{align}
Going to the Fourier domain $X\to K$ we obtain
\begin{align}
   -\Phi(K) + (-\Phi'(K)D_x+\hat V(x,K) * - E)\eps + O( \eps^2).
\end{align}
Here $*$ is convolution in the Fourier variable $K$.
To leading order we get
\begin{align}\label{eq:psi0}
    \Phi(K)\psi_0=0,\qquad \mbox{ i.e., }\quad \psi_0=\psi_+(x)\delta(K-1)+\psi_-(x)\delta(K+1).
\end{align}
The next order is
\begin{align}
    -\Phi(K)\psi_1 -\Phi'(K) D_x \psi_0 + \int \hat V(x,K-K')\psi_0(x,K') dK' - E\psi_0 =0.
\end{align}
We thus find the coupled equations
\begin{align}\label{eq:coupledpsi}
    -\Phi'( 1) D_x \psi_+ +\hat V(x,0)\psi_+ + \hat V(x,2) \psi_- &=E \psi_+\\
    -\Phi'(-1) D_x \psi_- +\hat V(x,0)\psi_- + \hat V(x,-2) \psi_+ &=E \psi_-.
\end{align}
This generates inter-valley coupling if by valleys we mean the wavenumbers close to $K =\pm1$.
Letting
\begin{align}
\label{e:h2}
    H_2 = \begin{pmatrix}  -\Phi'(1)D_x +\hat V(x,0) & \hat V(x,2) \\ \hat V(x, -2) & - \Phi'(-1) D_x + \hat V(x,0)\end{pmatrix},
\end{align}
we obtain the eigenproblem $H_2\psi = E \psi$, $\psi = (\psi_+,\psi_-)^T$, as a one-dimensional model incorporating valley coupling. 

We now build a two-dimensional Hamiltonian that incorporates a domain wall and valley coupling. Generalizing the above Hamiltonian to include a $y$-dependent potential,
\begin{align}
    H_2^y = \begin{pmatrix}  -\Phi'(1)D_x +\hat V(x,0,y) & \hat V(x,2,y) \\ \hat V(x, -2,y) & - \Phi'(-1) D_x + \hat V(x,0,y)\end{pmatrix},
\end{align}
we then define the $4\times4$ system
\begin{equation}
 \label{e:h4}
     H_4 = D_y I_2 \otimes \sigma_2 + \lambda y I_2 \otimes \sigma_3 + H_2^y \otimes \sigma_1.
\end{equation}
We choose this Pauli matrix representation such that $H_4$ is consistent with $H_2$ when $V$ is $y$-independent. In particular, for $e_\pm$ being the standard basis of $\mathbb{R}^2$ corresponding to valley index, $e_\pm \otimes \phi_0(y)$ are still in the kernel of the domain wall terms (first two terms), and $\phi_0(y)$ is an eigenfunction of $\sigma_1$, meaning when $V$ is $y$-independent the edge state eigenfunctions of $H_2$ naturally generalize to eigenfunctions of $H_4$. 
 
The operators in \eqref{e:h2} and \eqref{e:h4} are used for the numerical simulations of valley coupling effects.
\begin{remark}
Instead of $\eps V(x,\frac x\eps)$, we could consider a larger-amplitude lower-frequency perturbation $V(x)$ in \eqref{e:semiclassical}, another form of defect studied numerically in \cite{macdonald_junction, ValleyNature_2016}. Let us make a few brief remarks on this setting. In the semiclassical limit of large valley separation of order $2\eps^{-1}$, we derive as above the one-dimensional semiclassical problem
\begin{align}
    (-\Phi(\eps D_x) + V(x) -E)\psi=0.
\end{align}
Consider the well-studied case $\Phi(\xi)=1-\xi^2$ for concreteness and see, e.g., \cite{sjostrand} for an analysis of such semiclassical problems. 

We may find to arbitrary algebraic accuracy in powers of $\eps$ solutions of the form
\begin{align}
    \psi(x) = \sum_{\pm} e^{i S_\pm(x)/\eps} a_{\pm}(x;\eps),  \quad -\Phi(\partial_x S_\pm) + V(x) -E=0.
\end{align}
The latter eikonal equation for $S_\pm$ admits two solutions 
\begin{align}
   S_\pm(x) = \pm \int_0^x (\Phi(0) +E- V(z))^{\frac12} dz.
\end{align}
We assume $\Phi(0)+E-V(0)>0$ so that propagating modes exist at $x_0$ with local wavenumber given by $\xi=\sqrt{\Phi(0)+E-V(0)}$ at the energy of interest $E$. Assuming $\Phi(0)+E-V(x)$ remains positive for all values of $x$, then $\psi$ decomposes into two components $e^{i S_\pm(x)/\eps} a_{\pm}(x;\eps)$ that do not interact, resulting in extremely limited valley coupling. 

Now, if $\Phi(0)+E-V(x)<0$ is possible, then such values of $x$ are forbidden energetically. Points $x_0$ such that $\Phi(0)+E-V(x_0)=0$, assuming they are isolated (i.e., $V'(x_0)\not=0$), are (caustic) turning points. Any signal arriving at $x_0$ from the area where $\Phi(0)+E-V(x)>0$ will fully turn around so that the amplitudes $a_\pm$ are now fully coupled.  See, e.g., \cite{gerard1988precise} for a detailed analysis of coupling (scattering) coefficients in various semiclassical settings. 

We thus obtain a situation of extremely limited coupling if the eikonals $S_\pm$ remain real valued, and full valley-coupling at turning points if the eikonals $S_\pm$ may become complex-valued. We do not consider such a case further here.

\end{remark}

\subsection{Valley Conductivity}
\label{subsec:conductivity}
Now that we have constructed valley Hamiltonians $H_2$ and $H_4$, we define valley contributions to conductivity so that we can study the separation of valleys under varying forms of perturbation.
Recall we originally considered the conductivity object $\sigma_I= \Tr{i[H,P]\varphi'(H)}$ with $P(x)$ a regularized Heaviside function. Here we consider $H$  either $H_2$ or $ H_4$. However, for the two-valley system $\sigma_I$ will always be $0$ as when the two valleys are combined, we have a topologically trivial system. 

We thus introduce a conductivity corresponding to a single valley so we can see the separation in valley contributions to conductivity. $H_2$ and $H_4$ are naturally structured to separate the valleys, making the definition of single valley conductivities reasonably straightforward. \tm{To motivate our definition of valley conductivity,} consider a wavefunction $(\psi_+,\psi_-)^t$ and define \tm{valley switch function}
\begin{equation}\label{eq:valleyprojector}
   P_+= \begin{pmatrix}
    P & 0 \\ 0 & 0
   \end{pmatrix}.
\end{equation}
We consider the current on one side of the interface generated by $P$ via the expression $(\psi_+,\dot P_+\psi_+)= \Tr{\dot P_+ \psi_+\otimes \psi_+^*}$. If $U(t) = e^{-i Ht}$ is the evolution of the Schr\"odinger equation,
then 
\begin{equation}
    \dot P_+ := \partial_t (U(t)P_+U^*(t)) = U(t)[H,P_+]U^*(t).
\end{equation} 
We wish to rewrite including all wavefunctions corresponding to the edge. We pick $\varphi \in \fs(0,1;-E_0,E_0)$ with $[-E_0,E_0]$ strictly in the bulk gap to capture spectral interface information. In particular, $\varphi'(H)$ captures all relevant modes. 
Decomposing the Hamiltonian $H$ as $(H_{ij})$ for $1\leq i,j\leq 2$ valley indices, we obtain
\[
 [H,P_+] = \begin{pmatrix}
  [H_{11},P] & -P H_{12} \\ H_{21} P & 0
 \end{pmatrix}
\]
so that, with the same decomposition of $\varphi'(H)$, 
\[
  [H,P_+]\varphi'(H) = \begin{pmatrix}
   [H_{11},P] \varphi'_{11}(H) - P H_{12}\varphi'_{21}(H) & * \\ * & H_{21} P \varphi'_{12}(H)
  \end{pmatrix}.
\]
The matrix trace therefore gives
\begin{equation}\label{eq:valleyconductivity}
  \sigma_{I,+}(H) = \Tr i[H,P_+]\varphi'(H) = \Tr i[H_{11},P] \varphi'_{11}(H) + \Tr i[H_{21},P] \varphi'_{12}(H).
\end{equation}
The conductivity for the opposite valley is
\[
  \sigma_{I,-}(H):= \Tr i[H,P_-]\varphi'(H) = \Tr i[H_{22},P] \varphi'_{22}(H) + \Tr i[H_{12},P] \varphi'_{21}(H).
\]
For the above equation, $H$ is first-order with an off-diagonal term, which is a multiplication operator in the case of $H_2$ and $H_4$. This yields $[H_{21},P]=0$. 
Hence
\begin{align}
\sigma_{I,+}(H) = \Tr  i[H_{11},P]\varphi_{11}'(H),\qquad
\sigma_{I,-}(H) = \Tr  i[H_{22},P]\varphi_{22}'(H).
\end{align}
\subsection{Numerical example}
\label{subsec:numericalvalleys}


First, we illustrate the deviation of the valley conductivity~\eqref{eq:valleyconductivity} from the quantized value $\pm 1$ in the presence of a compactly supported potential $V$ which couples the two valleys, see~\eqref{e:semiclassical}. Consider the leading ansatz \eqref{eq:psi0} with
\begin{align}
   \psi_0(x) = e^{ix/\eps} \psi_+(x) +  e^{-ix/\eps} \psi_-(x)
\end{align}
with the two Dirac cones at frequencies $\pm \eps^{-1}$.
Plugging the ansatz into \eqref{eq:model1d}, dividing by $\eps$ and multiplying by $e^{\pm ix/\eps}$, we obtain, consistently with \eqref{eq:coupledpsi} after neglecting highly oscillatory terms $e^{\pm 2ix/\eps}D_x \psi_\pm$ and $e^{\pm 2ix/\eps}E$ that formally average to $0$ according to \eqref{eq:coupledpsi}, the following one-dimensional coupled Hamiltonian:
\begin{equation}\label{e:h2_num}
    H^\varepsilon_2 = \begin{pmatrix}  -D_x  & 0 \\  0 & D_x \end{pmatrix}
    + V(x,x/\varepsilon) \begin{pmatrix} 1 & e^{-2i  x/\varepsilon} \\ e^{2i  x/\varepsilon} & 1 \end{pmatrix},
\end{equation}
so that to leading order $H^\eps_2 \psi = E \psi$ for $\psi = (\psi_+(x) , \psi_-(x))^T$.
We choose a two-scale potential of the form:
\begin{equation}\label{eq:numpotential_1d}
    V(x, X) = V_0 \chi(x) \cos(\omega X),
\end{equation}
where $\chi(x)$ is a smooth, compactly supported function. Model~\eqref{e:h2_num} allows us to inspect the effects of the frequency of the fast component of the potential at a finite value of $\varepsilon$, such that some coupling is achieved at frequencies different from $\pm 2$, and homogenizes to~\eqref{e:h2} in the limit $\varepsilon \to 0$, when $\Phi'(\pm 1) = \pm 1$. 

Here, we perform computations on the interval $[-50,50]$ with periodic boundary conditions comparing all three models: \eqref{e:h2_num} with $\varepsilon = 1$ and potential~\eqref{eq:numpotential_1d}, and the models~\eqref{e:h2} (one-dimensional) and~\eqref{e:h4} (two-dimensional) obtained in the limit $\varepsilon \to 0$ with $\omega =2$ with respective effective potentials described by
\begin{equation}\label{eq:numpotential_h2_h4}
    \widehat{V}(x, 0) = 0, \quad  \widehat{V}(x, \pm 2) = \frac{V_0}{2} \chi(x) \quad \text{and} \quad \widehat{V}(x,0,y) = 0, \quad \widehat{V}(x,\pm 2,y) = \frac{V_0}{2} \chi(x) \xi(y).
\end{equation}
We choose functions $\chi(x)$ and $\xi(y)$ supported in the interval $[-20,20]$ and $[-10,10]$ and taking constant value $1$ in the interval $[-10,10]$ and $[-5,5]$, respectively. Details of the pseudo-spectral discretization approach are postponed to section~\ref{sec:numerics}. The resulting conductivity $\sigma^+$ is plotted on Figure~\ref{fig:valley_conductivity} as a function of $x_0$, the center of the spatial switch function~\eqref{eq:valleyprojector} and $\omega$, the spatial frequency of the defect~\eqref{eq:numpotential_1d}. In particular, we observe that at low frequencies ($\omega \approx 0$) the conductivity stays very close to the quantized value $\sigma^+ = 1$, while at $\omega = 2$, a maximum coupling between the valleys is achieved, locally reducing the conductivity measure by a large factor. Note that as $\varepsilon \to 0$, the limit model~\eqref{e:h2} shows coupling only at frequency $\omega = 2$. 
On the other hand, away from the defect the perturbation of the conductivity decays extremely fast regardless of the frequency $\omega$. 

As a second experiment, we compare the valley conductivity as a function of perturbation amplitude $V_0$ and energy $E$ for Hamiltonians~\eqref{e:h2} and~\eqref{e:h4}, that is by numerically computing the conductivity~\eqref{eq:valleyconductivity} when shifting the Hamiltonian as $H \to H - E$. We observe that the effect of the intervalley coupling term is to open a band gap at the Dirac cones $E = 0$, reducing the conductivity in a small range of energies around zero for both the one-dimensional, two-valley Hamiltonian~\eqref{e:h2} and the two-dimensional, two-valley Hamiltonian~\eqref{e:h4}, showing excellent qualitative agreement between the two models.

\begin{figure}[ht!]\centering
\includegraphics[width=.4\textwidth]{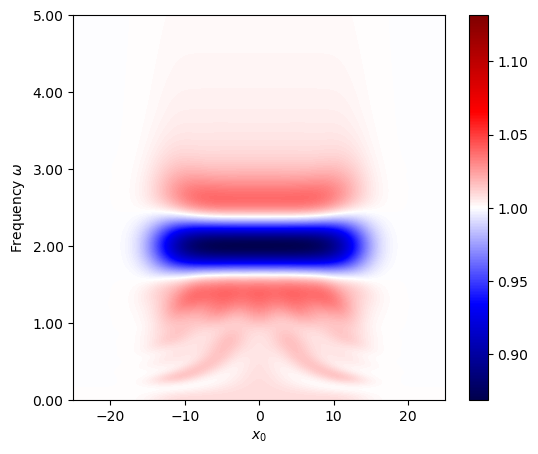}
\hspace{.5in}
\includegraphics[width=.4\textwidth]{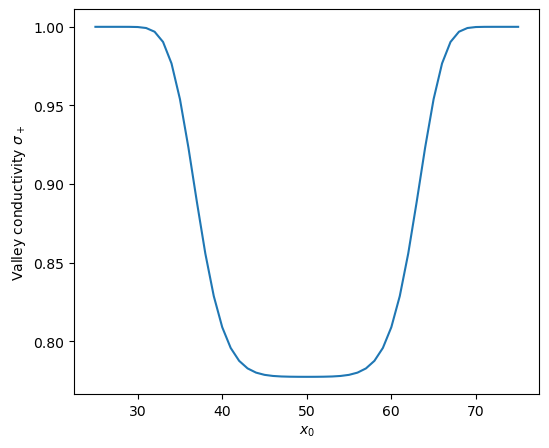}\\
\includegraphics[width=.4\textwidth]{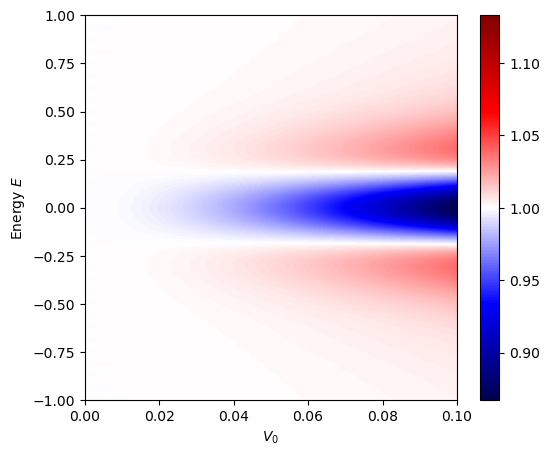}
\hspace{.5in}
\includegraphics[width=.4\textwidth]{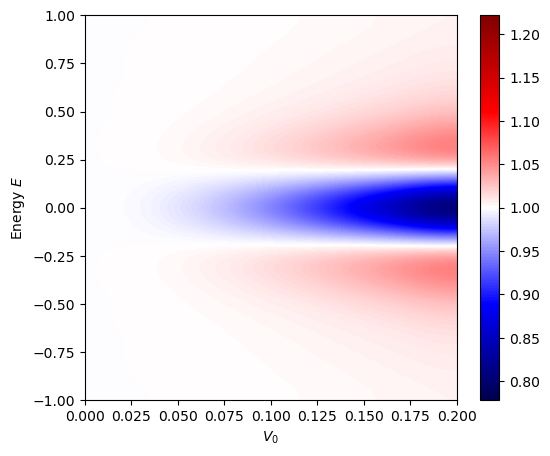}

\caption{Valley conductivity $\sigma_{I,+}$ using the 1D models~\eqref{e:h2_num},~\eqref{e:h2} (left) and 2D model~\eqref{e:h4} (right), as a function of: spatial switch function center $x_0$ and frequency $\omega$ of the perturbation with amplitude $V_0 = 0.1$, using model~\eqref{e:h2_num} (top left); spatial switch function center $x_0$, with a perturbation with amplitude $V_0 = 0.2$, using model~\eqref{e:h2} (top right); as a function of amplitude $V_0$ and energy $E$ with frequency $\omega = 2$ and at $x_0 = 0$ for the 1D model~\eqref{e:h2} (bottom left) and 2D model~\eqref{e:h4} (bottom right).} \label{fig:valley_conductivity}
\end{figure}

\subsection{One dimensional scattering theory}
\label{subsec:scattering}

The conductivity $\sigma_{I,+}$ in Fig.\ref{fig:valley_conductivity} (left) is close to $1$ for $\omega$ away from $2$ as expected from minimal valley coupling and \pc{drops significantly} for $\omega=2$, which is consistent with strong valley-coupling. For intermediate values of $\omega$, we also observe a conductivity larger than $1$. We now briefly consider a simplified one-dimensional model that exhibits such a behavior.

We revisit the one dimensional scattering problem \eqref{eq:model1d}, which we rewrite for $\Phi(\xi)=\frac{1}{2}(\xi^2-1)$ at $\eps=1$ and for a slightly different expression for the potential as the Helmholtz equation
\begin{align}\label{eq:helm1d}
    (D_x^2 +V(x)-E) \psi=0.
\end{align}
Assume $E=k^2>0$ so that when $V=0$, we obtain $e^{\pm i k x}$ as standing (propagating) modes associated with each valley. We recall that $-D_x^2=\partial^2_x$. In the presence of $V\not=0$, these modes are coupled as described by \eqref{e:h2} in the semiclassical regime. We wish to derive an explicit expression for valley conductivities using a scattering matrix formalism in the setting where $P'=\delta_{x_0}$ for $x_0\in\Rm$ and where $\varphi'=\delta_E$ so that we estimate conductivities at fixed $x_0$ and energy $E$.

The two valley conductivities estimate the signals moving from left to right proportional to $e^{ikx}$ and from right to left proportional to $e^{-ikx}$, respectively. Let us derive explicit expressions for those in terms of coefficients of scattering matrices.  We wish to compute $\sigma_I^\pm[x_0,E]$ for $\pm$ corresponding to left moving and right moving signals, respectively.

Consider \eqref{eq:helm1d} with $E=k^2$ and $V$ compactly supported on $[x_L,x_R]$. We may then construct a standard scattering matrix such that $\psi(x<x_L)=e^{ikx}+R_+e^{-ikx}$ and $\psi(x>x_R)=T_+e^{ikx}$ for an incoming plane wave from the left of the support of $V$ and such that $\psi(x<x_L)=T_-e^{-ikx}$ and $\psi(x>x_R)=e^{-ikx}+R_-e^{ikx}$ for an incoming plane wave from the right of the support of $V$. 

Let us now decompose $V$ into two subcomponents $V_{L,R}$ supported on $[x_L,x_0]$ and $[x_0,x_R]$, respectively when $x_L<x_0<x_R$. For each subcomponent, define the scattering matrices
\[
  S^L=\begin{pmatrix}
   R_+^L & T_-^L \\ T_+^L & R_-^L
  \end{pmatrix},\qquad \mbox{ and }\qquad  S^R=\begin{pmatrix}
   R_+^R & T_-^R \\ T_+^R & R_-^R
  \end{pmatrix},
\]
where the coefficients $R$ and $T$ are interpreted as above. 

These two matrices could be appropriately combined to construct a scattering matrix on the interval $[x_L,x_R]$ called $S$; we do not need the details of this object.

Both $S^L$ and $S^R$ as well as $S$ for the whole interval are unitary scattering matrix so that $SS^*=S^*S=I$.
In particular, we observe that $|T_+|=|T_-|$ and $|R_+|=|R_-|$.

Consider the (unique) solution $\psi_L(x)$ of \eqref{eq:helm1d} at energy $E$ equal to $e^{ikx}+Re^{-ikx}$ at $x_L$ and let $(u^L,v^L)$  be defined such that $\psi_L(x_0)=u^L e^{ikx_0}+v^Le^{-ikx_0}$. For this to hold, we formally assume that $V=0$ at $x_0$.
Using the above matrices, we find
\[
  u^L = T^L_+ + R^L_- v^L,\quad v^L = R^R_+ u^L,\qquad u^L=\frac{T_+^L}{1-R^L_-R^R_+}, \quad v^L=\frac{ R^R_+ T_+^L}{1-R^L_-R^R_+}.
\]
Considering the solution equal to $e^{-ikx}+Re^{ikx}$ at $x_R$ instead, we find for $(u^R,v^R)$ the solution at $x_0$
\[
   v^R=T_-^R + R_+^R u^R,\quad u^R=R^L_- v^R,\qquad v^R=\frac{T_-^R}{1-R_-^LR_+^R},\quad u^R=\frac{R_-^LT_-^R}{1-R_-^LR_+^R}.
\]
The above corresponding solutions $\psi_L(x)$ and $\psi_R(x)$ of \eqref{eq:helm1d} are orthogonal in the following sense. \tb{Assume an inner product} $(u,v)_M= (2M)^{-1} \int_{-M}^M \bar u v dx$ on the interval $(-M,M)$. The contributions of that integral on the support of $V$ vanishes in the limit $M\to\infty$. The two solutions constructed above are of the form
\[
  e^{ikx} +R_+ e^{-ikx} \ \mbox{ and } \ T_+e^{ikx}
\]
for the first solution on the left and right of the support of $V$ and 
\[
  T_-e^{-ikx} \ \mbox{ and } \ e^{-ikx} + R_-e^{ikx}
\]
for the second solution in the same sense.
The inner product of these two functions for $k\not=0$ is asymptotically given by $R_+^*T_-+T_+^*R_-=0$. In other words, these two solutions correspond to two simple branches of absolutely continuous spectrum (parametrized by $E$) of the operator that live in orthogonal subsets of the Hilbert space.  

The conductivity associated to the signal moving from left to right, which we associate with one of the valleys is given by the signal at $x_0$ moving to the right for the mixture of states having energy $E$. Thanks to the above orthogonality, the trace (sum over densities of orthogonal modes evaluated at $x_0$) is given by
\[
  \sigma_I^+[x_0,E] :=  |u^L|^2+|u^R|^2 =  \frac{|T_+^L|^2+|R^L_-|^2|T_-^R|^2}{|1-R^L_-R^R_+|^2}.
\]
The conductivity associated to the signal moving from the right to the left is given by
\[
  \sigma_I^-[x_0,E] := -( |v^L|^2+|v^R|^2) =  - \frac{|T_-^R|^2+|R^R_+|^2|T_+^L|^2}{|1-R^L_-R^R_+|^2}.
\]
The sum of the nominators is given by
\[
   |T_+^L|^2+(1-|T^L_-|^2)|T_-^R|^2 - |T_-^R|^2- (1-|T^R_+|^2)|T_+^L|^2=0.
\]
So, we do verify that $\sigma_I^++\sigma_I^-=0$ as it should be for a globally topologically trivial material.

When $x_0$ is outside the support of $V$, say $x_0\geq x_R$, then $T^R=1$ and $R^R=0$ so that 
\[
  \sigma_I^+ = |T_+|^2 + |R_-|^2 =1.
\]
Similarly, $\sigma_I^-$ would be equal to $-1$ outside of the support of $V$. It is therefore only on the support of the fluctuation $V$ that we expect to observe spatial variations (in $x_0$) of the valley conductivities.

Let us consider a simplified setting. Assume that $x_0=\frac12(x_L+x_R)$ and that $V$ is even about $x_0$. Then by symmetry, $R_\pm^L=R^R_\mp$ and $T_\pm^L=T_\mp^R$ so that 
\[
  \sigma_I^+ = \dfrac{|T^L_+|^2(1+|R_-^L|^2)}{|1-(R_-^L)^2|^2} = \dfrac{(1-|R_-^L|^2)(1+|R_-^L|^2)}{|1-(R_-^L)^2|^2}= \dfrac{1-|R_-^L|^4}{|1-(R_-^L)^2|^2}.
\]
Denoting $R^L_-=e^{i\theta} r$, we find
\[
  \sigma_I^+ = \dfrac{1-r^4}{|1-e^{2i\theta}r^2|^2}  \approx 1+2\cos 2\theta\ r^2,
\]
for $r$ small (corresponding to $V$ small). The term $\cos 2\theta$ is positive or negative depending on the value of $\theta$ we obtain from the scattering theory. When $\theta=\frac\pi2$ so that $R_-^L$ is purely imaginary, then we observe the largest drop in conductivity. When $\theta=0$ however, we observe an increase in the local conductivity for small $V$ before it decreases of large values of $V$. These observations are consistent with what we obtained in numerical simulations in Fig.\ref{fig:valley_conductivity} and show that in the presence of intervalley-coupling, the conductivities associated to left-going and right-going modes are spatially dependent and hence not protected topologically.

As a concrete example, assume $V=V_0\chi_{|[-a,a]}$ as a barrier potential for the Left domain above leading to the computation of $R_-^L$. With $q=\sqrt{k^2-V_0}$ assumed positive, standard calculations in the scattering theory of the one-dimensional Schr\"odinger equation then yield
\[
  R^L_- =  \frac{V_0 \sin(2qa) e^{-2ika}}{(2k^2-V_0)\sin(2qa)+2ikq  \cos(2qa)}.
\]
When the width $a$ is small, then to leading order in $a$, the phase of $R^L_-$ is close to $e^{-i\frac \pi 2}$. It may however, take arbitrary values for larger width $a$. Note that $R^L_-=0$ when $2qa=n\pi$ for $n\in\Nm$.




\section{Junction topology}\label{sec:junctions}


This section addresses asymmetric transport 
through a junction in the $xy-$plane. 
We introduce a junction conductivity and compute it by means of a Fedosov-H\"ormander formula, which extends what exists for the flat interface case~\cite{bal3, QB}.
Pseudo-differential calculus will be used throughout; see Appendix B for notation and background.

\sq{Although our main result will apply to a more general class of models and geometries, we use Region 2 from Figure~\ref{fig:tblg} as a motivating example.
As shown, the region selects one vertex from the underlying periodic hexagonal structure.
We push all other vertices to infinity (recall that the side length of each triangle in Figure~\ref{fig:tblg} is inversely proportional to the twist angle, with the latter assumed to be small), so that our model consists of six slices alternating between AB and BA stacking.

The described geometry supports surface waves that propagate along the interfaces between the AB and BA regions. These waves have a preferred direction, e.g. moving from left to right. (For numerical examples, see Section~\ref{sec:numerics}.) The role of the junction conductivity is to quantify the net current associated with these waves at a particular energy. One can define the conductivity so that it measures the contribution from any combination of interfaces (more details below Corollary~\ref{cor:invP}; see also Figure~\ref{fig:H1}). 

Our main result (Theorem \ref{thm:bic} below) is a bulk-interface correspondence, which equates the conductivity to the difference of two integrals depending only on the associated insulators.
In particular, the conductivity does not depend on local properties of the material, such as the way the insulators are joined together at each interface.
Theorem \ref{thm:bic} allows for straight-forward evaluations of the conductivity using previously developed theory for interface models, e.g. \cite{massatt2021,QB}.

The proof of Theorem \ref{thm:bic} requires 
familiar stability properties of the junction conductivity, resembling those of \cite{B-higher-dimensional-2021,bal3,QB,quinn2022asymmetric}.
We 
tie the conductivity to a Fredholm index (Theorem \ref{thm:idx}), with the latter known to be quantized and invariant with respect to a large class of perturbations.
This implies a conservation law (Corollary \ref{cor:invP}) showing (among other things) 
that the currents entering and leaving the junction are the same. 
We also prove invariance of the junction conductivity with respect to the corresponding density of states (Theorem \ref{thm:invvarphi}), semiclassical rescaling (Theorem \ref{thm:h}), and relatively bounded (Theorem \ref{thm:bdd}) and compact (Theorem \ref{thm:compact}) perturbations of the Hamiltonian.
We conclude the section with an application to a $2 \times 2$ Dirac model, which is then analyzed numerically in Section \ref{sec:numerics}.

\subsection{Junction conductivity for tBLG}
Let us now introduce more mathematical detail. The junction Hamiltonian corresponding to region 2 from Figure \ref{fig:tblg} is
}
\begin{align}\label{eq:TBGH}
    \Op (\sigma) =
    H = \begin{pmatrix} \Omega I + D \cdot \sigma^{(\eta)} & \lambda U^*(x,y) \\ \lambda U(x,y) & -\Omega I + D \cdot \sigma^{(\eta)}\end{pmatrix},
\end{align}
where the notation on the left-hand side is defined in Appendix B, and $U(x,y) := \frac{1}{2} ((1+\tilde{m}(x,y)) A + (1-\tilde{m}(x,y)) A^*)$. Here, $\tilde{m} \in \mathcal{C}^\infty$ is positive in regions of AB stacking, negative in regions of BA stacking, and zero on the boundaries between regions. We assume that $\tilde{m} \in \{-1,1\}$ away from these boundaries. 
\sq{Recall \eqref{eq:He}, where the analogous Hamiltonian for an interface (rather than junction) model is defined.}

To construct such a function $\tilde{m}$ explicitly,
define $f_\Theta: \mathbb{T} \rightarrow \mathbb{T}$ by $f_\Theta(\theta) = \sin (3 \theta)$ and
let $f \in \mathcal{C}^\infty (\mathbb{R}^2)$ such that $f(x,y) = r f_\Theta (\theta)$ for all $r \ge 1$, where $(r,\theta)$ are the polar coordinates corresponding to $(x,y)$.
For concreteness, take $f(x,y) := \chi (r) r f_\Theta (\theta)$, where $\chi \in \fs (0,1;\eps,1)$ and $0<\eps<1$. \sq{(Recall the definition of $\fs$ below \eqref{eq:sfdef}.)}
An example of such a function $f$ is plotted in Figure \ref{fig:fg} (left panel).
Since $\chi$ is smooth and vanishes near the origin, it is immediate that $f$ is smooth.
Let $m \in \fs (-1, 1)$ be monotonically increasing with $m(0) = 0$, 
and define $\tilde{m} (x,y) := m(f(x,y))$. \sq{Here, $\fs (a_1,a_2)$ is the union of $\fs (a_1,a_2;c_1,c_2)$ over all $c_1 \le c_2$.}
This definition \sq{of $\tilde{m}$} is consistent with the above constraints, as
$f_\Theta > 0$ if and only if $2k\pi/3 < \theta < (2k+1)\pi/3$ for some $k \in \{0,1,2\}$.
Since $m$ and all of its derivatives are bounded, it follows that $\sigma \in S^1_{1,0}$ with the right-hand side defined in Appendix B. 
Moreover, the structure of $H$ (with $D \cdot \sigma^{(\eta)}$ on the diagonal and no other derivatives) implies the existence of 
a constant $c>0$ such that $|\sigma_{\min} (x, y, \xi, \zeta)| \ge c \aver{\xi, \zeta} - 1$, where $\sigma_{\min}$ is the smallest magnitude eigenvalue of $\sigma$, and for any vector $u$, we define $\aver{u}=\sqrt{1+|u|^2}$.
We see that when $f(x,y)$ is sufficiently large (resp. small), 
$\sigma = \sigma_+$ (resp. $\sigma = \sigma_-$), with $\Op (\sigma_\pm) := H_\pm$ and $H_\pm$ defined in section \ref{sec:tblg}.
By Appendix A, there exists $E>0$ such that $\sigma_\pm$ (equivalently $H_\pm$) has a spectral gap in the interval $(E_1, E_2) := (-E,E)$.

Analogous to \eqref{eq:sigmaI} and \cite{B-bulk-interface-2019,B-higher-dimensional-2021,bal3,QB}, we define the \emph{junction conductivity} by
\begin{align}\label{eq:sigmaIjunction}
    \sigma_I (H,P) := \Tr i[H,P] \varphi'(H),
\end{align}
where $\varphi \in \fs (0,1;-E,E)$. 
The operator (of point-wise multiplication by) $P$ is still a regularized indicator function,
but it can no longer depend on only one variable.
Indeed, for the operator on the right-hand side of \eqref{eq:sigmaIjunction} to be trace-class, it is necessary that $\nabla P$ decay to $0$ 
along the level curves of $f$ (otherwise, the symbol of $[H,P] \varphi'(H)$ would not decay in $\aver{x,y}$). 
This condition is satisfied if we set $P(x,y) = \chi_p (g(x,y))$, where $\chi_p \in \fs (0,1)$ and $g \in \mathcal{C}^\infty (\mathbb{R}^2)$ such that $C_1\aver{x,y} \le \aver{f,g} \le C_2 \aver{x,y}$ for some positive constants $C_1$ and $C_2$.
To construct such a function $g$ explicitly, we can for example
define $g_\Theta : \mathbb{T} \rightarrow \mathbb{T}$ by $g_\Theta (\theta) = \cos (\theta) + \cos (\pi/6) = \cos (\theta) + \sqrt{3}/2$ (note that $g_\Theta < 0$ if and only if $5\pi/6<\theta <7\pi/6$) and take $g(x,y) = \chi(r) r g_\Theta (\theta)$.
See Figure \ref{fig:fg} (center panel) for a plot of $g$.
Since the zeros of $f_\Theta$ and $g_\Theta$ are disjoint, it is clear that the growth condition for $\aver{f,g}$ is satisfied.
As we will show in detail below (see Proposition \ref{trclass}), the symbol of $\varphi '(H)$ (resp. $[H,P]$) decays rapidly in $\aver{f(x,y),\xi,\zeta}$ (resp. $\aver{g(x,y)}$), meaning that $\sigma_I$ in \eqref{eq:sigmaIjunction} is well defined.

\begin{figure}
    \begin{subfigure}{.33\textwidth}
    \includegraphics[scale=0.14]{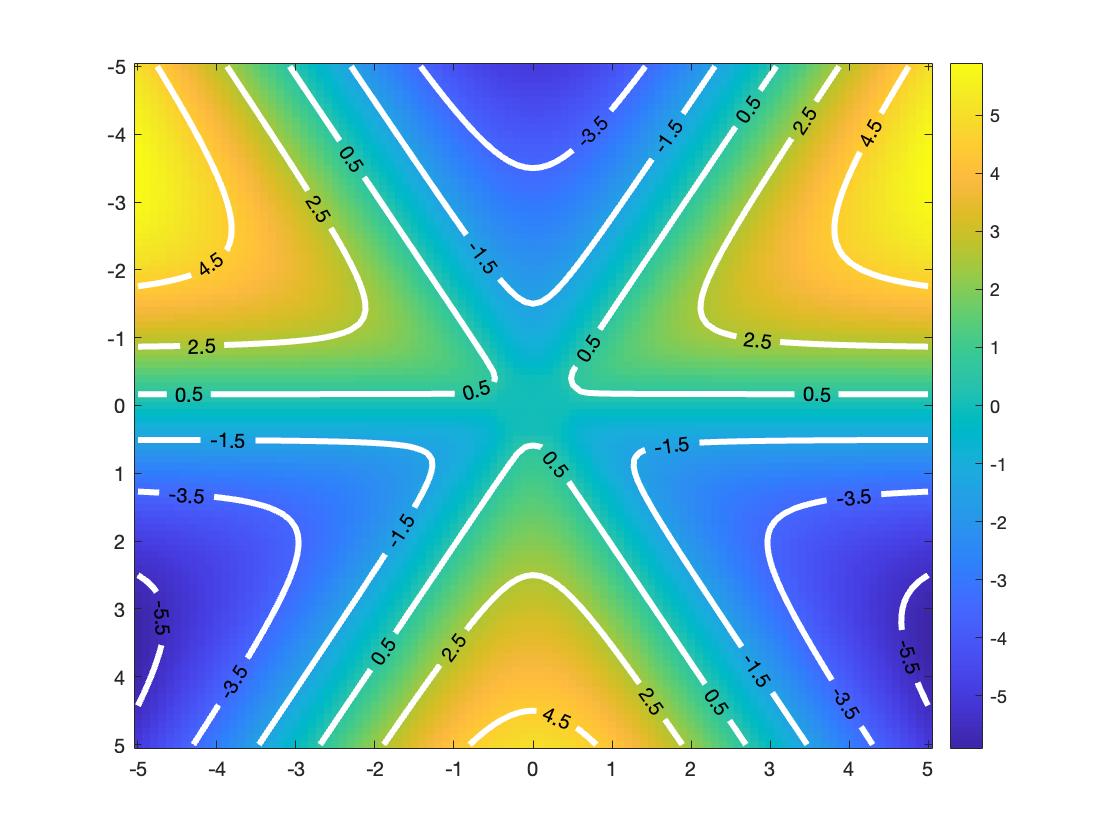}
    \end{subfigure}
    \begin{subfigure}{.33\textwidth}
    \includegraphics[scale=0.14]{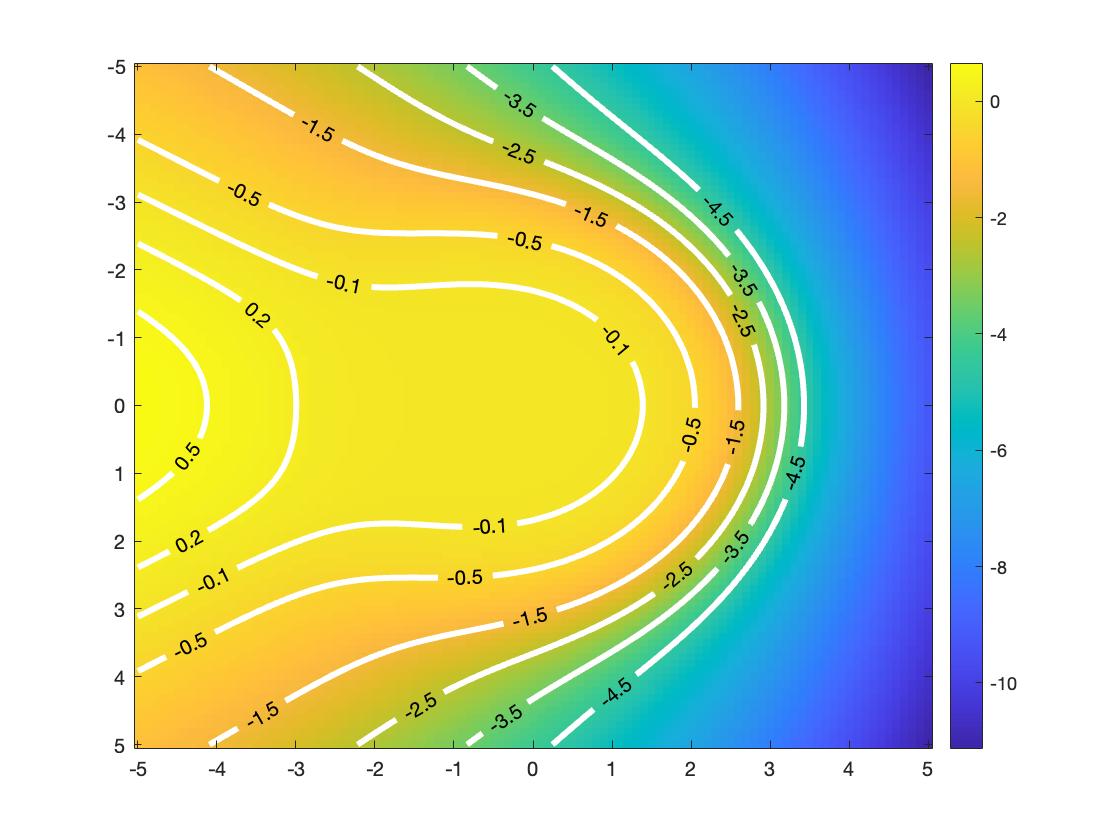}
    \end{subfigure}
    \begin{subfigure}{.33\textwidth}
    \includegraphics[scale=0.14]{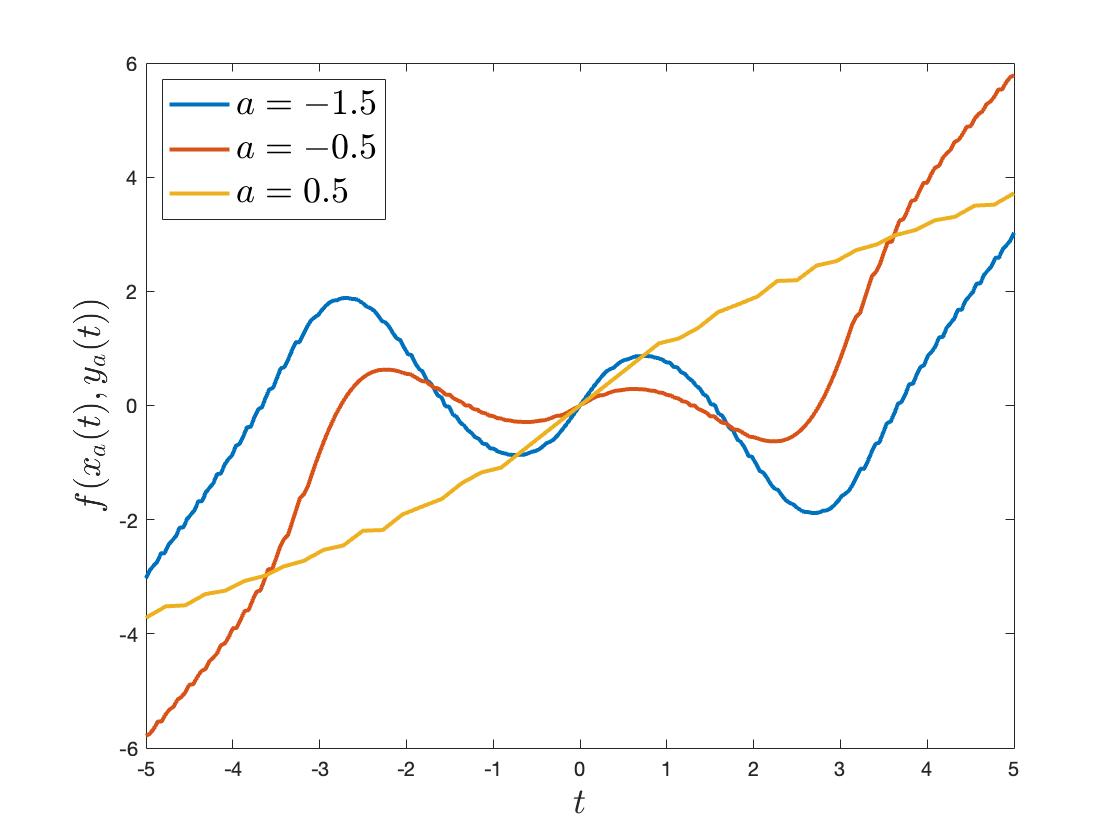}
    \end{subfigure}
\caption{Plots of $f$ (left) and $g$ (center) in the $xy-$plane, with several level curves included. The right panel shows $f$ evaluated along the level curves $\{g = a\}$ for three values of $a$. The curves are parametrized by the variable $t$, with the convention that $y_a (t) \rightarrow \pm \infty$ as $t \rightarrow \pm \infty$.}
\label{fig:fg}
\end{figure}

Due to the factor of $\chi$ in the definition of $g$, the level set $g^{-1} (0)$ is not a curve in the $xy-$plane. Still, there exists a smooth function $(x_0, y_0): \mathbb{R} \rightarrow \mathbb{R}^2$ 
whose range separates the $xy-$plane into two regions, $R_+$ and $R_-$, such that $g \ge 0$ in $R_+$ and $g \le 0$ in $R_-$.
More specifically, we can take
\begin{align*}
    (x_0 (t), y_0 (t)) = \begin{cases}
    (\frac{\sqrt{3}}{2} t, \frac{1}{2} t), &t \le -1\\
    (-\frac{\sqrt{3}}{2} t,\frac{1}{2} t), &t \ge 1
    \end{cases}
\end{align*}
and smoothly connect $(x_0, y_0)$ for $|t| < 1$ so that $g(x_0 (t), y_0 (t)) = 0$ for all $t \in \mathbb{R}$.
For $(t,\xi,\zeta) \in \mathbb{R}^3$,
define 
\begin{align}\label{eq:tauz}
    \tau (t, \xi, \zeta) := \sigma (x_0(t), y_0(t), \xi, \zeta),\qquad \tau_z := z - \tau.
\end{align}
\sq{
By assumption,
there exist $t_0 > 0$ and symbols
$$\tau_{\pm} (\xi,\zeta) \in \Bigg \{ \begin{pmatrix} \Omega I + \xi \sigma_1 + \eta \zeta \sigma_2 & \lambda A^* \\ \lambda A & -\Omega I + \xi \sigma_1 + \eta \zeta \sigma_2\end{pmatrix}, 
\begin{pmatrix} \Omega I + \xi \sigma_1 + \eta \zeta \sigma_2 & \lambda A \\ \lambda A^* & -\Omega I + \xi \sigma_1 + \eta \zeta \sigma_2\end{pmatrix} \Bigg\}$$
such that $\tau (t,\xi,\zeta) = \tau_{\pm} (\xi,\zeta)$ whenever $\pm t \ge t_0$.
The main result of this section is the following bulk-interface correspondence.
\begin{theorem}\label{thm:bic}
    Let $\alpha \in (E_1,E_2)$ and define $\sbz := z-\tau_{\pm}$ with $z=\alpha+i\omega$. 
    Then 
    \begin{align}\label{eq:bic}
        2\pi \sigma_I = \frac{i}{8\pi^2} (\invariantplus - \invariantminus),\qquad \invariantpm = \int_{\mathbb{R}^3} \tr [\sbz^{-1} \partial_\xi \sbz, \sbz^{-1} \partial_\zeta \sbz] \sbz^{-1} d\omega d\xi d\zeta.
    \end{align}
\end{theorem}
We postpone the proof to Appendix B.
Although motivated and introduced by the tBLG model from Section \ref{sec:tblg},
the above result applies to 
any Hamiltonian $H$ and corresponding energy interval $(E_1, E_2)$ satisfying (H1) below; see Section \ref{subsec:general}.
Moreover, the function $g$ (and hence the regularized indicator function $P$)
does not need to be defined as in the above example.

An immediate consequence of Theorem \ref{thm:bic} is}

\sq{
\begin{corollary}\label{thm:main}
    Let $\tilde{H} = \Op (\tau)$ with $\tau = \tau (y,\xi,\zeta)$ and $\tilde{P} (x) =\tilde{P} \in \fs (0,1)$. Then $\sigma_I (H,P) = \sigma_I (\tilde{H}, \tilde{P})$.
\end{corollary}
\begin{proof}
Recall the 
definition \eqref{eq:tauz} of $\tau$, where $\sigma$ is the symbol of $H$ (that is, $\Op (\sigma) = H$).
It follows that the bulk symbol $\tau_+$ is the same for $H$ and $\tilde{H}$, 
and thus so is the quantity $I_+$ appearing in \eqref{eq:bic}. The same holds for $\tau_-$ and $I_-$. 
We apply Theorem \ref{thm:bic} and the proof is complete.
\end{proof}

Observe that Corollary \ref{thm:main} reduces the junction conductivity to the simpler setting of a flat interface separating two insulators.
The insulating materials are described by $\tau_\pm = \tau (\pm \infty, \xi, \zeta)$. For the tBLG Hamiltonian \eqref{eq:TBGH} and spatial filter $P$ in the above example, we can now apply Theorem \ref{thm:tblg-bulk} and use the fact that
$$m(f(x_0 (y), y_0 (y))) =: m(y) = m \in \fs (-1,1)$$
to obtain
\begin{align*}
    2\pi \sigma_I (H,P) = -2\eta \, \sign(\Omega).
\end{align*}}
If $P_j  = \chi_p (g_j (x,y))$ for $j\in \{0,1,2\}$ with the $g_j$ appropriately chosen to satisfy the growth condition for $\aver{f, g}$, then 
$
    \sigma_I (H, \sum_j P_j) = \sum_j \sigma_I (H, P_j).
$
Thus if $g_j (x,y)  = \chi (r) r g_{j_\Theta} (\theta)$ with $g_{j_\Theta}= \cos (\theta - 2\pi j/3) + \cos (\pi/6)$, then the $2\pi/3-$rotational symmetry of the tBLG Hamiltonian implies that
$
    \sigma_I (H, \sum_j P_j) = 3 \sigma_I (H, P_0).
$
\sq{That is, each term $\sigma_I (H, P_j)$ in the sum over $j$ sees the same transition from $\sigma_-$ to $\sigma_+$.}
In words, the superposition of three appropriately chosen regularized indicator functions increases the conductivity by a factor of $3$. This makes perfect sense, as the conductivity is now measured through three distinct regions ($\supp \nabla P_j$ for $j=0,1,2$), each of which contributes the same value of $\sq{-}\eta \, \sign(\Omega)/\pi$.

\subsection{General setting}\label{subsec:general}

\sq{In this subsection, we provide a full class of operators $H$ and $P$ to which Theorem \ref{thm:bic} applies.}
Let us introduce some notation.
As above, the Hamiltonians $H$ act on $\mathcal{H} := L^2 (\mathbb{R}^2) \otimes \mathbb{C}^n$.
We label the spatial coordinates $(x,y) \in \mathbb{R}^2$ and the corresponding dual variables $(\xi, \zeta) \in \mathbb{R}^2$.
Given 
\sq{functions $u \in \mathcal{C}^\infty (\mathbb{R}^2)$ and $A \in \mathcal{C}^\infty (\mathbb{R}^2; \mathbb{C}^{n\times n})$,}
fixed matrices $A_1$ and $A_2$ \sq{in $\mathbb{C}^{n\times n}$},
and fixed constants $c_1 \le c_2$,
we write $A \in \fs (A_1, A_2; c_1, c_2;u)$ to mean that 
\begin{align}\label{eq:fsudef}
    A = 
    \begin{cases}
    A_1, & u < c_1\\
    A_2, & u > c_2
    \end{cases}.
\end{align}
We let $\fs (A_1, A_2; u)$ denote the union of $\fs (A_1, A_2; c_1, c_2;u)$ over all $-\infty<c_1\le c_2<\infty$.
If $A = A(\alpha)$ is a function of just one variable, then we define $\fs (A_1, A_2; c_1, c_2) := \fs (A_1, A_2; c_1, c_2; \alpha)$ and $\fs (A_1, A_2) := \fs (A_1, A_2; \alpha)$.
Note that if $A(x,y) = \chi_A (u(x,y))$ for some $\chi_A \in \fs (A_1, A_2)$, then $A\in \fs (A_1, A_2; u)$.

Let $(r, \theta)$ denote the polar coordinates for $(x,y)$.
Fix $k\in \mathbb{N}$ and let 
$\Theta_k \subset [0, 2\pi)$ be any set containing $2k$ elements.
Let $f_{\Theta}\in \mathcal{C}^\infty (\mathbb{T})$ such that $f_\Theta(\theta) = 0$ if and only if $\theta \in \Theta_k$, and $f'_\Theta (\theta) \ne 0$
whenever $\theta \in \Theta_k$. 
This means $f_\Theta$ changes sign across every point in $\Theta_k$ and
is bounded away from $0$ for all $\theta$ away from $\Theta_k$.
Let $f \in \mathcal{C}^\infty (\mathbb{R}^2)$ such that
$f(x,y) = f(r,\theta) := r f_\Theta (\theta)$ whenever $r \ge 1$.
We assume the following:\\\\
\textbf{(H1)}
Let $H = \Op (\sigma)$ be a symmetric pseudo-differential operator, with $\sigma$ independent of $h$.
Suppose $\sigma \in S^m_{1,0}$ for some \sq{$m>0$,} 
and $|\sigma_{\min} (x, y, \xi, \zeta)| \ge c \aver{\xi, \zeta}^m - 1$ for some $c > 0$. 
Suppose there exist values $-\infty < E_1 < E_2 < \infty$
and
symbols
$\sigma_{\pm} \in S^m_{1,0}$ independent of $(x,y)$ with no spectrum in $(E_1, E_2)$, 
such that
$\sigma (\cdot, \cdot, \xi, \zeta) \in \fs (\sigma_- (\xi,\zeta), \sigma_+ (\xi, \zeta); f)$ for all $(\xi,\zeta) \in \mathbb{R}^2$.\\
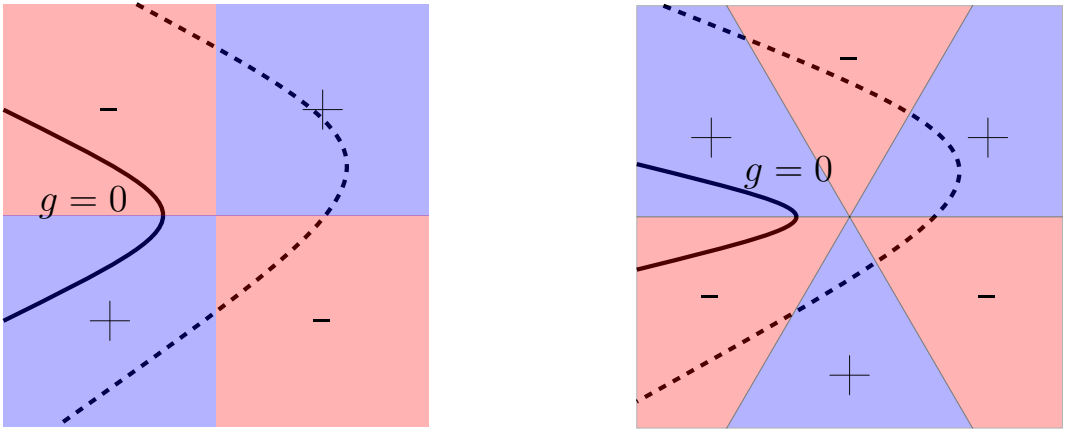
\begin{figure}
    \begin{subfigure}{.5\textwidth}
        \centering
        \begin{tikzpicture}[scale=.7]
        \draw [ultra thick] (-4,2) .. controls (0,0) and (0,0) .. (-4,-2);
        \draw [dashed, ultra thick] (-1.5,4) .. controls (4,1) and (4,1) .. (-3,-4);
        
        \fill[red,opacity=.3] (-4,0) rectangle (0,4);
        \fill[red,opacity=.3] (0,0) rectangle (4,-4);
        \fill[blue,opacity=.3] (-4,0) rectangle (0,-4);
        \fill[blue,opacity=.3] (0,0) rectangle (4,4);
        
        \node at (2,2) {\Huge +};
        \node at (-2,2) {\Huge -};
        \node at (-2,-2) {\Huge +};
        \node at (2,-2) {\Huge -};
        
        \node at (-2.5,0.25) {\Large $g = 0$};
    
        \end{tikzpicture}
    \end{subfigure}
    \begin{subfigure}{.5\textwidth}
        \centering
        \begin{tikzpicture}[scale=.7]
        \draw [ultra thick] (-4,1) .. controls (0,0) and (0,0) .. (-4,-1);
        \draw [dashed, ultra thick] (-3.5,4) .. controls (4,1) and (4,1) .. (-4,-3.5);
        
        \draw[fill = blue, opacity = 0.3] (-4,0) -- (0,0) -- (-2.31,4) -- (-4,4) --cycle;
        \draw[fill = red, opacity = 0.3] (-2.31,4) -- (0,0) -- (2.31,4) --cycle;
        \draw[fill = blue, opacity = 0.3] (2.31,4) -- (0,0) -- (4,0) -- (4,4) --cycle;
        \draw[fill = red, opacity = 0.3] (0,0) -- (4,0) -- (4,-4) -- (2.31,-4) --cycle;
        \draw[fill = blue, opacity = 0.3] (-2.31,-4) -- (0,0) -- (2.31,-4) --cycle;
        \draw[fill = red, opacity = 0.3] (-4,0) -- (0,0) -- (-2.31,-4) -- (-4,-4) --cycle;

        \node at (0,3) {\Huge -};
        \node at (2.6,1.5) {\Huge +};
        \node at (2.6,-1.5) {\Huge -};
        \node at (0,-3) {\Huge +};
        \node at (-2.6,-1.5) {\Huge -};
        \node at (-2.6,1.5) {\Huge +};
        
        \node at (-1.15,0.85) {\Large $g = 0$};
        \end{tikzpicture}
    \end{subfigure}
\caption{Illustrations of $\sigma$ in the $xy$-plane for $k = 2$ (left) and $k=3$ (right). The labels $\pm$ indicate the sign of $f$. Two example level curves $\{g = 0\}$ are presented in each case (solid and dashed curves), for some arbitrary choices of $g$ satisfying the appropriate assumptions. The scale of this picture is much smaller than $1$, so that we are not in the region where $g(r,\theta) = r g_\Theta (\theta)$.} 
\label{fig:H1}
\end{figure}

\sq{Above, $\sigma_{\min}$ denotes the smallest magnitude eigenvalue of $\sigma$.}
By \cite{Bony2013,Hormander}, we know (H1) implies that $H : \mathcal{D} (H) \rightarrow \mathcal{H}$ is self-adjoint with $\mathcal{D} (H) = \mathcal{H}^m$ (and the right-hand side the standard Hilbert space of functions with square-integrable derivatives up to order $m$).
\sq{Moreover, for any $z \in \mathbb{C}$ with $\Im z \ne 0$, we have $(z-H)^{-1} \in \Op (S^{-m}_{1,0})$.}

Let 
$\{\theta_1, \theta_2\} \subset \mathbb{T}$ such that
$\theta_1, \theta_2 \notin \Theta_k$.
Let $g_{\Theta}\in \mathcal{C}^\infty (\mathbb{T})$ such that $g_\Theta(\theta) = 0$ if and only if $\theta \in \{\theta_1, \theta_2\}$, and $g'_\Theta (\theta_j) \ne 0$ for $j \in \{1,2\}$.
Let $g \in \mathcal{C}^\infty (\mathbb{R}^2)$ such that
$g(x,y) = g(r,\theta) = r g_\Theta (\theta)$ for all $r \ge 1$. 
\sq{Our assumptions on $f$ and $g$ imply the existence of positive constants $C_1 < C_2$ such that
\begin{align}\label{eq:fg}
    C_1 \aver{x,y} \le \aver{f(x,y), g(x,y)} \le C_2 \aver{x,y}.
\end{align}
Moreover,}
there exists a smooth function $(x_0, y_0):\mathbb{R} \rightarrow \mathbb{R}^2$ with $\lim_{|t| \rightarrow \infty} |(x_0 (t), y_0 (t))| = \infty$ and whose range separates the $xy-$plane into two regions, $R_{+}$ and $R$, such that $g \ge 0$ in $R_{+}$. 
One can for example take a curve contained in $\{g = 0\}$, in which case
$(x_0 (t), y_0 (t)) = (\beta_1 t, \beta_2 t)$ for all $|t|$ sufficiently large,
with the constants $\beta_1$ and $\beta_2$ depending only on $\sgn t$.
Our convention is that $R_+$ be on the \emph{right} side of the curve defined by $(x_0, y_0)$ as $t$ increases. That is, if you were to move along the curve with $t$ increasing, then $R_+$ would be on your right.
\sq{Thus $\tau$ in Theorem \ref{thm:bic} can always be defined by \eqref{eq:tauz} in this more general setting.}

\sq{
We conclude this subsection with the following result, which verifies that our theory applies to the tBLG Hamiltonian.
\begin{proposition}\label{prop:H1tBLG}
    Define $H=\Op (\sigma)$ by \eqref{eq:TBGH}. Then $H$ satisfies (H1).
\end{proposition}
We refer to Appendix \ref{sec:pf} for the proof.}

\subsection{Stability of the junction conductivity}

\sq{In this subsection, we show that the junction conductivity $\sigma_I (H,P)$ is invariant with respect to a large class of perturbations of the Hamiltonian $H$, spatial filter $P$ and density of states $\varphi'$. While interesting in their own right, these results are also used in the proof of Theorem \ref{thm:bic} (see Appendix \ref{sec:pf}). Throughout, we will assume that $H = \Op (\sigma)$ satisfies (H1), and
let $\varphi \in \fs (0,1;E_1,E_2)$ and $P \in \fs (0,1;g(x,y))$ be smooth switch functions. 
Let $\projP \in \modfs (0,1;g(x,y))$ be a projector ($\projP^2 = \projP$). 
Here, $\modfs(A_1,A_2;u)$ is the space of functions $A$ satisfying \eqref{eq:fsudef} for some $c_1 \le c_2$. (Note that $A \in \modfs(A_1,A_2;u)$ need not be continuous.)
Let $\Phi \in \mathcal{C}^\infty_c (E_1, E_2)$ such that $\varphi ' \in \mathcal{C}^\infty_c (\{ \Phi = 1 \}^\circ)$, where $U^\circ$ denotes the interior of $U$. We will tie the junction conductivity to a Fredholm index, which will be used to prove that the former is stable.

\begin{theorem}\label{thm:idx}
    Let $\genH$ be a self-adjoint pseudo-differential operator with $\genH \in \Op (S(\aver{\xi,\zeta}^m))$ and $\Phi (\genH) \in \Op (S (\aver{f(x,y),\xi,\zeta}^{-\infty}))$. 
    Then $[\genH, P] \varphi ' (\genH)$ is trace-class so that $\sigma_I (\genH, P)$ is well-defined. Moreover, defining $U (\genH) := e^{i2\pi \varphi (\genH)}$, we have that $\projP U(\genH) \projP$ is a Fredholm operator on the range of $\projP$, with 
    \begin{align*}
        2\pi \sigma_I (\genH, P) = {\rm Index} (\projP U(\genH) \projP).
    \end{align*}
\end{theorem}
Note that (H1) is a strict subset of the class of operators $\genH$ to which the above result applies.
(That any $H$ satisfying (H1) must also satisfy the assumptions of Theorem \ref{thm:idx} will be proved rigorously in Appendix \ref{sec:pf}; see Proposition \ref{trclass}.)
For a proof of Theorem \ref{thm:idx}, see Appendix \ref{sec:pf}.

\medskip

Theorem \ref{thm:idx} is a powerful result with many implications. 
For example, we immediately have 
\begin{corollary}\label{cor:invP}
If $\genH$ satisfies the assumptions of Theorem \ref{thm:idx} and
$P_1, P_2 \in \fs (0,1; g(x,y))$, then $\sigma_I (\genH,P_1)=\sigma_I (\genH, P_2)$.
\end{corollary}
\begin{proof}
    Applying Theorem \ref{thm:idx}, we have $2\pi \sigma_I (\genH,P_1) = {\rm Index} (\projP U(\genH) \projP) = 2\pi \sigma_I (\genH,P_2)$.
\end{proof}
For an intuitive explanation of the above corollary,} 
recall that our model is of
two types of materials ($+$ and $-$) that are smoothly glued together at a junction. See Figure \ref{fig:H1} for two examples, where we
assume for concreteness that $P$ transitions from $0$ to $1$ in the vicinity of $\{g=0\}$.
\sq{The solid curve (in both panels) can see only one transition ($\sigma_+$ to $\sigma_-$ and $\sigma_-$ to $\sigma_+$ in the left and right panels, respectively), while the dashed level curves contain multiple transitions ($+ \rightarrow - \rightarrow + \rightarrow -$ for the left panel and $- \rightarrow + \rightarrow - \rightarrow + \rightarrow - \rightarrow +$ for the right panel).}
\sq{We refer also to Figure \ref{fig:fg} (right panel) for an illustration of $f$ evaluated along various level curves of $g$ when $k=3$.}
As expected, the conductivity only cares about the ``starting'' and ``ending'' topology along the level curve and is unaffected by oscillations in between. 
\sq{Observe that Corollary \ref{cor:invP} describes a conservation law, as the solid and dashed curves respectively measure the conductivity entering and leaving the junction.

We will use Theorem \ref{thm:idx} to prove invariance of $\sigma_I (\genH, P)$ with respect to $\genH$. Since a Fredholm index must be integer-valued, we know that any continuous perturbation 
of $\genH$ would leave the junction conductivity unchanged; see Theorems \ref{thm:h}-\ref{thm:compact} below. 
First, we show that the conductivity is independent of the density of states we choose.


\begin{theorem}\label{thm:invvarphi}
    Take $\genH$ as in Theorem \ref{thm:idx} and let
    $\varphi_1 \in \fs (0,1)$ such that $\varphi'_1 \in \mathcal{C}^\infty_c (\{\Phi = 1\}^\circ)$. Then
    \begin{align*}
        \sigma_I(\genH,P) = \Tr i [\genH,P] \varphi '_1 (\genH).
    \end{align*}
\end{theorem}

We now show that the conductivity is independent of semiclassical rescaling. To this end, define $H_h := \Op _h (\sigma)$ for $0<h\le 1$, where $H=\Op (\sigma)$ is assumed to satisfy (H1). This allows for an expansion of $\sigma_I (H_h, P)$ in the semiclassical parameter $h$, which 
will be used to prove Theorem \ref{thm:main} (see Appendix \ref{sec:pf}). 
We refer to \cite{B-higher-dimensional-2021,bal3,QB}, where a similar semiclassical analysis is employed.
\begin{theorem}\label{thm:h}
    For any $h \in (0,1]$ 
    we have $\sigma_I (H_h, P) = \sigma_I (H,P)$. 
\end{theorem}

Let $\pert$ be a symmetric pseudo-differential operator 
and 
define $\Hmu := H+\mu \pert$ for $\mu \in [0,1]$.
\begin{theorem}\label{thm:bdd}
    If 
    $\pert \in \Op (S^m_{1,0})$, then 
    $\sigma_I (\Hmu, P) = \sigma_I (H,P)$ for all $\mu >0$ sufficiently small.
\end{theorem}
Theorem \ref{thm:bdd} asserts the stability of the junction conductivity with respect to relatively bounded perturbations that are small enough. To avoid this smallness condition, 
we must make the stronger assumption that the perturbation be relatively compact; that is, we require the symbol of $\pert (i\pm H)^{-1}$ to decay at infinity in all variables.
\begin{theorem}\label{thm:compact}
    If 
    $W \in \Op (S^m_{1,0} \cap S (\aver{\xi, \zeta}^{m-\delta} \aver{x,y}^{-\delta}))$ for some $\delta > 0$, then $\sigma_I (\Hmu, P) = \sigma_I (H,P)$ for all $\mu \in [0,1]$.
\end{theorem}
We refer to Appendix \ref{sec:pf} for the proofs of Theorems \ref{thm:invvarphi}-\ref{thm:compact}.
Recall that $m=1$ for the tBLG Hamiltonian \eqref{eq:TBGH}.
}

\subsection{Dirac model}
We conclude this section with an application of \sq{Theorem \ref{thm:bic}} to the $2 \times 2$ Dirac Hamiltonian given by
\begin{equation}\label{e:hm}
    H = D_x \sigma_1 + D_y \sigma_2 + \tilde{m}(x,y) \sigma_3,
\end{equation}
where $\tilde{m} (x,y) = m(f(x,y))$ for some \sq{$m \in \fs (-1,1)$}, and
\begin{align*}
    f(r, \theta) = \chi (r) r f_\Theta (\theta), \qquad \chi \in \fs (0,1;\eps,1), \qquad f_\Theta (\theta) = \sin (k \theta)
\end{align*}
for some $0 < \eps < 1$ and $k \in \mathbb{N}_+$.
Note that when $k = 1$, we recover the setting of a flat interface analyzed in detail in \cite{B-bulk-interface-2019,B-EdgeStates-2018,B-higher-dimensional-2021,bal3,QB}, while $k=3$ yields the same hexagonal structure analyzed at the beginning of this section. \sq{As with Proposition \ref{prop:H1tBLG}, it is straightforward to verify that $H$ defined by \eqref{e:hm} satisfies (H1), and thus the above theory applies. The bulk spectral gap in this case is $(E_1,E_2)=(-1,1)$.}

Take $P (x,y) = \chi_p (g(x,y))$ for $\chi_p \in \fs (0,1)$, where $g(x,y) = g(r,\theta) = \chi (r) r g_\Theta (\theta)$ and $g_\Theta (\theta) = \cos (\theta - \theta_1) - \cos (\theta_0)$, for some $0 < \theta_0 < \pi$ and $-\frac{\pi}{k} < \theta_1 < \frac{\pi}{k}$ satisfying
$$\theta_+ \notin \Theta_k, \quad 2\pi - \theta_+ \notin \Theta_k, \quad \Theta_k = \Big \{\frac{j}{k \pi}: j \in \mathbb{Z}\Big\}, \quad \theta_+ := \theta_0 + \theta_1.$$
This way, the zeros of $g_\Theta$ and $f_\Theta$ are disjoint.
Indeed, $g_\Theta (\theta) = 0$ if and only if $\cos (\theta-\theta_1)= \cos (\theta_0)$, which occurs exactly when $\{\theta + 2\pi j : j \in \mathbb{Z}\} \cap \{\theta_+, 2\pi - \theta_+\} \ne \emptyset$.
It follows that 
$\tau (t,\xi, \zeta) = \xi \sigma_1 + \zeta \sigma_2 + \mu (t) \sigma_3$ 
for some
$\mu \in \cup_{\eps_1, \eps_2 \in \{-1,1\}} \fs (\eps_1, \eps_2)$, where the $\eps_j$ are determined by $\theta_0$ and $\theta_1$. 
We have reduced the problem to computing the interface conductivity for the translation-invariant $2 \times 2$ Dirac system \cite{B-bulk-interface-2019,bal3,QB}, and hence $2 \pi \sigma_I = \frac{1}{2} (\eps_1 - \eps_2)$.
For the case $(k,\theta_1)=(3,0)$ which is analyzed numerically in the following section, we have
\begin{align*}
    2\pi \sigma_I = \begin{cases}
    -1, & 0<\theta_0 < \pi/3\\
    1, & \pi/3 < \theta_0 < 2\pi/3\\
    -1, & 2\pi/3 < \theta_0 <\pi
    \end{cases}.
\end{align*}
If $g$ from Figure \ref{fig:H1} were to have the above form, then $\theta_1 = 0$ and $\theta_0$ would be the angle between [the line making up the top of the level curve $\{g = 0\}$] and [the positive $x-$axis].
So $\pi/2 < \theta_0 < \pi$ and $2\pi/3 < \theta_0 < \pi$ for the left and right panels, respectively.

\section{Numerical Simulations} \label{sec:numerics}


We now present results of numerical simulations illustrating the stability of the interface and junction conductivities analyzed in the preceding sections. We augment these spectral results with simulations of wavepackets propagating across the junction considered in section~\ref{sec:junctions} (case $k=3$ in Fig.~\ref{fig:H1}).


\subsection{Methods and Discretization}
All computations are conducted on a finite interval of size $L$ for one-dimensional problems such as~\eqref{e:h2} (see section~\ref{subsec:numericalvalleys}) or a rectangular domain of size $L_x \times L_y$ for two-dimensional problems such as conductivities for the Hamiltonian~\eqref{e:h4} or~\eqref{e:hm}, equipped with periodic boundary conditions. We use a pseudo-spectral discretization where all functions are approximated using truncated Fourier series:
\[
    u(x) = \sum_{k = -K}^{K} \widehat{u}_{k} e^{2i\pi k x / L} \quad \text{or} \quad u(x,y) = \sum_{k = -K_x}^{K_x}\sum_{l = -K_y}^{K_y} \widehat{u}_{k,l} e^{2i\pi k x / L_x + 2i\pi l y / L_y },
\] 
such that derivation operators are represented as diagonal matrices. Pointwise multiplication operators such as $Vu(x) = v(x) u(x)$ are computed using Fourier interpolation on a uniform real-space grid of size $3(K+1)$ in the one-dimensional case and $3(K_x+1) \times 3(K_y+1)$ in the two-dimensional one employing the discrete Fourier transform $\mathcal{F}$, avoiding the aliasing of products for small values of $K$ and sharp or highly oscillatory mass or potential profiles $m(x,y), v(x,y)$. 

Computations of conductivities using finite domains are accomplished using the approximation
\begin{equation}\label{eq:numconductivity}
    2\pi \widetilde{\sigma}(H) = 2\pi \Tr i Q [H,P] \phi'(H),
\end{equation}
where $P = p(x,y)$ and $Q = q(x,y)$ are pointwise multiplication operators where, as before, $p(x,y)$ is a spatial smooth switch function across a certain contour, and $q(x,y)$ is a spatial filter designed to mask the boundaries which may host artificial, unwanted domain walls (necessary to have space-dependent coefficients such as $m(x,y)$ remain smoothly connected across the periodic boundaries of the torus). More precisely, we set here:
\begin{equation}\label{eq:numfilters}
    \begin{aligned}
    p_{x_0,\delta}(x,y) &= \begin{cases}
        1, & \max \left \{ x - x_0,  \quad -\cos(\theta) (y - L_y/2) \pm \sin(\theta) (x - x_0 - L_x/4) \right \} \geq \delta, \\
        0, & \max \left \{ x - x_0,  \quad -\cos(\theta) (y - L_y/2) \pm \sin(\theta) (x - x_0 - L_x/4) \right \} \leq -\delta,
    \end{cases}\\
    &\text{ and } \\
    q_\delta(x,y) &= \begin{cases}
        1, & \max \left \{ \vert x -L_x/2 \vert - L_x/4, \quad \vert y - L_y/2 \vert - 3L_y/8 \right \} \leq - \delta , \\
        0, & \max \left \{ \vert x -L_x/2 \vert - L_x/4, \quad \vert y - L_y/2 \vert - 3L_y/8 \right \} \geq \delta,
    \end{cases}
    \end{aligned}
\end{equation}
such that $\delta$ is a small parameter controlling the width of transition regions, $\theta = \pi - \frac{\pi}{12}$ is as in Section~\ref{sec:junctions} and $x_0$ is the abscissa at which the support of $\nabla p_{x_0, \delta}$ crosses the horizontal axis: see Figure~\ref{fig:mass_switch_profile} for such a profile.
Valley-projected numerical conductivities $\widetilde{\sigma}_\pm$ are defined similarly~\eqref{eq:valleyconductivity}. 

\begin{figure}[ht!]\centering
    \includegraphics[width=.4\textwidth]{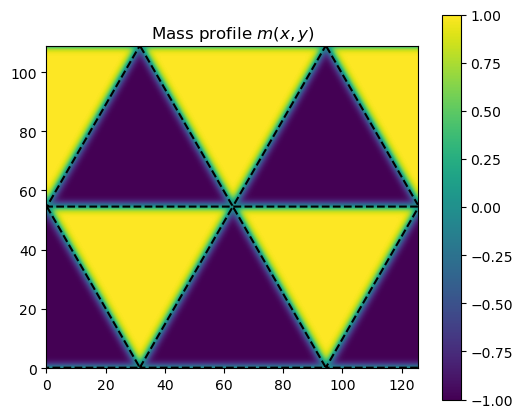}
    \hspace{.5in}
    \includegraphics[width=.4\textwidth]{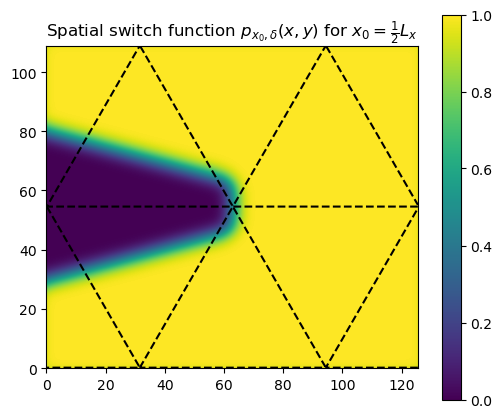}\\
    \includegraphics[width=.32\textwidth]{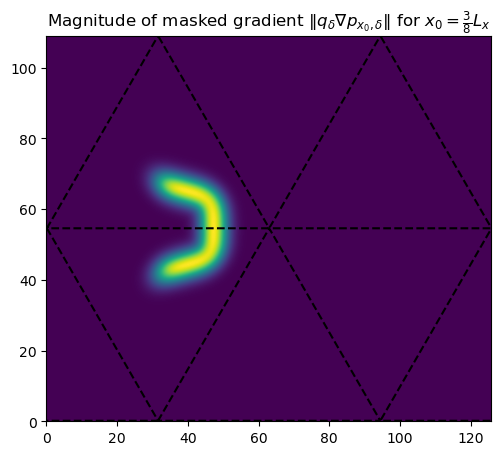}
    \includegraphics[width=.32\textwidth]{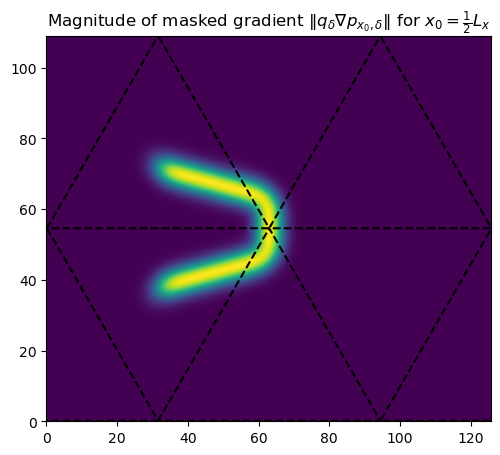}
    \includegraphics[width=.32\textwidth]{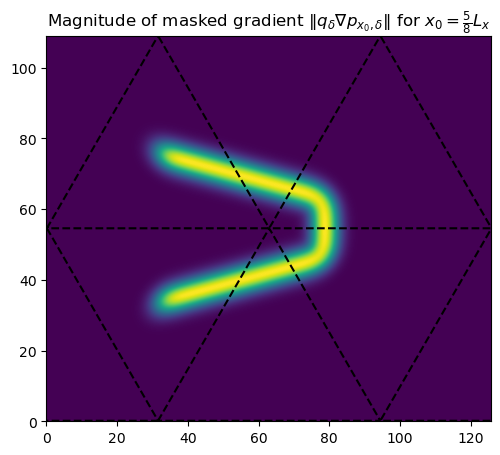}
\caption{Top left: mass profile $m(x,y)$ yielding a three-fold symmetric topological junction. The contour $m(x,y) = 0$ is highlighted in dashed black line. Top right: spatial switch function $p_{x_0, \delta}(x,y)$ for $x_0 = L_x/2$ used in the computation of the conductivity~\eqref{eq:numconductivity}. Bottom row: magnitude profiles of the gradient of the spatial switch function $\nabla p_{x_0, \delta}(x,y)$ multiplied by the filter $q_0(x,y)$~\eqref{eq:numfilters} for the three values (from left to right) $x_0 = 3L_x/8,\ L_x/2,\ 5L_x/8$ presented in Table~\ref{tab:junction_invariants_num}.}
\label{fig:mass_switch_profile}
\end{figure}

Practically, the evaluation of~\eqref{eq:numconductivity} is achieved by direct diagonalization of the discretized Hamiltonian $H$, resulting in the sum-over-states formula:
\[
    2\pi \widetilde{\sigma}(H) = 2\pi \sum_{j} \phi'(\lambda_j) \left \langle u_j  \left \vert Q_\delta \begin{pmatrix}0 & \partial_x p_{x_0, \delta} - i \partial_y p_{x_0, \delta} \\ \partial_x p_{x_0, \delta} + i \partial_y p_{x_0, \delta} & 0 \end{pmatrix}  \right \vert u_j \right \rangle,
\]
where $\lambda_j$, $u_j$ are the eigenvalues and eigenvectors of $H$ and the scalar product is evaluated by quadrature on the uniform real-spaced grid discussed earlier.

The spectral accuracy of the numerical conductivity~\eqref{eq:numconductivity} using our spectral discretization scheme can be established rigorously by observing that the relevant edge states vanish exponentially fast in the bulk. We do not pursue this analysis here and refer to \cite{QB} for the case of a single straight domain wall for elliptic partial differential operators, including the massive Dirac operator.

\subsection{Invariant Calculations for Junctions}
As a first numerical experiment illustrating the results of Section~\ref{sec:junctions}, we compute the conductivity in a hexagonal junction setup as in region 2 from Figure~\ref{fig:tblg}, i.e. a vertex with three-fold symmetry with the topology as in Figure~\ref{fig:H1}. Specifically, our choice of spatial projector~\eqref{eq:numfilters} with $\theta = \pi - \frac{\pi}{12}$ corresponds to a function $g$ whose level set $g^{-1}(0)$ fully encloses only the left horizontal branch of the boundary set $m^{-1}(0)$, i.e. the negative real axis. 

\begin{table}[b!]
    \centering
    \begin{tabular}{|c||c|c|c|c|} \hline
                        & $N=8$     & $N=16$    & $N=32$    & $N=64$    \\  \hline
      $x_0=3L_x/8$      & $0.36144$ & $0.92310$ & $0.99923$ & $0.99983$ \\  \hline
      $x_0 = L_x/2$       & $0.53574$ & $0.72033$ & $0.99612$ & $0.99997$ \\  \hline
      $x_0= 5L_x/8$     & $0.05648$ & $0.88647$ & $0.99901$ & $0.99993$ \\  \hline
    \end{tabular}
    \caption{Conductivity values~\eqref{eq:numconductivity} computed for different values of the spatial switch function center $x_0$ along the $x$-axis, and discretization parameter $N = K_x = K_y$.}
    \label{tab:junction_invariants_num} 
\end{table}

Results from Section~\ref{sec:junctions} predict that the quantity $2\pi\sigma_I$ takes the value $1$ independently of the exact choice of such a spatial projector, in particular whether $x_0 < L_x/2$ and the level set $g^{-1}(0)$ intersects the negative real axis only, or $x_0 > L_x/2$ and it crosses the other five boundaries. Our numerical computations, presented in Table~\ref{tab:junction_invariants_num}, are consistent with the theory, with the converged results ($N=64$) agreeing up to four digits of accuracy with the predicted value.

Table~\ref{tab:junction_invariants_num} also illustrates the stability and efficiency of our pseudo-spectral discretization approach. For all three choices of parameter $x_0$, we note the exponential convergence of the computed conductivity with the frequency cutoff parameter $N$, consistent with the robustness of the conductivity demonstrated in earlier sections. Numerical experiments (not presented) also indicate fast exponential convergence in the size of the domain $L_x \times L_y$, consistent with the exponential localization of eigenmodes in the energy window of interest around boundaries of interest as well as artificial domain walls induced by periodization.

\subsection{Wavepacket Propagation through Junctions}
\begin{figure}[b!]
    \centering
    \includegraphics[width=.32\textwidth]{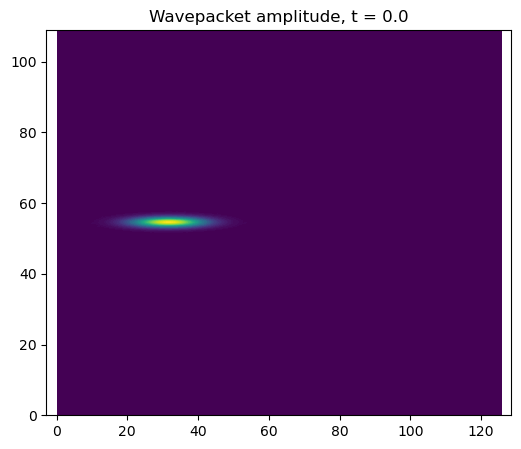}
    \includegraphics[width=.32\textwidth]{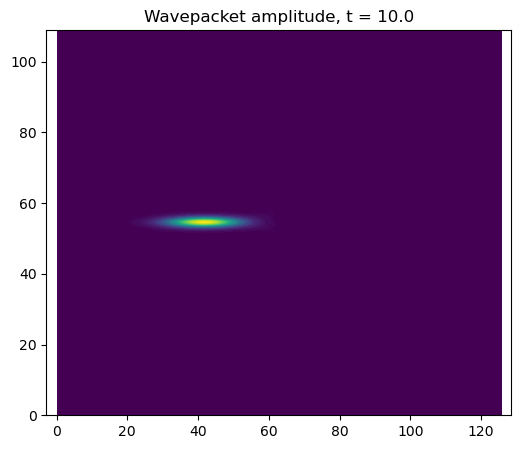}
    \includegraphics[width=.32\textwidth]{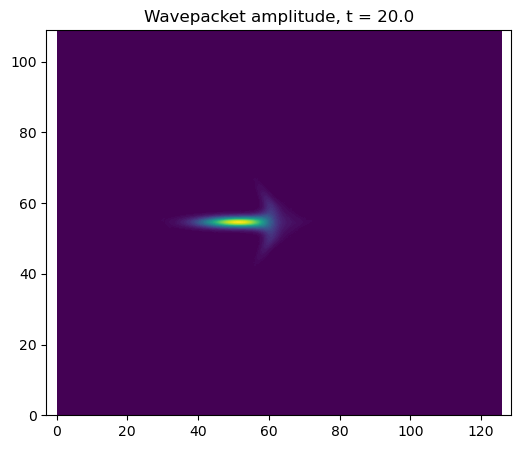} \\
    \includegraphics[width=.32\textwidth]{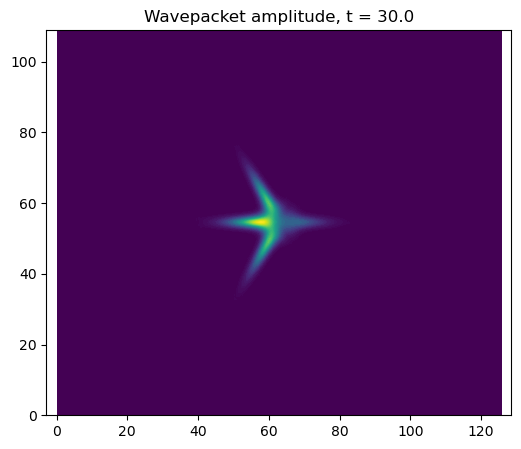}
    \includegraphics[width=.32\textwidth]{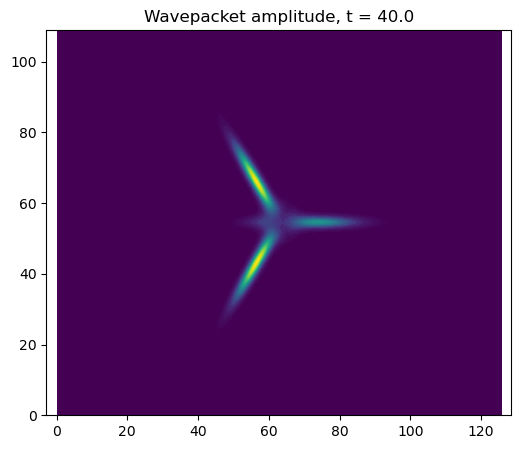}
    \includegraphics[width=.32\textwidth]{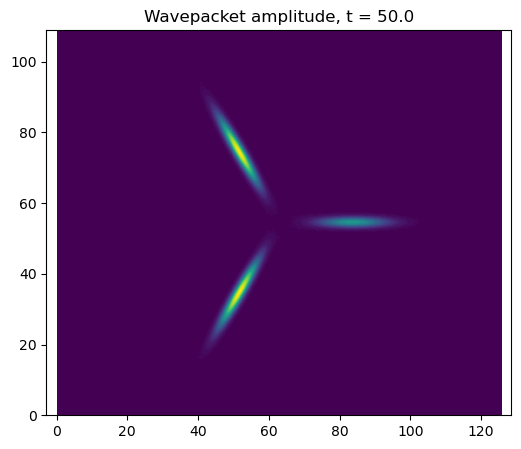}
    \caption{Scattering of a wavepacket incoming from the left along a horizontal branch towards a hexagonal cross, with a rotationally symmetric mass profile. }
    \label{fig:propagation_clean}
\end{figure}

Next, we study numerically the time-domain propagation of wavepackets through a hexagonally symmetric junction given by the same mass profile as plotted in Figure~\ref{fig:mass_switch_profile}. Let us note that the stable topological invariants presented in Section~\ref{sec:junctions}, which can be interpreted as counting the net number of modes or currents traveling in certain directions along the interfaces~\cite{prodan2016bulk,B-EdgeStates-2018,B-bulk-interface-2019}, only partly predict the dynamics of localized wavepackets traveling across a junction. 
In particular, the sign of the conductivity identifies certain branches of the interface $m^{-1}(0)$ as incoming channels on which net propagation of current is towards the junction, while on others current is outgoing, directed outwards from the junction (with the sum of current on all channels always being zero).
On the other hand, the steady-state current generated by the state $\phi'(H)$ is equally distributed between all channels, and does not allow to distinguish how the current 'originating' from one of the incoming channels distributes itself.

Semiclassical analysis, in the case of non-intersecting, straight or curved interfaces, shows that for Dirac operators such as the one considered here~\eqref{e:hm}, these currents have a local interpretation via the construction of long-lived localized wavepackets travelling unidirectionally~\cite{bal2021edge}, with some dispersion due to the curvature of the interface.
Hence, wavepackets travelling along an interface towards the junction are expected to scatter into any number wavepackets travelling on the outgoing channels.

While such rigorous analysis has not been carried out in the case of junctions, our numerical simulations of such scattering across a hexagonal junction show that a Gaussian wavepacket, initially travelling along the incoming channel supported by the horizontal edge to the left of the junction, splits three-fold at the intersection along the three outgoing branches of the junction: see snapshots of the evolution presented in Figure~\ref{fig:propagation_clean}. 
More precisely, an initial Gaussian packet is prepared at $t=0$ on the left horizontal branch which forms an incoming channel towards the intersection, and we compute its time evolution until $t=50$ where the packet has completed its travel through the intersection.
In contrast to the numerical results presented in~\cite{bal2021edge}, we observe remarkably little dispersion of the packet despite the sharp turns in the direction of propagation. 

As mentioned above, the topology measured by the quantized junction conductivity, while heuristically linked to the robustness and stability of the propagating edge modes, does not distinguish between channels and thus does not constrain the entries of the scattering matrix between modes supported by incoming and outgoing channels.
We observe on Figure~\ref{fig:propagation_clean} that the wavepacket does not split equally between all three branches, with a smaller amount travelling straight across the intersection.
\begin{figure}[b!]
    \centering
    \includegraphics[width=.32\textwidth]{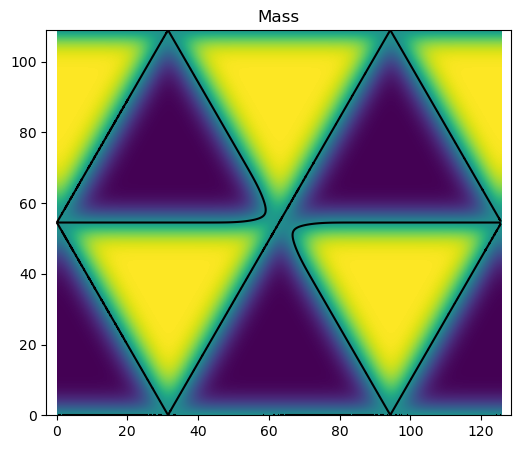}
    \includegraphics[width=.32\textwidth]{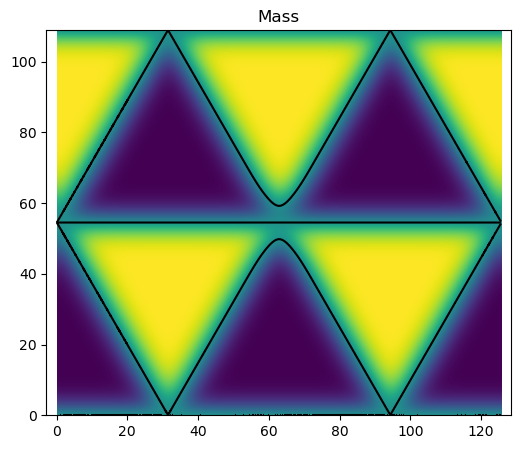}
    \includegraphics[width=.32\textwidth]{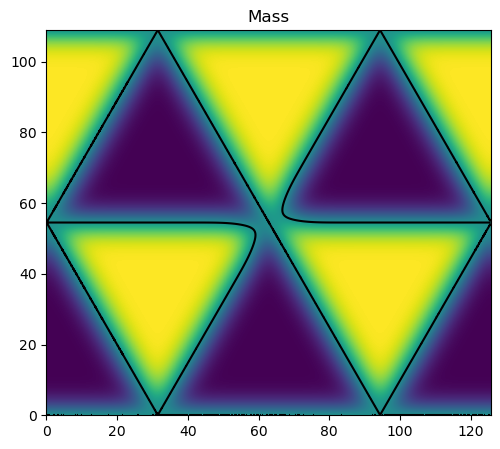} \\
    \includegraphics[width=.32\textwidth]{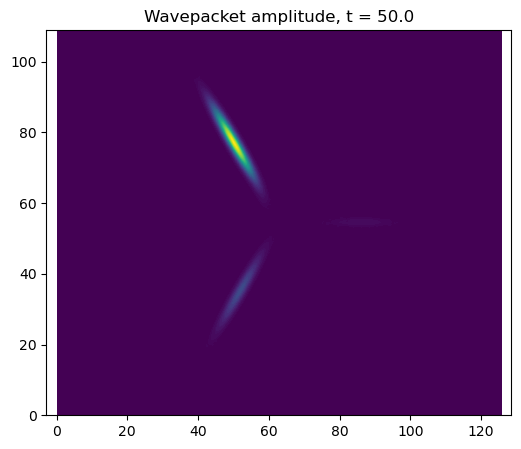} 
    \includegraphics[width=.32\textwidth]{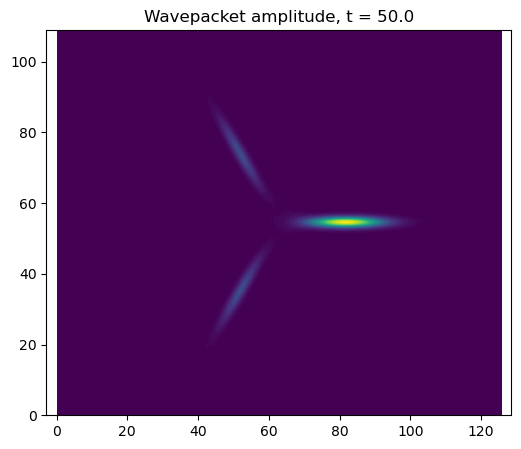} 
    \includegraphics[width=.32\textwidth]{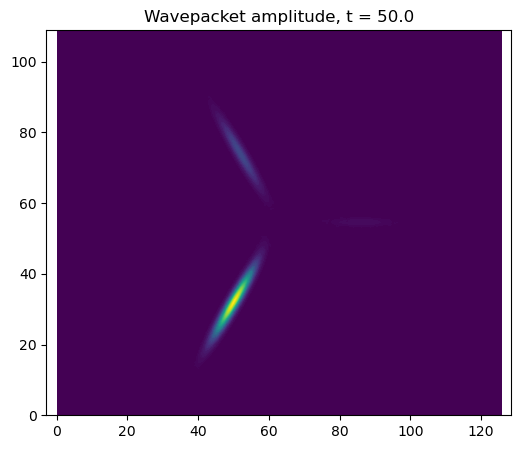}
    \caption{Scattering of a wavepacket through a hexagonal cross with broken symmetry profiles, allowing to steer a wavepacket incoming from the left towards any outgoing branch.}
    \label{fig:propagation_broken_symmetry}
\end{figure}

To illustrate the absence of stability of the scattering with respect to perturbations, we present in Figure~\ref{fig:propagation_broken_symmetry} the final state of simulations with the same initial wavepacket as in Figure~\ref{fig:propagation_clean}, but where a small, local perturbation of the mass term breaks the rotational symmetry: namely we set
\[
    m(x,y) = m_0(x,y) - \frac{1}{4}  \exp \left ( - \frac{ (x-L_x/2)^2 + (y-L_y/2)^2 }{2\sigma^2} \right ) \begin{pmatrix} \sin(\theta_m) \\ \cos(\theta_m) \end{pmatrix} \cdot \begin{pmatrix} x-L_x/2 \\ y-L_y/2 \end{pmatrix},
\]
where $m_0(x,y)$ is the symmetric mass profile as in Figure~\ref{fig:mass_switch_profile}, $\sigma = 5$ is a width parameter and $\theta_m = 2\pi/3,\ 0,\ -2\pi/3$ respectively from left to right. These simulations indicate that the scattering can be completely manipulated by such local mass perturbations, and is not a protected property of the system: indeed, the wavepacket can be steered towards any of the three outgoing branches.
\section{Conclusion}
We consider continuum partial (and pseudo) differential models of twisted bilayer graphene (tBLG) that accurately describe macroscopic transport properties in tBLG. The twisted layers generate a macroscopic triangular moir\'e pattern separating regions of AB and BA stacking as shown in Fig. \ref{fig:tblg}. 

The primary focus of this work is to analyze the asymmetric transport observed at the interface between insulating AB and BA stackings (region 1 in Fig. \ref{fig:tblg}) and at the junctions of the triangular pattern (region 2 in Fig. \ref{fig:tblg}). 

Asymmetry is quantified by observables of the system referred to as conductivities. A familiar {\em line conductivity} analyzed for such pseudo-differential equations in \cite{B-bulk-interface-2019,B-higher-dimensional-2021,bal3,QB} describes the asymmetry in region 1. We adapted the derivations in \cite{massatt2021} in the context of Floquet topological insulators to compute the line conductivity and show that it took integral values by means of a bulk-interface correspondence. 

The aforementioned asymmetry takes opposite values in each valley of the problem so that asymmetric transport is observed only when inter-valley coupling is negligible. We presented a class of simplified Hamiltonians modeling such valley coupling. We defined partial line conductivities for the two valleys and showed formally and numerically that these conductivities were no longer quantized and hence no longer stable with respect to perturbations. 

Finally, to constrain transport across junctions in region 2 of Fig. \ref{fig:tblg}, we introduced a notion of {\em junction conductivity} defined for a large class of domain walls modeling the insulating regions near the junction and for a large class of (pseudo-)differential operators including the Dirac operators modeling tBLG. Adapting techniques developed in \cite{QB}, we demonstrated that the junction conductivity was indeed quantized and obtained by a Fedosov-H\"ormander formula generalizing that for line conductivities. Numerical simulations of junction conductivities and other transport phenomena for a simple model of Dirac equation confirmed the theoretical predictions.


\section*{Acknowledgments} This research was partially supported by the National Science Foundation, Grants DMS-1908736, DMS-1819220 and EFMA-1641100.

\bibliographystyle{abbrv}
\bibliography{junction}

\begin{thebibliography}{10}

\bibitem{Peeters_2018}
M.~An\dj{}elkovi\ifmmode~\acute{c}\else \'{c}\fi{}, L.~Covaci, and F.~M.
  Peeters.
\newblock {DC} conductivity of twisted bilayer graphene: Angle-dependent
  transport properties and effects of disorder.
\newblock {\em Phys. Rev. Materials}, 2:034004, Mar 2018.

\bibitem{avron1994charge}
J.~E. Avron, R.~Seiler, and B.~Simon.
\newblock Charge deficiency, charge transport and comparison of dimensions.
\newblock {\em Communications in mathematical physics}, 159:399--422, 1994.

\bibitem{B-bulk-interface-2019}
G.~Bal.
\newblock Continuous bulk and interface description of topological insulators.
\newblock {\em J. Math. Phys.}, 60(8):081506, 2019.

\bibitem{B-EdgeStates-2018}
G.~Bal.
\newblock Topological protection of perturbed edge states.
\newblock {\em Comm. Mathematical Sciences}, 17(1):193--225, 2019.

\bibitem{bal3}
G.~Bal.
\newblock Topological invariants for interface modes.
\newblock {\em Comm. Partial Differential Equations}, 47:1636--1679, 2022.

\bibitem{B-higher-dimensional-2021}
G.~Bal.
\newblock Topological charge conservation for continuous insulators.
\newblock {\em To appear in J. Math. Phys. and arXiv:2106.08480}, 2023.

\bibitem{bal2021edge}
G.~Bal, S.~Becker, A.~Drouot, C.~F. Kammerer, J.~Lu, and A.~Watson.
\newblock Edge state dynamics along curved interfaces.
\newblock {\em To appear in SIAM J. Math. Anal. and arXiv:2106.00729}, 2021.

\bibitem{massatt2021}
G.~Bal and D.~Massatt.
\newblock Multiscale invariants of {Floquet} topological insulators.
\newblock {\em Multiscale Modeling \& Simulation}, 20(1):493--523, 2022.

\bibitem{Hughes_2013}
B.~A. Bernevig and T.~L. Hughes.
\newblock {\em Topological Insulators and Topological Superconductors}.
\newblock Princeton University Press, 2013.

\bibitem{macdonald_2011}
R.~Bistritzer and A.~H. MacDonald.
\newblock Moiré bands in twisted double-layer graphene.
\newblock {\em Proceedings of the National Academy of Sciences - PNAS},
  108(30):12233--12237, 2011.

\bibitem{Bony2013}
J.-M. Bony.
\newblock {\em On the Characterization of Pseudodifferential Operators (Old and
  New)}, pages 21--34.
\newblock Springer New York, New York, NY, 2013.

\bibitem{cances2022simple}
E.~Canc\`es, L.~Garrigue, and D.~Gontier.
\newblock Simple derivation of moir\'e-scale continuous models for twisted
  bilayer graphene.
\newblock {\em Phys. Rev. B}, 107:155403, Apr 2023.

\bibitem{pablo2018}
Y.~Cao, V.~Fatemi, S.~Fang, K.~Watanabe, T.~Taniguchi, E.~Kaxiras, and
  P.~Jarillo-Herrero.
\newblock Unconventional superconductivity in magic-angle graphene
  superlattices.
\newblock {\em Nature}, 556(7699):43--50, 2018.

\bibitem{carr_2018}
S.~Carr, D.~Massatt, S.~B. Torrisi, P.~Cazeaux, M.~Luskin, and E.~Kaxiras.
\newblock Relaxation and domain formation in incommensurate two-dimensional
  heterostructures.
\newblock {\em Phys. Rev. B}, 98:224102, Dec 2018.

\bibitem{srolovitz2016}
S.~Dai, Y.~Xiang, and D.~J. Srolovitz.
\newblock Twisted bilayer graphene: Moir{\'e} with a twist.
\newblock {\em Nano Letters}, 16(9):5923--5927, 09 2016.

\bibitem{sjostrand}
M.~Dimassi and J.~Sjostrand.
\newblock {\em Spectral Asymptotics in the Semi-Classical Limit}.
\newblock London Mathematical Society Lecture Note Series. Cambridge University
  Press, 1999.

\bibitem{Drouot:19b}
A.~Drouot.
\newblock The bulk-edge correspondence for continuous honeycomb lattices.
\newblock {\em Communication in Partial Differential Equations},
  44(12):1406–1430, 2019.

\bibitem{weinstein2011}
V.~Duchêne and M.~I. Weinstein.
\newblock Scattering, homogenization, and interface effects for oscillatory
  potentials with strong singularities.
\newblock {\em Multiscale Modeling \& Simulation}, 9(3):1017--1063, 2011.

\bibitem{elbau2002equality}
P.~Elbau and G.~Graf.
\newblock Equality of bulk and edge hall conductance revisited.
\newblock {\em Communications in mathematical physics}, 229(3):415--432, 2002.

\bibitem{fruchart_2013}
M.~Fruchart and D.~Carpentier.
\newblock An introduction to topological insulators.
\newblock {\em Comptes Rendus Physique}, 14(9):779 -- 815, 2013.
\newblock Topological insulators / Isolants topologiques.

\bibitem{gerard1988precise}
C.~G{\'e}rard and A.~Grigis.
\newblock Precise estimates of tunneling and eigenvalues near a potential
  barrier.
\newblock {\em Journal of differential equations}, 72(1):149--177, 1988.

\bibitem{Hormander1979}
L.~H{\"o}rmander.
\newblock The weyl calculus of pseudo-differential operators.
\newblock {\em Communications on Pure and Applied Mathematics}, 32(3):359--443,
  1979.

\bibitem{Hormander}
L.~H{\"o}rmander.
\newblock {\em The analysis of linear partial differential operators III:
  Pseudo-differential operators}.
\newblock Springer Science \& Business Media, 2007.

\bibitem{Kalton}
N.~Kalton.
\newblock Trace-class operators and commutators.
\newblock {\em Journal of Functional Analysis}, 86(1):41--74, 1989.

\bibitem{Lerner}
N.~Lerner.
\newblock {\em Metrics on the phase space and non-selfadjoint
  pseudo-differential operators}, volume~3.
\newblock Springer Science \& Business Media, 2011.

\bibitem{McCann_2013}
E.~McCann and M.~Koshino.
\newblock The electronic properties of bilayer graphene.
\newblock {\em Reports on Progress in Physics}, 76(5):056503, apr 2013.

\bibitem{prodan2016bulk}
E.~Prodan and H.~Schulz-Baldes.
\newblock {\em Bulk and boundary invariants for complex topological
  insulators}.
\newblock Springer, Berlin, 2016.

\bibitem{macdonald_junction}
Z.~Qiao, J.~Jung, Q.~Niu, and A.~H. MacDonald.
\newblock Electronic highways in bilayer graphene.
\newblock {\em Nano Letters}, 11(8):3453--3459, 08 2011.

\bibitem{QB}
S.~Quinn and G.~Bal.
\newblock Approximations of interface topological invariants.
\newblock {\em arXiv preprint arXiv:2112.02686}, 2021.

\bibitem{quinn2022asymmetric}
S.~Quinn and G.~Bal.
\newblock Asymmetric transport for magnetic dirac equations.
\newblock {\em arXiv preprint arXiv:2211.00726}, 2022.

\bibitem{klaus2018}
P.~Rickhaus, J.~Wallbank, S.~Slizovskiy, R.~Pisoni, H.~Overweg, Y.~Lee,
  M.~Eich, M.-H. Liu, K.~Watanabe, T.~Taniguchi, T.~Ihn, and K.~Ensslin.
\newblock Transport through a network of topological channels in twisted
  bilayer graphene.
\newblock {\em Nano Letters}, 18(11):6725--6730, 11 2018.

\bibitem{san_jose_2013}
P.~San-Jose and E.~Prada.
\newblock Helical networks in twisted bilayer graphene under interlayer bias.
\newblock {\em Phys. Rev. B}, 88:121408, Sep 2013.

\bibitem{kaxiras_2021}
H.~Tang, S.~Carr, and E.~Kaxiras.
\newblock Geometric origins of topological insulation in twisted layered
  semiconductors.
\newblock {\em Physical Review B}, 104(15), Oct 2021.

\bibitem{Perez_2014}
G.~Usaj, P.~M. Perez-Piskunow, L.~E.~F. Foa~Torres, and C.~A. Balseiro.
\newblock Irradiated graphene as a tunable floquet topological insulator.
\newblock {\em Phys. Rev. B}, 90:115423, Sep 2014.

\bibitem{watson2023bistritzer}
A.~B. Watson, T.~Kong, A.~H. MacDonald, and M.~Luskin.
\newblock {B}istritzer--{M}ac{D}onald dynamics in twisted bilayer graphene.
\newblock {\em Journal of Mathematical Physics}, 64(3):031502, 2023.

\bibitem{witten_2016}
E.~Witten.
\newblock Three lectures on topological phases of matter.
\newblock {\em La Rivista del Nuovo Cimento}, 39(7):313--370, 2016.

\bibitem{Eugene_2013}
F.~Zhang, A.~H. MacDonald, and E.~J. Mele.
\newblock Valley chern numbers and boundary modes in gapped bilayer graphene.
\newblock {\em Proceedings of the National Academy of Sciences},
  110(26):10546--10551, 2013.

\bibitem{ValleyNature_2016}
J.~Zhou, S.~Cheng, W.-L. You, and H.~Jiang.
\newblock Effects of intervalley scattering on the transport properties in
  one-dimensional valleytronic devices.
\newblock {\em Scientific Reports}, 6(1):23211, 2016.

\bibitem{Zworski}
M.~Zworski.
\newblock {\em Semiclassical analysis}, volume 138.
\newblock American Mathematical Soc., 2012.

\end{thebibliography}

\appendix
\numberwithin{equation}{section}

\section{Proof of Theorem \ref{thm:tblg-bulk}}\label{sec:aa}

\tb{For the bilayer problem with coupling  $A=\eps\frac12(\sigma_1+i\sigma_2)$,  we recall that the tBLG Hamiltonian is
\[
   H_+(\xi) = \begin{pmatrix} 
   \Omega & \bar \xi & 0 & 0 \\ \xi & \Omega & \eps & 0 \\ 
   0 & \eps & -\Omega & \bar \xi \\  0 & 0 & \xi & -\Omega
   \end{pmatrix}.
\]
Consider first the case $\Omega>0$ and $\eps>0$. 
Note that for $\xi=0$, $\pm\Omega$ are eigenvalues associated to eigenvectors $e_1$ and $e_4$ while $\pm \sqrt{\Omega^2+\eps^2}$ are the last eigenvalues with eigenspaces supported on the span of $(e_2,e_3)$. 

The constraint ${\rm det}(H-E)=0$ on the eigenvalues is given by
\[
  (|\xi|^2-(E-\Omega)^2) (|\xi|^2-(E+\Omega)^2)= \eps^2 (E^2-\Omega^2).
\]
The four eigenvalues (called $E_j$ for $1\leq j\leq 4$ in the main text) come in pairs with opposite signs, are all different, and are given explicitly via
\[
  E_\pm^2 = (\Omega^2 + \frac12 \eps^2 + |\xi|^2)  \pm \sqrt{(4\Omega^2+\eps^2)|\xi|^2 + \frac 14 \eps^4} .
\]
To reflect this symmetry, we denote the eigenvalues as $-E_+=E_1$, $-E_-=E_2$, $E_-=E_3$ and $E_+=E_4$.
We observe that $E^2_-$ is large for $|\xi|$ large while $E=0$ implies $(|\xi|^2-\Omega^2)^2+\eps^2\Omega^2=0$ which is impossible, and hence the presence of a spectral gap. More precisely, we obtain that $|\xi|^2\to E_-^2$ is a convex function with a minimum equal to
\[
  E^2_{\rm min} = \frac {\Omega^2\eps^2}{4\Omega^2+\eps^2} \quad \mbox{and  attained at } \quad |\xi|^2=\frac{2\Omega^2(2\Omega^2+\eps^2)}{4\Omega^2+\eps^2}.
\]
The equations for the eigenvectors are
\[
  (\Omega-E) u_1 + \bar\xi u_2=0 \quad \xi u_1 + (\Omega-E)u_2+\eps u_3=0
\]
\[
  \eps u_2 - (\Omega+E)u_3 + \bar\xi u_4 =0,\quad \xi u_3 - (\Omega+E)u_4=0.
\]
We wish to construct eigenvectors that are smooth in the variable $\xi$. For $E=E_+$ the largest eigenvalue, we can choose the smooth eigenvector
\[
   u(\xi) = c_u \Big(\frac{\bar\xi}{E-\Omega} (\Omega+E-\frac{|\xi|^2}{\Omega+E}),  (\Omega+E-\frac{|\xi|^2}{\Omega+E}), \eps, \frac{\eps\xi}{\Omega+E} \Big)
\]
with $c_u=c_u(\xi)$ chosen such that $|u|=1$. Since $E_+\pm\Omega>0$ is bounded below by a positive constant, $\Cm\ni \xi\to c_u(\xi)$ is a smooth function bounded above and below by positive constants. 

This vector converges as $|\xi|\to\infty$ to an expression that depends on $(\eps,\Omega)$ in a non-trivial manner. Namely, in the limit $|\xi|\to\infty$:
\[
 0< E_\pm^2 = |\xi|^2 \pm \sqrt{4\Omega^2+\eps^2} |\xi| + O(1), 
\]  
and $u(\xi)$ converges in the same sense to
\[
   u_+^A(\hat\xi) = c_\infty (\hat{\bar\xi} \beta,\beta,\eps ,\eps \hat\xi)
   ,\quad \beta=2\Omega+\sqrt{4\Omega^2+\eps^2} \approx \Omega+E - \frac{|\xi|^2}{\Omega+E},\quad c_\infty= (2\eps^2+2\beta^2)^{-\frac12}.
\]
We obtain similarly that the differential 1-form $du$ converges to $du_+^A$ as $|\xi|\to\infty$. As a result, we have
\[
  W_+^4 = \frac i{2\pi} \lim_{R\to\infty} \dint_{\mathbb{S}^1_R} (u,du) = \frac i{2\pi} c_\infty^2 \dint_{\mathbb{S}^1}\big((\beta e^{-i\theta},\beta,\eps,\eps e^{i\theta}), (-i\beta e^{-i\theta},0,0,i\eps e^{i\theta})  \big) d\theta.
\]
This implies that
\[
  W_+^4= \frac i{2\pi} c_\infty^2 (-i\beta^2+i\eps^2) 2\pi = \dfrac{\beta^2-\eps^2}{2(\eps^2+\beta^2)}.
\]

Let us now look at the second positive eigenvalue $E=E_->0$. The (smooth) function $\xi\to E_-(\xi)-\Omega$ vanishes at $\xi=0$ (and in fact at two other values of $\xi$). Thus we may no longer divide by $(\Omega-E)$ in the construction of a smooth vector $u(\xi)$. However, it is easy to verify that $\xi^{-1}(E(\xi)-\Omega)$ is a smooth function. 
Eliminating $u_2$ and $u_4$ from the above system for $u$, we obtain
\[
  u = c_u\Big(\eps, \frac{E-\Omega}{\bar\xi }\eps , - (\xi- \frac{(\Omega-E)^2}{\bar \xi}) , \frac{-\xi}{\Omega+E}  (\xi- \frac{(\Omega-E)^2}{\bar \xi}) \Big).
\]
This converges as $|\xi|\to\infty$ to
\[
 u_-^A = c_\infty(\eps,\eps \hat\xi,-\beta\hat\xi,-\beta \hat\xi^2).
\]
Following the same procedure as above, we deduce that 
\[
  W_+^3 = \frac{-\eps^2-3\beta^2}{2(\eps^2+\beta^2)},\qquad 
  \qquad W_+=W_+^3+W_+^4 = -1.
\]


We now consider the second Hamiltonian and replace $A$ by $A^*$:
\[
   H_-(\xi) = \begin{pmatrix} 
   \Omega & \bar \xi & 0 &   \eps  \\ \xi & \Omega &0 &0\\ 
   0 & 0 & -\Omega & \bar \xi \\   \eps & 0 & \xi & -\Omega
   \end{pmatrix}.
\]
We observe that the equation for the eigenvalues $E$ is the same as in the previous case. For the largest positive eigenvalue, one finds
\[
  u = c_u\Big( \Omega+E-\frac{|\xi|^2}{\Omega+E}, \frac{\xi}{E-\Omega} ( \Omega+E-\frac{|\xi|^2}{\Omega+E}), \frac{\hat\xi}{\Omega+E} \eps,\eps\Big).
\]
In the limit $|\xi|\to\infty$, we have
\[
   u_+^{A^*} =c_\infty (\beta,\beta\hat\xi,\hat{\bar\xi} \eps,\eps),\quad \mbox{so that}\quad
   W_-^4= \dfrac{\eps^2-\beta^2}{2(\eps^2+\beta^2)}.
\]

For the eigenvalue $0<E=E_-$, we find
\[
 u = c_u \Big( \frac {E-\Omega}{|\xi|} \eps, \eps \frac{\bar\xi}{\Omega+E} \big( \frac{(\Omega-E)^2}{\xi}-\bar\xi\big),  \frac{(\Omega-E)^2}{\xi}-\bar\xi \Big).
\]  
In the limit $|\xi|\to\infty$, 
\[
  u_-^{A^*}  = c_\infty \big( \hat{\bar\xi} \eps,\eps , -\beta \hat{\bar \xi}^2 , -\beta  \hat{\bar \xi} \big)\quad \mbox{so that}\quad
   W_-^3= \dfrac{\eps^2+3\beta^2}{2(\eps^2+\beta^2)} \quad \mbox{ and } \quad W_-=W_-^3+W_-^4=1.
\]

\begin{figure}[htbp]
\centering
\includegraphics[width = .5\textwidth]{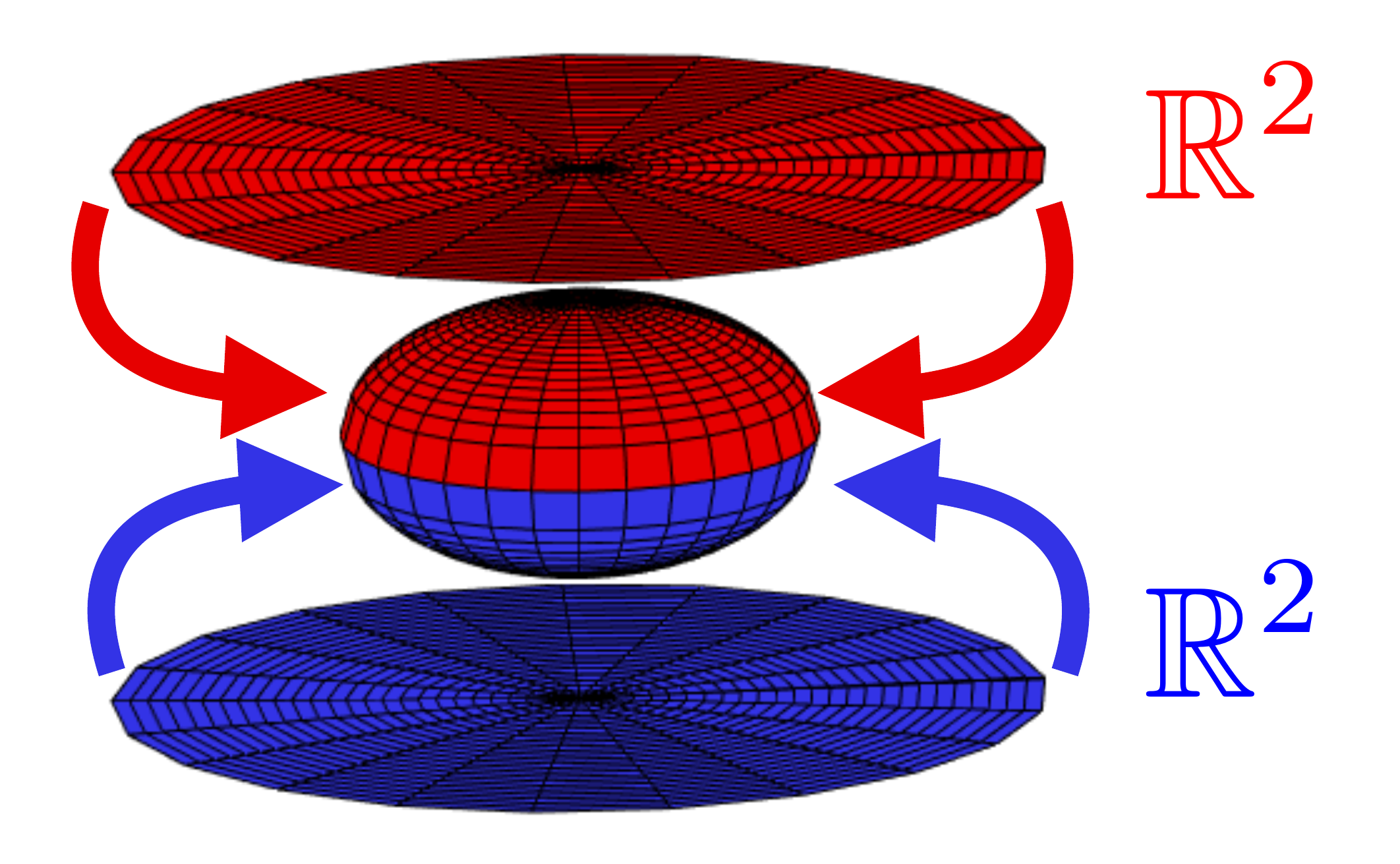}
\caption{\tb{A family of two projectors} $(\pm,\xi)\to \Pi_\pm(\xi) := |\psi_\pm(\xi)\rangle\langle\psi_\pm(\xi)|$ is mapped to the sphere by mapping each plane $\pm$ onto a hemisphere while preserving continuity along the `circle at infinity', \tb{i.e., the limits $\lim_{|\xi|\to\infty}\Pi_+(|\xi|\hat\xi)=\lim_{|\xi|\to\infty}\Pi_-(|\xi|\hat\xi)$ exist and agree for all $\hat\xi\in\mathbb{S}^1$.}}
\label{fig:glue}
\end{figure}

Upon inspection of the limiting vectors $u_\pm^A$ and $u_\pm^{A^*}$, we observe that $u_+^A$ and $u_+^{A^*}$ are vectors that are not defined up to a multiplicative phase. The same holds true for $u_-^A$ and $u_-^{A^*}$. However, we 
observe that 
\[
  v_-^{A^*} = \sigma_2\otimes I_2 u_-^{A^*}  = i \big(  \beta \hat{\bar \xi}^2 , \beta  \hat{\bar \xi},\hat{\bar\xi} \eps,\eps  \big) = i \hat{\bar\xi} \big( \beta \hat{\bar\xi},\beta,\eps,\eps\hat\xi\big) =  i \hat{\bar\xi} u_+^A.
\]
So the following gluing procedure is possible. Let $\Pi_+(\xi)=u_+^4\otimes u_+^4(\xi)$ be the projector where $u_+^4(\xi)$ converges to $u_+^A(\hat\xi)$ as $|\xi|\to\infty$. Let similarly $\Pi_-(\xi)=v_-^3\otimes v_-^3(\xi)$ with $v_-^3(\xi)=\sigma_2\otimes I_2 u_-^3$ converging to $v_-^{A^*}$ as $|\xi|\to\infty$. Since $v_-^{A^*}$ is proportional to $u_+^A$, then $\Pi_+(|\xi|\hat\xi)$ and $\Pi_-(|\xi|\hat\xi)$ have the same (well defined) limit as $|\xi|\to\infty$ for each $\hat\xi\in\mathbb{S}_1$. We can then apply the gluing procedure of \cite{bal3} to obtain that
\[
 -1 = W_+^4-W_\sq{-}^3 = \frac i{2\pi}\dint_{\Rm^2} \sq{\tr} \Pi_+ d\Pi_+ \wedge d\Pi_+ - \sq{\tr} \Pi_- d\Pi_- \wedge d\Pi_\sq{-}\in\Zm
\]
is an integer-valued Chern number integrated over a sphere (the compactification of two planes $\Rm^2$ over which the two projectors $\Pi_\pm$ are defined); see Figure \ref{fig:glue}.


We can similarly glue the $E_+$ branch of $A^*$ with $E_-$ branch of $A$ and get another invariant $W_+^3-W_-^4=1$ using that 
\[
  v_+^{A^*} =  \sigma_2\otimes I_2 u_+^{A^*} = -i (\hat{\bar\xi} \eps ,\eps , -\beta, -\beta \hat\xi) = -i \hat{\bar\xi} u_-^A.
\]
In some sense, the second Hamiltonian with coupling term $A^*$ should be replaced by the equivalent operator $\sigma_2\otimes I_2 \hat H_{A^*} \sigma_2 \otimes I_2$. No matter the basis, we observe that in order to get bulk-difference invariants, we need to glue projectors in the AB and BA sectors corresponding to different eigenvalues. 

\medskip

This concludes the proof of Theorem \ref{thm:tblg-bulk} when $\Omega>0$, $\eps\equiv\lambda>0$ and $\eta=1$. It is reasonably straightforward to observe that the sign of $\eps$ has no influence on the derivation. All we need is $\eps\not=0$ so that a gap opens. The sign of $\Omega$, however, matters since for $\Omega<0$ and $E>0$, it is $(\Omega+E)^{-1}$ that is no longer bounded.} 



\tm{
\sq{To address how $W$ changes depending on system parameters $\Omega, \eps$ and $\eta$, define}
\begin{equation}
H_\pm (\xi; \Omega, \eps, \eta) = \Omega \sigma_3\otimes I_2 + I_2 \otimes (\xi_1\sigma_1 + \eta \xi_2\sigma_2) + \frac{\eps}{2} (\sigma_1\otimes \sigma_1  \pm \sigma_2\otimes \sigma_2).
\end{equation}
We then define the following symmetry operations:
\begin{align*}
S_1 = \sigma_1 \otimes I_2, \qquad S_2\psi(\xi) = \psi(-\xi),\qquad
S_3 = \sigma_1\otimes \sigma_1.
\end{align*}
We let $W(\Omega,\eps,\eta)$ denote the bulk-difference invariant as described above. For $\Omega, \eps, \eta > 0$, by the above calculation we have
\[
W(\Omega,\eps,\eta) = -2.
\]
We have the symmetry relations
\begin{align*}
& S_1 H_\pm(\xi; \Omega,\eps,\eta)S_1 = H_\mp (\xi; -\Omega,\eps,\eta), \\
& S_2 H_\pm(\xi; \Omega,\eps,\eta)S_2 = -H_\pm(\xi; -\Omega,-\eps,\eta),\\
& S_3 H_\pm(\xi; \Omega,\eps,\eta)S_3 =  H_\pm(\xi; -\Omega, \eps,-\eta). 
\end{align*}
Since $(S_j\psi, d(S_j\psi)) = (S_j\psi, S_jd\psi) = (\psi,d\psi)$ for $j \in \{1,3\}$, and likewise $\int (\psi,d\psi) = \int (S\psi,dS_2\psi)$, the connection is invariant to the symmetry operations. However, when the sign of the Hamiltonian changes, i.e. if $S_j H S_j = - \tilde H$, the invariant of $H$ is minus the invariant of $\tilde H$ as the order of the bands is reversed, and the sum of the four connnections (including positive and negative energies) is zero. We thus conclude
\begin{equation*}
     W(\Omega,\eps,\eta) = - W(-\Omega,\eps,\eta) = W(\Omega,-\eps,\eta) = W(-\Omega,\eps,-\eta) = - W(\Omega,\eps,-\eta).
\end{equation*}
This concludes the proof of \eqref{eq:invthm1} in the theorem. 

For \eqref{eq:invthm2}, we can apply the bulk-edge correspondence in \cite{QB} since $H_e$ satisfies the required ellipticity conditions. 
That is, $H_e = \Op (\sigma)$, where $\sigma \in S^{1}_{1,0}$ with all of its singular values bounded below by $\aver{\xi,\zeta}$ whenever $\aver{y,\xi,\zeta}$ is sufficiently large. Moreover, 
the bulk Hamiltonians $H_\pm$ 
have a spectral gap at $0$, as demonstrated by
the above calculations. 
Thus $H_e$ satisfies (H1) from \cite{QB}, 
meaning that \eqref{eq:invthm2} holds.}

\section{Proofs from Section \ref{sec:junctions} and relevant notation}\label{sec:pf}
\subsection{Notation and functional setting}\label{subsec:notation}
We first briefly define the notation used in Section \ref{sec:junctions} regarding pseudo-differential operators. For a more detailed exposition, see \cite[section 2]{QB} and references therein.

Given a parameter $h \in (0,1]$ and a symbol $a(x,\xi;h) \in \mathcal{S}' (\mathbb{R}^d \times \mathbb{R}^d) \otimes \mathbb{M}_n$, we define the Weyl quantization of $a$ as the operator 
\begin{align}\label{eq:weylquant}
    \Op _h (a) \psi (x) :=
    \frac{1}{(2\pi h)^d} \int_{\mathbb{R}^{2d}}
    e^{i(x-y)\cdot \xi/h}
    a(\frac{x+y}{2}, \xi;h) \psi (y) dy d\xi,
    \qquad
    \psi \in \mathcal{S} (\mathbb{R}^d) \otimes \mathbb{C}^n.
\end{align}
Here, $\mathcal{S}$ denotes the Schwartz space \sq{(with $\mathcal{S}'$ its dual)} and $\mathbb{M}_n$ the space of Hermitian $n \times n$ matrices.
\sq{We define $\Op (a) := \Op _1 (a)$.}

A function $u: \mathbb{R}^{2d} \rightarrow [0,\infty)$ is called an order function if there exist constants $C_0 > 0$, $N_0 > 0$ such that
$u(X) \le C_0 \aver{X-Y}^{N_0} u(Y)$ for all $X,Y \in \mathbb{R}^{2d}$. Here we use the notation $\aver{X} := \sqrt{1 + |X|^2}$.
Note that if $u_1$ and $u_2$ are order functions, then so is $u_1 u_2$.

We say that $a \in S(u)$ if
for every $\alpha \in \mathbb{N}^{2d}$, there exists $C_\alpha > 0$ such that $|\partial^\alpha a (X;h)| \le C_\alpha u(X)$
for all $X \in \mathbb{R}^{2d}$ and $h \in (0,1]$.
We write $S(u^{-\infty})$ to denote the intersection over $s \in \mathbb{N}$ of $S(u^{-s})$.
For $\delta \in [0,1]$ and $k \in \mathbb{R}$, we say that $a(X;h) \in S^k_\delta (u)$ if for every $\alpha \in \mathbb{N}^{2d}$, there exists $C_\alpha > 0$ such that
\begin{align}\label{eq:symbolm}
    |\partial^\alpha a (X;h)| \le 
    C_\alpha u(X) h^{-\delta |\alpha| - k},
\end{align}
uniformly in $X \in \mathbb{R}^{2d}$ and $h \in (0,1]$.
If either $k$ or $\delta$  are omitted, they are assumed to be zero.

By \cite[Chapter 7]{sjostrand}, we know that if $a \in S(u_1)$ and $b \in S(u_2)$, then $\Op_h (c) := \Op_h (a) \Op_h(b)$ is a pseudo-differential operator, with
\begin{align*}
    c (x,\xi) =
    (a \sharp_h b) (x,\xi) :=
    \Big( e^{i\frac{h}{2} (\partial_x \cdot \partial_\zeta - \partial_y \cdot \partial_\xi)}
    a(x,\xi) b(y,\zeta) \Big)
    \Big \vert_{y=x,\zeta=\xi}
\end{align*}
and $c \in S(u_1 u_2)$.
We write $A \in \Op_h (S (u))$ to mean that $A = \Op_h (a)$ for some $a\in S(u)$.

Following \cite{Bony2013,Hormander1979,Lerner}, we define the H\"ormander class $S^{m}_{1,0}$ to be the space of symbols $a(x,\xi)$ that satisfy
\begin{align}\label{eq:symbol10}
    |(\partial^\alpha_\xi \partial^\beta_x a) (x,\xi)|
    \le 
    C_{\alpha, \beta} \aver{\xi}^{m - |\alpha|};
    \qquad
    \alpha,\beta \in \mathbb{N}^{d}.
\end{align}
\sq{By \cite{Bony2013,Hormander}, we know that if $H$ satisfies (H1), then $H : \mathcal{D} (H) \rightarrow \mathcal{H}$ is self-adjoint with $\mathcal{D} (H) = \mathcal{H}^m$.}

\sq{By \cite{Bony2013}, we know that if $a \in S^{m}_{1,0}$ is Hermitian-valued and satisfies that the smallest singular value $|a_{\min}(x,\xi)| \ge c\aver{\xi,\zeta}^m - 1$ for some $c > 0$, then
for all $\Im (z) \ne 0$,
$R_z := (z - \Op (a))^{-1}$ is a bijection of $L^2 (\mathbb{R}^d) \otimes \mathbb{C}^n$ onto $\mathcal{H}^m$, and $R_z = \Op (r_z)$ with $r_z \in S^{-m}_{1,0}$.
Defining $z =: \lambda + i\omega$, we apply \cite[Proposition 8.4]{sjostrand} to obtain
\begin{align*}
    |\partial^\alpha_x \partial^\beta_\xi r_z (x,\xi;h)|\le
    C_{Z_0,\alpha, \beta}
    \Big(
    1+
    \frac{\sqrt{h}}{|\omega|}
    \Big)^{2d+1} |\omega|^{-(|\alpha|+|\beta|)-1},
    \qquad
    z \in Z_0
\end{align*}
for any compact set $Z_0 \subset \mathbb{C}$.
To derive the above result, it helps to write
$
    (z-H)^{-1} = (i-H)^{-1} (I + (z-i)(z-H)^{-1}),
$
where $(i-H)^{-1} \in S^{-m}_{1,0}$ is independent of $z$, and $(z-H)^{-1}$ is bounded with $\norm{(z-H)^{-1}} \le \frac{C}{|\omega|}$.
The result from \cite{Bony2013} does not address how symbolic bounds of $r_z$ depend on the semi-classical parameter $h$.
That is, given $\Op (r_{z,h}) = R_{z,h} := (z-\Op_h(a))^{-1}$, we would like bounds on $r_{z,h}$ that are uniform in $h$. This is where we use \cite[Chapter 8]{sjostrand}, which tells us there exists $h_0 > 0$ such that for all $h \in (0,h_0]$, we have $\Op (r_{z,h}) = \Op_h (\tilde{r}_{z,h})$, with $\tilde{r}_{z,h} \in S(\aver{\xi}^{-m})$.
Since the symbolic bounds from \cite{Bony2013} must be continuous in $h$, this means
$\tilde{r}_{z,h} \in S(\aver{\xi}^{-m})$ uniformly in $h \in (0,1]$, as desired.}

Given $\phi \in \mathcal{C}^\infty_0 (\mathbb{R})$, there exists an almost analytic extension $\tilde{\phi} \in \mathcal{C}^\infty_0 (\mathbb{C})$ that satisfies
\begin{align} \label{aae}
    |\bar{\partial} \tilde{\phi}| \le C_N |\Im z|^N, &\quad 
    N \in \{0,1,2,\dots\}; \qquad
    \tilde{\phi} (\lambda) = \phi(\lambda), \quad \lambda \in \mathbb{R}.
\end{align}
We now recall \cite[Theorem 8.1]{sjostrand}.
If $H$ is a self-adjoint operator on a Hilbert space, then
\begin{align} \label{HSformula}
    \phi (H) =
    -\frac{1}{\pi} \int \bar{\partial} \tilde{\phi} (z) (z-H)^{-1} d^2 z,
\end{align}
where $\bar{\partial} := \frac{1}{2} \partial_{\Re z} + \frac{i}{2} \partial_{\Im z}$ and
$d^2 z$ is the Lebesgue measure on $\mathbb{C}$.
(\ref{HSformula}) is known as the Helffer-Sj\"ostrand formula.

\sq{Suppose $\fm \in L^1 (\mathbb{R}^{2d})$, and 
$|\partial^\alpha a(x,\xi;h)| \le C_\alpha \fm(x,\xi)$ for all $\alpha \in \mathbb{N}^{2d}$ and $h \in (0,1]$ (meaning that $a \in S(\fm)$).
Then by \cite[Theorem 9.4]{sjostrand}, $\Op_h (a)$ is trace-class with
$
    \norm{\Op_h (a)}_1 \le C \max_{|\alpha| \le 2d+1} C_\alpha \norm{\fm}_{L^1}
$
and
\begin{align}\label{eq:ophtrace}
    \Tr \Op_h (a) =
    \frac{1}{(2\pi h)^d}
    \int_{\mathbb{R}^{2d}} \tr a(x,\xi;h) dx d\xi,
\end{align}
where $C$ depends only on $d$ and $\tr$ is the standard matrix trace.}

\subsection{Proofs}

\sq{\begin{proof}[Proof of Proposition \ref{prop:H1tBLG}]
    We have
    \begin{align*}
        \sigma = \begin{pmatrix} \Omega I + (\xi,\zeta) \cdot \sigma^{(\eta)} & \lambda U^*(x,y) \\ \lambda U(x,y) & -\Omega I + (\xi,\zeta) \cdot \sigma^{(\eta)}\end{pmatrix},
    \end{align*}
    with $U:\mathbb{R}^2 \rightarrow \mathbb{R}$ smooth and bounded. Thus it is clear that $\sigma$ is Hermitian-valued, with $\sigma \in S^1_{1,0}$ and $|\sigma_{\min} (x, y, \xi, \zeta)| \ge c \aver{\xi, \zeta}- 1$ for some $c > 0$. 
    By assumption, $U(x,y) = A$ (resp. $U(x,y) = A^*$) whenever $f(x,y)$ (resp. $-f(x,y)$) is sufficiently large. Thus $\sigma_\pm$ from (H1) are well defined, 
    with the calculations in Appendix \ref{sec:aa} proving that $\sigma_\pm$ both have a spectral gap at $0$. This completes the proof.
\end{proof}}

\sq{\begin{proof}[Proof of Theorem \ref{thm:idx}]
    We will prove that
    \begin{align}\label{eq:idx}
        2\pi \sigma_I (\genH,P) = \Tr [U,P] U^* = \Tr [U,\projP] U^* = {\rm Index} (\projP U \projP),
    \end{align}
    where we use the shorthand $U := U (\genH)$.

    We begin by proving the first equality in \eqref{eq:idx}.
    Let $\Uminusone := U-I$, and observe that $\Uminusone \in \mathcal{C}^\infty_c (\{\Phi=1\}^\circ)$. We have
    \begin{align*}
        \Tr [U,P]U^* = \Tr [\Uminusone,P]U^* = \Tr [\Uminusone,P]\Uminusone^* + \Tr [\Uminusone,P].
    \end{align*}
    Using that $[(z-\genH)^{-1},P] = (z-\genH)^{-1} [\genH,P] (z-\genH)^{-1}$, the Helffer-Sj\"ostrand formula \eqref{HSformula} and cyclicity of the trace \cite{Kalton} imply that
    \begin{align*}
        \Tr [\Uminusone,P] \Uminusone^* = \Tr \Big(-\frac{1}{\pi} \int_\mathbb{C} \bar{\partial} \tilde{\Uminusone} (z) [\genH,P] (z-\genH)^{-2} d^2 z\Uminusone^* (\genH)\Big).
    \end{align*}
    After integrating by parts in $\partial$
    and using that
$\partial \tilde{\Uminusone} = \tilde{\Uminusone'}$ for some almost analytic extension $\tilde{\Uminusone'}$ of $\Uminusone'$,
    we see that
    \begin{align}\label{eq:traces}
        \Tr [\Uminusone,P] \Uminusone^* = \Tr [\genH,P] \Uminusone ' \Uminusone^* = 2\pi \sigma_I (\genH,P) - \Tr 2\pi i [\genH, P] \varphi 'U
        = 2\pi \sigma_I (\genH,P) -\Tr [\genH,P] \Uminusone '.
    \end{align}
    We have thus shown that
    \begin{align*}
        \Tr [U,P]U^* = 2\pi \sigma_I (\genH,P) -\Tr [\genH,P] \Uminusone '+ \Tr [\Uminusone,P].
    \end{align*}
    By the same logic used in \eqref{eq:traces}, we obtain that
    $\Tr [\genH,P] \Uminusone ' = \Tr [\genH,P] \Uminusone ' \Phi= \Tr [\Uminusone,P] \Phi$, hence
    \begin{align*}
        \Tr [U,P]U^* = 2\pi \sigma_I (\genH,P) +\Tr [\Uminusone,P](1-\Phi).
    \end{align*}
    Again using cyclicity of the trace, this means
    \begin{align*}
        \Tr [U,P]U^* = 2\pi \sigma_I (\genH,P) +\Tr \Psi_1 [\Uminusone,P]\Psi_2, 
    \end{align*}
    for some $\Psi_j = \psi_j (\genH)$, where $\psi_j \in \mathcal{C}^\infty (\mathbb{R})$ vanishes on $\supp (\Uminusone)$. The first equality in \eqref{eq:idx} follows.

    We now prove the second equality of \eqref{eq:idx}.
    Let $\chi (x,y) = \chi_1 (g(x,y))$ where $\chi_1 \in \mathcal{C}_c^\infty$ such that $\chi (\projP-P)= \projP-P$. 
Then 
\begin{align*}
    [U,\projP-P] U^* = [\Uminusone,\projP-P] U^* = \Uminusone\chi (\projP-P) U^* - (\projP-P) \chi \Uminusone U^*
\end{align*}
is trace class, as $\Uminusone \chi$ and $\chi \Uminusone$ are both trace-class by the $\Psi$DO calculus with $P-P_1$ and $U^*$ bounded.
Since we know that $[U,P] U^*$ is trace-class, this proves that $$[U,\projP]U^* = [U,\projP-P] U^* + [U,P] U^*$$ is trace-class.
Now,
\begin{align*}
    \Tr [U,\projP-P] U^* = \Tr \Uminusone \chi (\projP-P) U^* - \Tr (\projP-P) \chi \Uminusone U^*,
\end{align*}
with 
\begin{align*}
    \Tr \Uminusone\chi (\projP-P) U^* = \Tr (\projP-P) U^* \Uminusone\chi
\end{align*}
and
\begin{align*}
    \Tr (\projP-P) \chi \Uminusone U^* = \Tr \chi (\projP-P) \chi \Uminusone U^* =\Tr (\projP-P) \chi \Uminusone U^* \chi =\Tr (\projP-P) \Uminusone U^* \chi.
\end{align*}
Since $[\Uminusone,U^*] = 0$, we have proven that $\Tr [U,\projP-P] U^* = 0$, which verifies the second equality of \eqref{eq:idx}. 

Finally, the last equality of \eqref{eq:idx} follows immediately from \cite[Proposition 2.4]{avron1994charge} and the fact that $[U,\projP]U^*$ is trace-class.
\end{proof}

\begin{proof}[Proof of Theorem \ref{thm:invvarphi}]
    For $\mu \in [0,1]$, define $\varphi_\mu := \varphi + \mu (\varphi_1 - \varphi)$.
    Theorem \ref{thm:idx} implies that $\Tr i [\genH, P] \varphi'_\mu (\genH) = {\rm Index} (\projP U_\mu \projP)$ for all $\mu \in [0,1]$, where $U_\mu := e^{i2\pi \varphi_\mu (\genH)}$. With $\Uminusone_\mu := U_\mu - I$, 
    the Helffer-Sj\"ostrand formula \eqref{HSformula} implies
    \begin{align*}
        U_{\mu_2} - U_{\mu_1} = \Uminusone_{\mu_2} - \Uminusone_{\mu_1} = -\frac{1}{\pi} \int_\mathbb{C} \bar{\partial} (\tilde{\Uminusone}_{\mu_2} (z) -\tilde{\Uminusone}_{\mu_1} (z)) (z-\genH)^{-1} d^2 z.
    \end{align*}
    Since $|\bar{\partial} (\tilde{\Uminusone}_{\mu_2} (z) - \tilde{\Uminusone}_{\mu_1} (z))| \le C |\mu_2 - \mu_1|$ uniformly in $\mu_1, \mu_2 \in [0,1]$ and $z \in \mathbb{C}$,
    it follows that $\norm{U_{\mu_2} - U_{\mu_1}} \rightarrow 0$ as $\mu_2-\mu_1 \rightarrow 0$.
    Therefore, by \cite[Theorem 19.1.5]{Hormander1979}, ${\rm Index} (\projP U_\mu \projP)$ is independent of $\mu \in [0,1]$, and the result is complete.
\end{proof}

We 
now state the following regularity result.}
\begin{proposition} \label{trclass}
Suppose $H = \Op (\sigma)$ satisfies (H1),
define
\begin{equation}\label{eq:Hh}
H_h := \Op _h (\sigma (x,y,\xi,\zeta))
\mbox{ for }\  0 < h \le 1,
\end{equation}
and let 
$\Phi \in \mathcal{C}^\infty_c (E_1, E_2)$.
Then for all $P \in \fs (0,1;g(x,y))$, we have
\begin{align*}
\Phi (H_h) \in \Op_h(S(\langle f(x,y),\xi,\zeta \rangle^{-\infty}))
\qquad \text{and} \qquad
[H_h,P]\Phi (H_h) \in \Op_h(S(\langle x,y,\xi,\zeta \rangle^{-\infty})) .
\end{align*}
\end{proposition}

\sq{By the paragraph above \eqref{eq:ophtrace}, Proposition \ref{trclass} immediately implies that $\sigma_I (H_h,P)$ is well defined. 
By Theorem \ref{thm:idx}, Proposition \ref{trclass} also implies that $2\pi \sigma_I (H) = {\rm Index} (\projP U(H) \projP)$ for any $H$ satisfying (H1).

\begin{proof}
    We emulate the proofs of \cite[Lemmas 3.3 and 3.4]{quinn2022asymmetric}.
    
    For any $p>0$, we can write $\Phi (H_h) = (i-H_h)^{-p} \Phi_p (H_h)$ with $\Phi_p \in \mathcal{C}^\infty_c (E_1, E_2)$.
    Since $H$ satisfies (H1), 
    we know that $(i-H_h)^{-1} \in S^{-m}_{1,0}$ (see Appendix \ref{subsec:notation} for more details). It follows from the composition calculus that
$\Phi (H_h) \in \Op_h (S(\aver{\xi,\zeta}^{-\infty}))$.
By assumption on $\sigma_\pm$, we know that
$H_{\pm,h} := \Op_h (\sigma_\pm)$ has a spectral gap in $(E_1, E_2)$ for all $h \in (0,1]$, 
hence $\Phi (H_{\pm,h}) = 0$.
We can thus write $\Phi (H_h) = \phi(x,y) (\Phi (H_h) - \Phi (H_{+,h})) + (1-\phi (x,y)) (\Phi (H_h) - \Phi (H_{-,h}))$, for some $\phi \in \fs (0,1;f)$.
The Helffer-Sj\"ostrand formula \eqref{HSformula} implies that
\begin{align*}
    \Phi (H_h) - \Phi (H_{+,h}) = \frac{1}{\pi}\int_{\mathbb{C}} \bar{\partial} \tilde{\Phi} (z) (z-H_h)^{-1} (H_h - H_{+,h}) (z-H_{+,h})^{-1} d^2 z.
\end{align*}
Since $\sigma - \sigma_+$ vanishes whenever $f(x,y)$ is sufficiently large, it follows that $\Phi (H_h) - \Phi (H_{+,h}) \in \Op_h (S (\aver{f_+(x,y)}^{-\infty}))$, where $f_+ := \max \{f, 0\}$. Since $\phi$ vanishes whenever $-f(x,y)$ is sufficiently large, we conclude that $\phi(x,y) (\Phi (H_h) - \Phi (H_{+,h})) \in \Op_h (S (\aver{f(x,y)}^{-\infty}))$. The same reasoning shows that $(1-\phi (x,y)) (\Phi (H_h) - \Phi (H_{-,h})) \in \Op_h (S (\aver{f(x,y)}^{-\infty}))$.
We have thus shown that
$$\Phi (H_h) \in \Op_h (S(\aver{\xi,\zeta}^{-\infty}) \cap S (\aver{f(x,y)}^{-\infty})) = \Op_h(S(\langle f(x,y),\xi,\zeta \rangle^{-\infty})),$$ which proves the first result. 

For the second claim, observe that $[H_h, P] = (1-P)H_h P - P H_h (1-P) \in \Op_h (S (\aver{g(x,y)}^{-\infty} \aver{\xi,\zeta}^m))$.
The composition calculus and 
\eqref{eq:fg} then imply that
$[H_h,P]\Phi (H_h) \in \Op_h(S(\langle x,y,\xi,\zeta \rangle^{-\infty})),$ and the proof is complete.
\end{proof}

\begin{proof}[Proof of Theorem \ref{thm:h}]
    For $h \in (0,1]$, define $U_h := e^{i2\pi \varphi (H_h)}$. Proposition \ref{trclass} and Theorem \ref{thm:idx} imply that $2\pi\sigma_I (H_h, P) = {\rm Index} (\projP U_h \projP)$ for all $h \in (0,1]$. 
    It thus suffices to show that for any fixed $h \in (0,1]$, ${\rm Index} (\projP U_{h'} \projP)$ is constant over $h'$ in an open neighborhood of $h$. Using the Helffer-Sj\"ostrand formula \eqref{HSformula}, we write
    \begin{align*}
        U_{h'} - U_h = \frac{1}{\pi} \int_\mathbb{R} \bar{\partial} \tilde{\Uminusone} (z) (z-H_{h'})^{-1} (H_{h'}-H_h) (z-H_h)^{-1} d^2 z.
    \end{align*}
    Since $\norm{(z-H_{h'})^{-1}} \le |\Im z|^{-1}$ and $H_{h'}-H_h \in \Op (|h'-h| S (\aver{\xi,\zeta}^m))$ with $(z-H_h)^{-1} \in \Op ( S (\aver{\xi, \zeta}^{-m}))$ for all $\Im z \ne 0$ (with bounds growing at most algebraically in $|\Im z|^{-1}$), the rapid decay of $\bar{\partial} \tilde{\Uminusone} (z)$ near the real axis implies that $\norm{U_{h'} - U_h} \le C |h'-h|$. The result then follows from \cite[Theorem 19.1.5]{Hormander1979}.
\end{proof}

\begin{proof}[Proof of Theorem \ref{thm:bdd}]
We will apply Theorem \ref{thm:idx}. 
First, observe that 
$\Hmu = \Op (\sigmamu)$ with $\sigmamu \in \Op (S^{m}_{1,0})$. Moreover, 
(H1) and the assumption $\pert \in \Op (S^m_{1,0})$ implies that $|\sigmamu_{\min} (x,y,\xi,\zeta)| \ge c \aver{\xi,\zeta}^m-1$ whenever $\mu>0$ is sufficiently small,
where $\sigmamu_{\min}$ denotes the smallest magnitude eigenvalue of $\sigmamu$.
Since $\sigmamu$ is Hermitian-valued, it follows that whenever $\mu$ is sufficiently small, 
$\Hmu$ is self-adjoint with the same domain of definition $\mathcal{D} (\Hmu) = \mathcal{D} (H) = \mathcal{H}^m$, and 
$(z-\Hmu)^{-1} \in \Op (S^{-m}_{1,0})$ for all $\Im z \ne 0$. We refer to \cite{Bony2013} for more details, particularly Corollary 2 and the paragraph following Theorem 3.

Next, we verify that
$\Phi (\Hmu) \in \Op (S (\aver{f(x,y),\xi,\zeta}^{-\infty}))$ for all $\mu > 0$ sufficiently small.
Let $\Phi_0 \in \mathcal{C}^\infty_c (E_1, E_2)$ such that $\Phi \in \mathcal{C}^\infty_c (\{\Phi_0 = 1\}^\circ)$. With
$\Theta := \Phi (\Hmu) - \Phi (H)$ and $\Theta_0 := \Phi_0 (\Hmu) - \Phi_0 (H)$, it follows that
\begin{align*}
    \Phi (\Hmu) = \Phi (H) + \Theta, \qquad \Theta = \Theta \Theta_0 + \Phi (H) \Theta_0 + \Theta \Phi_0 (H).
\end{align*}
By assumption on $\pert$, $1-\Theta_0$ is invertible whenever $\mu>0$ is sufficiently small, with $(1-\Theta_0)^{-1} \in \Op (S(1))$.
By Proposition \ref{trclass}, $\Theta = (\Phi (H) \Theta_0 + \Theta \Phi_0 (H)) (1-\Theta_0)^{-1} \in \Op (S (\aver{f(x,y),\xi,\zeta}^{-\infty}))$ and thus $\Phi (\Hmu) \in \Op (S (\aver{f(x,y),\xi,\zeta}^{-\infty}))$ for all $\mu > 0$ sufficiently small, as desired.

We have shown that $\Hmu$ satisfies the assumptions of Theorem \ref{thm:idx}, hence $2\pi\sigma_I (\Hmu, P) = {\rm Index} (\projP \Umu \projP)$ whenever $\mu$ is sufficiently small, with $\Umu := e^{i2\pi \varphi (\Hmu)}$.
As in the 
proof of Theorem \ref{thm:h} above, we use the Helffer-Sj\"ostrand formula and the property 
\begin{align*}
    (z-\Hmu)^{-1} \in \Op (S (\aver{\xi,\zeta}^{-m})), \qquad \Im z \ne 0
\end{align*}
to verify that $\Umu$ is continuous in $\mu$ (with respect to operator norm). 
The result follows.
\end{proof}

\begin{proof}[Proof of Theorem \ref{thm:compact}]
    We again must show that $\Hmu$ satisfies the assumptions of Theorem \ref{thm:idx}, this time for all $\mu \in [0,1]$.
    The hypothesis (H1) and our growth constraint on $W$ imply that $\Hmu = \Op (\sigmamu)$ with $\sigmamu$ Hermitian-valued and satisfying $\sigmamu \in S^m_{1,0}$ and $|\sigmamu_{\min} (x,y,\xi,\zeta)| \ge c \aver{\xi,\zeta}^m - 1$. 
    As in the proof of Theorem \ref{thm:bdd}, this means $\Hmu$ is self-adjoint with domain of definition $\mathcal{D} (\Hmu) = \mathcal{H}^m$, and $(z-\Hmu)^{-1} \in \Op (S^{-m}_{1,0})$ whenever $\Im z \ne 0$ \cite{Bony2013}.
    

To obtain decay of the symbol of $\Phi (\Hmu)$, we again use the identity
\begin{align}\label{eq:theta}
    \Theta (1-\Theta_0) = \Phi (H) \Theta_0 + \Theta \Phi_0 (H),
\end{align}
with $\Theta$ and $\Theta_0$ defined in the proof of Theorem \ref{thm:bdd}. Proposition \ref{trclass} and \eqref{eq:theta} imply that $\Theta (1-\Theta_0) \in \Op (S (\aver{f(x,y),\xi,\zeta}^{-\infty}))$. By assumption on $\pert$, an application of the Helffer-Sj\"ostrand formula reveals that $\Theta_0 \in \Op (S (\aver{x,y,\xi,\zeta}^{-\delta}))$. Thus there exist $\Theta_{00} \in \mathcal{C}^\infty_c (\mathbb{R}^4)$ and $\Theta_{01} \in S(1)$ as small as necessary such that $\Theta_0 = \Theta_{00} + \Theta_{01}$. It follows that
\begin{align*}
    \Theta (1-\Theta_{01}) = \Theta (1-\Theta_{0}) + \Theta \Theta_{00}\in \Op (S (\aver{f(x,y),\xi,\zeta}^{-\infty})),
\end{align*}
hence
\begin{align*}
    \Theta = (\Theta (1-\Theta_{0}) + \Theta \Theta_{00})(1-\Theta_{01})^{-1}\in \Op (S (\aver{f(x,y),\xi,\zeta}^{-\infty})).
\end{align*}
We have shown that $\Theta$, and therefore $\Phi (\Hmu)$, has symbol with decay required by Theorem \ref{thm:idx}.

As in the proof of Theorem \ref{thm:bdd}, we can now conclude that $2\pi\sigma_I (\Hmu, P) = {\rm Index} (\projP \Umu \projP)$ for all $\mu \in [0,1]$,
and the result follows.
\end{proof}

We now prove Theorem \ref{thm:bic}, which 
states a bulk-interface correspondence. 
The strategy is to use an asymptotic expansion in the semiclassical parameter $h$ and apply Theorem \ref{thm:h} to eliminate terms that are not $O(1)$.
A similar technique is used in \cite{B-higher-dimensional-2021,bal3,QB}.


\begin{proof}[Proof of Theorem \ref{thm:bic}]
By Theorem \ref{thm:invvarphi}, $\sigma_I (H,P)$ is independent of $\varphi \in \fs (0,1;E_1,E_2)$, thus we can take $\varphi ' \in \mathcal{C}^\infty_c (\anot, \alpha)$ for some $\anot > E_1$. We will use the shorthand $\sigma_I := \sigma_I (H,P)$.

Let $\Op_h \nu_h := \varphi ' (H_h)$.
Using \cite[equation (7.19) and the preceding paragraph]{sjostrand} and \cite[Theorems 3.14 \& 4.16 and their proofs]{Zworski} as in \cite{QB}, we verify that
\begin{align*}
        \nu_h + \frac{1}{\pi} \int_\mathbb{C} \bar{\partial} \tilde{\varphi '} (z) \tilde{q}_{z,h} d^2 z \in S^{-2} (\aver{f(x,y),\xi,\zeta}^{-\infty}),
\end{align*}
where
\begin{align*}
    \tilde{q}_{z,h} = \sigma_z^{-1} +\frac{ih}{2} \{\sigma_z^{-1}, \sigma_z\}\sigma_z^{-1}, \qquad \{a,b\} := \partial_\xi a \partial_x b + \partial_\zeta a \partial_y b - \partial_x a \partial_\xi b - \partial_y a \partial_\zeta b
\end{align*}
and $\sigma_z:= z-\sigma$.
With $\Op _h (\kappa_{h}) := [H_h,P]$, we have that
\begin{align*}
    \kappa_{h} + ih \kone - \frac{h^2}{4} \ktwo \in S^{-3} (\aver{g(x,y)}^{-\infty} \aver{\xi, \zeta}^m), 
\end{align*}
where
\begin{align}\label{eq:k12}
\kone := \partial_\xi \sigma \partial_x P
    +\partial_\zeta \sigma \partial_y P,
    \qquad \ktwo := \partial_{\xi \xi} \sigma \partial_{xx} P
    + 2\partial_{\xi \zeta} \sigma \partial_{xy} P
    + \partial_{\zeta \zeta} \sigma \partial_{yy} P.
\end{align}
Since $\nu_h \in S (\aver{f(x,y),\xi,\zeta}^{-\infty})$ and $\kappa_h \in S^{-1} (\aver{g(x,y)}^{-\infty} \aver{\xi,\zeta}^m)$,
the composition calculus implies that
\begin{align*}
    \kappa_h \sharp_h \nu_h - \kappa_h \nu_h + \frac{ih}{2} \{\kappa_h, \nu_h\} \in S^{-3} (\aver{x,y,\xi,\zeta}^{-\infty}),
\end{align*}
with $S^{-3}$ (rather than $S^{-2}$) above because $\kappa_h$ is $O(h)$ in $S(\aver{g(x,y)}^{-\infty} \aver{\xi,\zeta}^m)$.
Therefore,
\begin{align*}
    \sigma_I &= \frac{i}{(2\pi h)^2} \tr \int_{\mathbb{R}^4}\kappa_h \sharp_h \nu_h d R_4 
    = \frac{i}{(2\pi h)^2} \tr \int_{\mathbb{R}^4}\Big(\kappa_h \nu_h - \frac{ih}{2} \{\kappa_h, \nu_h\}\Big)dR_4+o(1)
    \end{align*}
    as $h \rightarrow 0$, with $dR_4:= d x d y d \xi d \zeta$.
    Since
    \begin{align*}
        \kappa_h \nu_h = (ih \kone -\frac{h^2}{4} \ktwo)\frac{1}{\pi}\int_{\mathbb{C}}\bar{\partial} \tilde{\varphi '} (z) \Big(\sigma_z^{-1} +\frac{ih}{2} \{\sigma_z^{-1}, \sigma_z\}\sigma_z^{-1}\Big) d^2z + h^{3}a_h, \qquad a_h \in S (\aver{x,y,\xi,\zeta}^{-\infty})
    \end{align*}
    and
    \begin{align*}
        \{\kappa_h, \nu_h\} = \Big\{ ih \kone, \frac{1}{\pi}\int_{\mathbb{C}}\bar{\partial} \tilde{\varphi '} (z) \sigma_z^{-1} d^2 z\Big\} + h^2 b_h, \qquad b_h \in S (\aver{x,y,\xi,\zeta}^{-\infty}),
    \end{align*}
    it follows that
    \begin{align*}
    \sigma_I
    &=
    \frac{i}{(2\pi h)^2}\frac{1}{\pi} \tr \int_{\mathbb{R}^4}\int_\mathbb{C} \bar{\partial} \tilde{\varphi '} (z) \Big( ih \kone \sigma_z^{-1}
    - \frac{h^2}{2}\kone \{\sigma_z^{-1}, \sigma_z\}\sigma_z^{-1}\\
    &\hspace{3cm}
    - \frac{h^2}{4} \ktwo \sigma_z^{-1} +\frac{h^2}{2} \{ \kone, \sigma_z^{-1} \}\Big) d^2 z d R_4 + o(1)
\end{align*}
as $h \rightarrow 0$. 
Since $\sigma_I$ is independent of $h$, it follows that the $O(h^{-1})$ term above vanishes, and thus
\begin{align}\label{eq:sigmaDiv}
    \sigma_I = \frac{i}{(2\pi)^3}\tr \int_{\mathbb{R}^4} \int_{\mathbb{C}}
    \bar{\partial} \tilde{\varphi '} (z) \Big( -\kone \sigma_z^{-1} \{\sigma_z, \sigma_z^{-1}\}-\frac{1}{2}k_2 \sigma_z^{-1}
    +\{ \kone, \sigma_z^{-1} \}\Big) d^2 z d R_4.
\end{align}
Observe that whenever $\aver{f(x,y), \xi, \zeta}$ is sufficiently large,
$z \mapsto \sigma_z^{-1}$ is holomorphic and thus the above integral over $z$ vanishes (this is verified via an integration by parts in $\bar{\partial}$).
Using that $k_1$ and $k_2$ vanish whenever $\aver{g (x,y)}$ is sufficiently large,
we can replace the above integration limit $\mathbb{R}^4$ by 
a sufficiently large rectangle
$B := B_{xy} \times B_{\xi\zeta}\subset \mathbb{R}^2 \times \mathbb{R}^2$.
We require that $B$ contain all points $(x_0(t),y_0(t),\xi,\zeta)$ for which $\sigma (x_0(t),y_0(t),\xi,\zeta)$ has an eigenvalue of $\alpha$. Moreover, assume that $B_{xy}$ contains $(x_0 (t_0), y_0 (t_0))$ and $(x_0 (-t_0), y_0 (-t_0))$, where we recall the definitions of $t_0$ above Theorem \ref{thm:bic} and $(x_0, y_0)$ below \eqref{eq:fg}.

At this point, $\{ \kone, \sigma_z^{-1} \}$ can be written in divergence form and the corresponding term converted (via integration by parts in $(x,y,\xi,\zeta)$) to an integral over the surface $\partial B$. 
An integration by parts in $\bar{\partial}$ then reveals that the contribution of this term vanishes. Thus we are left with
\begin{align}\label{eq:sigma12}
    \sigma_I = -\frac{i}{(2\pi)^3}\tr \int_{B} \int_{\mathbb{C}}
    \bar{\partial} \tilde{\varphi '} (z) \Big( \kone \sigma_z^{-1} \{\sigma_z, \sigma_z^{-1}\} +\frac{1}{2} \ktwo \sigma_z^{-1} \Big)
    d^2 z d R_4.
\end{align}
We now simplify the first term above.
Integrating by parts in $\bar{\partial}$ with $z =: \lambda+i\omega$, we see that
\begin{align*}
    \sigma_{I,1} := \frac{-i}{(2\pi)^3}\tr \int_{B} \int_{\mathbb{C}}
    \bar{\partial} \tilde{\varphi '} (z) \kone \sigma_z^{-1} \{\sigma_z, \sigma_z^{-1}\} d^2 z dR_4 = \frac{-1}{2 (2\pi)^3} \tr \int_B \int_{\anot}^{\alpha} \varphi ' (\lambda) \kone \sigma_z^{-1} \{\sigma_z, \sigma_z^{-1}\} \Big \vert ^{\omega = 0^+}_{\omega=0^-} d\lambda dR_4.
\end{align*}
Let $\mathcal{L}:= \{x_0 (t), y_0 (t) : t \in \mathbb{R}\}$ be the range of $(x_0, y_0)$. 
For simplicity, assume that $\mathcal{L} \subset g^{-1} (0)$.
By cyclicity of the trace and recalling the definition \eqref{eq:k12} of $k_1$, it follows that 
\begin{align*}
    \tr \kone \{\sigma_z^{-1}, \sigma_z\} \sigma_z^{-1} =
    -\tr (\partial_x P \partial_\xi \sigma_z \{\sigma_z^{-1}, \sigma_z\}_{\zeta,y} + \partial_y P \partial_\zeta \sigma_z\{\sigma_z^{-1}, \sigma_z\}_{\xi,x}) \sigma_z^{-1},
\end{align*}
where we have defined $\{a,b\}_{\alpha, \beta}:=\partial_\alpha a \partial_\beta b-\partial_\beta a \partial_\alpha b$
and used the fact that $(\partial_\xi \sigma, \partial_\zeta \sigma)=-(\partial_\xi \sigma_z, \partial_\zeta \sigma_z)$.
Integrating by parts in $x$ (first term) and $y$ (second term), we obtain
\begin{align*}
    \sigma_{I,1}&=-\frac{1}{2 (2\pi)^3} \tr \int_B \int_{\anot}^{\alpha} \varphi ' (\lambda) P \Big (\partial_x (\partial_\xi \sigma_z \{\sigma_z^{-1}, \sigma_z\}_{\zeta,y}\sigma_z^{-1}) + \partial_y (\partial_\zeta \sigma_z \{\sigma_z^{-1}, \sigma_z\}_{\xi,x} \sigma_z^{-1})\Big) \Big \vert ^{\omega = 0^+}_{\omega=0^-} d\lambda dR_4\\
    &\qquad +\frac{1}{2 (2\pi)^3} \tr \int_{B_{\xi\zeta}}\int_{B_{xy} \cap \mathcal{L}}\int_{\anot}^{\alpha} \varphi ' (\lambda)
    (\partial_\xi \sigma_z \{\sigma_z^{-1}, \sigma_z\}_{\zeta,y}\nu_x + \partial_\zeta \sigma_z\{\sigma_z^{-1}, \sigma_z\}_{\xi,x}\nu_y) \sigma_z^{-1}
    \Big \vert ^{\omega = 0^+}_{\omega=0^-} d\lambda d\ell dR_2\\
    &=: \sigma_{I,10} + \sigma_{I,11},
\end{align*}
where 
$dR_2 := d\xi d\zeta$ and $d\ell$ is the integration measure on $\mathcal{L}$.
Here, $\nu$ is the unit vector (outwardly) normal to the surface $\partial (\{g(x,y) \le 0\} \cap B_{xy})$.
Note that the other surface terms do not contribute; 
over $(\partial (\{g(x,y) \le 0\} \cap B_{xy})) \setminus (B_{xy} \cap \mathcal{L})$, either $P=0$ or the map $z \mapsto \sigma_z^{-1}$ is holomorphic and thus the difference between $\omega=0^+$ and $\omega=0^-$ vanishes.

We now verify that the volume term 
$\sigma_{I,10}$ vanishes.
First, observe that
\begin{align*}
    \int_B \int_{\anot}^{\alpha} \varphi ' (\lambda) P \Big (\partial_\xi (\partial_x \sigma_z \{\sigma_z^{-1}, \sigma_z\}_{\zeta,y}\sigma_z^{-1}) + \partial_\zeta (\partial_y \sigma_z \{\sigma_z^{-1}, \sigma_z\}_{\xi,x} \sigma_z^{-1})\Big) \Big \vert ^{\omega = 0^+}_{\omega=0^-} d\lambda dR_4 = 0,
\end{align*}
as we can use the fact that $P$ is independent of $(\xi,\zeta)$ to integrate the above left-hand side by parts in $\xi$ (first term) and $\zeta$ (second term) to obtain integrals over $\partial B_{\xi \zeta}$, over which $z \mapsto \sigma_z^{-1}$ is holomorphic. Therefore,
\begin{align*}
    \sigma_{I,10} = -\frac{1}{2 (2\pi)^3} \int_B \int_{\anot}^{\alpha} \varphi ' (\lambda) P \tr \Big (\partial_x &(\partial_\xi \sigma_z \{\sigma_z^{-1}, \sigma_z\}_{\zeta,y}\sigma_z^{-1}) + \partial_y (\partial_\zeta \sigma_z \{\sigma_z^{-1}, \sigma_z\}_{\xi,x} \sigma_z^{-1})\\
    &
    -\partial_\xi (\partial_x \sigma_z \{\sigma_z^{-1}, \sigma_z\}_{\zeta,y}\sigma_z^{-1}) - \partial_\zeta (\partial_y \sigma_z \{\sigma_z^{-1}, \sigma_z\}_{\xi,x} \sigma_z^{-1})  
    \Big) \Big \vert ^{\omega = 0^+}_{\omega=0^-} d\lambda dR_4.
\end{align*}
A brute force calculation 
reveals that the above trace vanishes, and thus indeed $\sigma_{I,10} = 0$.

For $\sigma_{I,11}$, we again use cyclicity of the trace to
verify that
\begin{align*}
    \tr (\partial_\xi \sigma_z \{\sigma_z^{-1}, \sigma_z\}_{\zeta,y} \sigma_z^{-1} \nu_x + \partial_\zeta \sigma_z\{\sigma_z^{-1}, \sigma_z\}_{\xi,x} \sigma_z^{-1} \nu_y) =\tr ( (\nu_x \partial_y \sigma_z - \nu_y \partial_x \sigma_z)
    \{\sigma_z^{-1}, \sigma_z\}_{\xi,\zeta} \sigma_z^{-1}).
\end{align*}
We recognize the first factor on the above right-hand side as the
derivative of $\sigma_z$ in the direction of $\mathcal{L}$ (with $t$ increasing).
Recalling the definition of $\tau$ in \eqref{eq:tauz}, we see that 
\begin{align*}
\partial_t \tau (t,\xi,\zeta) &= x_0 ' (t) \partial_x \sigma (x_0 (t), y_0 (t), \xi, \zeta) + y_0 ' (t) \partial_y \sigma (x_0 (t), y_0 (t), \xi, \zeta),\\
\partial_\xi \tau (t,\xi,\zeta) &= \partial_\xi \sigma (x_0 (t), y_0 (t), \xi, \zeta), \qquad \partial_\zeta \tau (t,\xi,\zeta) = \partial_\zeta \sigma (x_0 (t), y_0 (t), \xi, \zeta).
\end{align*}
Defining $N := \sqrt{(x_0'(t))^2 + (y_0'(t))^2}$, we have $(x_0',y_0') = N(-\nu_y,\nu_x)$ and $d\ell = N dt$.
We conclude that
\begin{align}\label{eq:final0}
    \sigma_{I,1} = \frac{1}{16\pi^3} \tr \int_{R} \int_{[\anot,\alpha]} \varphi' (\lambda) \partial_t \tau_z \{\tau_z^{-1}, \tau_z\}_{\xi,\zeta} \tau_z^{-1}\Big \vert^{\omega = 0^+}_{\omega = 0^-}d\lambda dR_3,
\end{align}
where 
$R = (t_1, t_2) \times B_{\xi\zeta}$ with $t_1 = \inf \{t:(x_0(t),y_0(t)) \in B_{xy}\}$ and $t_2 = \sup \{t:(x_0(t),y_0(t)) \in B_{xy}\}$, and
$dR_3 := dt d\xi d\zeta$.

\medskip

We next eliminate $\varphi '$ from \eqref{eq:final0}. Since $\partial_t \tau_z \{\tau_z^{-1}, \tau_z\}_{\xi,\zeta} \tau_z^{-1} \rightarrow 0$ as $|\omega| \rightarrow \infty$, we have
\begin{align*}
    \partial_t \tau_z \{\tau_z^{-1}, \tau_z\}_{\xi,\zeta} \tau_z^{-1} \Big \vert^{\omega=0^+}_{\omega =0^-} = -\int_{- \infty}^{+\infty} \partial_\omega (\partial_t \tau_z \{\tau_z^{-1}, \tau_z\}_{\xi,\zeta} \tau_z^{-1}) d\omega.
\end{align*}
Cyclicity of the trace and the fact that $\partial_\omega \tau_z = i$ imply that
\begin{align}\label{eq:omegaeps}
    \tr \partial_\omega (\partial_t \tau_z \{\tau_z^{-1}, \tau_z\}_{\xi,\zeta} \tau_z^{-1}) = -i \tr \eps_{ijk} \partial_k (\tau_z^{-1} \partial_i \tau_z \tau_z^{-1} \partial_j \tau_z \tau_z^{-1}),
\end{align}
where $\eps_{ijk}$ is the anti-symmetric tensor with $\eps_{123} = 1$, and the variables are identified by $(1,2,3) = (\xi,\zeta,t)$.
Pulling $\partial_k$ out of the integral over $\omega$ and integrating by parts, we get
\begin{align*}
    \tr \int_{R}\partial_t \tau_z \{\tau_z^{-1}, \tau_z\}_{\xi,\zeta} \tau_z^{-1} \Big \vert^{\omega=0^+}_{\omega = 0^-} dR_3 =i \int_{\partial R} \int_{-\infty}^{+\infty} \Theta d\omega d\Sigma, \qquad
    \Theta := \tr \eps_{ijk}\tau_z^{-1} \partial_i \tau_z \tau_z^{-1} \partial_j \tau_z \tau_z^{-1}\nu_k,
\end{align*}
where $\nu$ is the outward unit normal vector to the surface $\partial R$, and $\Sigma$ is the Euclidean surface measure in $\mathbb{R}^3$.
Thus we have shown that
\begin{align*}
    \sigma_{I,1} = \frac{i}{16\pi^3} \int_{[\anot,\alpha]} \varphi ' (\lambda)\int_{\partial R} \int_{-\infty}^{+\infty} \Theta d\omega d\Sigma d\lambda.
\end{align*}
Integrating by parts in $\lambda$, we obtain
\begin{align*}
    \sigma_{I,1} = \frac{i}{16\pi^3} \int_{\partial R} \int_{-\infty}^{+\infty} \Theta d\omega d\Sigma,
\end{align*}
with now $z = \alpha + i\omega$ in the above integrand (and from now on).
The fact that only the boundary term survives follows from
analyticity of $\Theta$ in $z$ over the region of integration (so that $\partial_\lambda \Theta = -i \partial_\omega \Theta$). 

Recall that $R = (-t_0,t_0) \times (-M,M)^2$, 
where $M>0$ can be chosen as large as necessary.
Observe that $|\Theta| \le C \aver{\omega}^{-3}$ and $|\Theta| \le C \aver{\xi,\zeta}^{-m-2}$ uniformly in $(t,\xi,\zeta) \in \partial R$ and $M>0$ sufficiently large, hence $|\Theta| \le C \aver{\omega}^{-3/2} \aver{\xi,\zeta}^{-\frac{m+2}{2}}$
by interpolation. It follows that $\int_{-\infty}^{+\infty} |\Theta|d\omega \le C\aver{\xi,\zeta}^{-\frac{m+2}{2}}$.
Therefore,
sending $M \rightarrow \infty$, we see that
    the contributions to $\sigma_{I,1}$ from the sides of $\partial R$ with normal vector in the $\xi$ and $\zeta$ directions vanish. Indeed, the area of these surfaces is proportional to $M$, with the maximum of the integrand bounded by $CM^{-1-m/2}$.
    Thus we are left with integrals over the sides corresponding to $t = \pm t_0$, over which $\tau (t,\xi,\zeta) = \tau_\pm (\xi, \zeta)$.
    As a consequence, $\sigma_{I,1} = \frac{i}{16\pi^3} (\invariantplus-\invariantminus)$.

\medskip

Recalling \eqref{eq:sigma12}, it remains to show that 
\begin{align} \label{eq:sigma2}
    \sigma_{I,2} :=-\frac{i}{2(2\pi)^3}\tr \int_{\mathbb{R}^4} \ktwo \int_{\mathbb{C}}
    \bar{\partial} \tilde{\varphi '} (z)\sigma_z^{-1}
    d^2 z d R_4 =0.
\end{align}
Define $P_\eps (x,y) := P (\eps (x-\tx, y-\ty))$, where
$(\tx, \ty) \in \mathbb{R}^2$ is chosen such that $P_\eps (x_0 (t), y_0(t)) = 1$ for all $\eps > 0$ and $t \in \mathbb{R}$.
Corollary \ref{cor:invP} implies that $\sigma_I (H,P_\eps)$ is independent of $\eps$, while $\sigma_{I,1}$ is independent of $\eps$ since 
$\invariantpm$ are. It follows that $\sigma_{I,2}$ must also be independent of $\eps>0$.

We will thus replace $P$ by $P_\eps$ and show that $\sigma_{I,2} \rightarrow 0$ as $\eps \downarrow 0$.
As stated in the paragraph below \eqref{eq:sigmaDiv}, the integral over $\mathbb{C}$ in \eqref{eq:sigma2} vanishes whenever $\aver{f(x,y),\xi,\zeta}$ is sufficiently large (uniformly in $\eps$). 
From the definition \eqref{eq:k12}, it follows that
$k_2 = \eps^2 \tilde{k}_{2,\eps}$, where 
$\tilde{k}_{2,\eps}$ vanishes whenever $\aver{\eps g(x-\tx, y-\ty)}$ is sufficiently large.
We conclude that 
there exist positive constants $M_1$ and $M_2$ such that $$|\sigma_{I,2}| \le C \eps^2 {\rm Vol} ( \{\aver{f(x,y)} \le M_1\} \cap \{\aver{\eps g(x-\tx,y-\ty)} \le M_2\}) \le C \eps,$$
with the above volume taken in the $xy$-plane.
We have thus shown that $\sigma_{I,2}=0$, and the proof is complete.
\end{proof}}





\end{document}